\documentclass[twoside,11pt]{article}

\usepackage[abbrvbib]{jmlr2e}
\usepackage{graphicx}
\usepackage{enumerate}
\usepackage{natbib, booktabs, multirow}
\usepackage{url} 
\usepackage{bm}
\usepackage{color}
\usepackage{algpseudocode}
\usepackage{xr}
\usepackage{mathtools}
\usepackage{amsmath,amssymb}
\usepackage{bbm, setspace,bigints, pdflscape}
\usepackage{longtable}
\usepackage{nicematrix}
\usepackage{algorithm}
\usepackage{subcaption, caption}
\usepackage{lscape}
\usepackage{minitoc}
\usepackage{tikz}
\newcommand*\circled[1]{\tikz[baseline=(char.base)]{
            \node[shape=circle,draw,inner sep=0.5pt] (char) {#1};}}
\hypersetup{
    colorlinks,
    linkcolor={blue!50!black},
    citecolor={blue!50!black},
    urlcolor={blue!50!black}
}

\def\Gauss{{ \mathrm{N} }}
\def\T{\mathrm{\scriptscriptstyle{T}}}

\def\ind{\mathbbm{1}}
\def\h{\mathbf{h}}
\def\v{\mathbf{v}}
\newcommand{\norm}[1]{\Vert#1\Vert}
\DeclareMathOperator*{\argmax}{arg\,max}
\DeclareMathOperator*{\argmin}{arg\,min}
\allowdisplaybreaks
\newtheorem{rem}{Remark}[section]

\newtheorem{assumption}[theorem]{Assumption}

\DeclareMathOperator{\tr}{tr}

\newcommand\ci{\perp\!\!\!\perp}

\definecolor{red}{RGB}{0,0,0}

\usepackage{lastpage}
\jmlrheading{27}{2026}{1-\pageref{LastPage}}{4/25; Revised
7/26}{8/26}{25-0686}{Yujie Chen, Anindya Bhadra, Antik Chakraborty}
\ShortHeadings{Likelihood-based inference in graphical models}{Chen, Bhadra and Chakraborty}
\firstpageno{1}

\begin{document}
\title{Likelihood-based Inference in Fully and Partially Observed Exponential Family Graphical Models with Intractable Normalizing Constants}
\author{\name Yujie Chen\email chen1866@purdue.edu \\
       \addr Department of Statistics\\
       Purdue University\\
       West Lafayette, IN 47907, USA
       \AND
       \name Anindya Bhadra \email bhadra@purdue.edu \\
       \addr Department of Statistics\\
       Purdue University\\
       West Lafayette, IN 47907, USA
       \AND
       \name Antik Chakraborty\email antik015@purdue.edu \\
       \addr Department of Statistics\\
       Purdue University\\
       West Lafayette, IN 47907, USA}

\editor{}

\maketitle
\begin{abstract}

 Probabilistic graphical models that encode an underlying Markov random field are fundamental building blocks of \emph{generative modeling} to learn \emph{latent representations} in modern multivariate data sets with complex dependency structures. Among these, the exponential family graphical models are especially popular, given their fairly well-understood statistical properties and computational scalability to high-dimensional data based on pseudo-likelihood methods. These models have been successfully applied in many fields, such as the Ising model in statistical physics and count graphical models in genomics. Another strand of models allows some nodes to be \emph{latent}, so as to allow the marginal distribution of the observable nodes to depart from exponential family to capture more complex dependence. These approaches form the basis of generative models in artificial intelligence, such as the Boltzmann machines and their restricted versions. A fundamental barrier to likelihood-based (i.e., both maximum likelihood and fully Bayesian) inference in both fully and partially observed cases is the intractability of the likelihood. The usual workaround is via adopting pseudo likelihood-based approaches, following the pioneering work of \citet{besag1974spatial}. The goal of this paper is to demonstrate that full likelihood-based analysis of these models is feasible in a computationally efficient manner under a logarithmically sparse setting. The chief innovation lies in 
 utilizing the \emph{tractable independence model} underlying an \emph{intractable graphical model}, to estimate the normalizing constant, as well as its gradient. Extensive numerical results, supporting theory and comparisons with pseudo likelihood-based approaches demonstrate the applicability of the proposed method.
\end{abstract}
\section{Introduction}
Exponential family graphical models provide a coherent framework for describing dependence in multivariate data, and common examples include the Ising model, the Poisson graphical model, and the Gaussian graphical model. However, apart from the Gaussian case, likelihood-based statistical inference is typically hindered by the presence of an intractable normalizing constant {\color{red} and one must resort to Monte Carlo methods for approximate inference. For example, \cite{chen2013evidence} use annealed importance sampling  or AIS \citep{neal2001annealed}, and \cite{ma2013estimating} propose a Monte Carlo estimator of the normalizing constant, specifically for the Ising model with varying temperatures, whereas \cite{lou2017dynamic} propose a search and sampling strategy. While these methods are quite general, they do not leverage the sparsity of the underlying graph, a common assumption in high-dimensional statistical inference, and hence can be very expensive to implement. Within the Bayesian literature, doubly intractable procedures \citep[e.g.,][]{murray2012mcmc,moller2006efficient,liang2010double} are popular but again lack scalability to high-dimensions. As a result, existing methods rely on pseudo-likelihood or approximate Bayesian methods to facilitate estimation \citep{ravikumar2010high,besag1974spatial} or noise contrastive estimation \citep{gutmann2012noise} which treats the normalizing constant as a parameter and devises a strategy to estimate this. Among non-likelihood methods, score matching \citep{Hyvarinen2005} is popular, but it cannot be straightforwardly applied to discrete models.} Although these methods could be consistent under certain assumptions, it is a classical result that full likelihood methods are {\it statistically efficient}, at least in an asymptotic sense and for correctly specified models \citep{fisher1922mathematical,rao1945information}. This apparent bottleneck also applies to the \emph{partially observed exponential family graphical models}, where latent variables are introduced in such a way that the joint distribution of latent and observed variables belongs to an exponential family, although the marginal distribution of the observed variables may not. This leads to the so-called \emph{product of experts} family of models for the observed variables \citep{hinton1998hierarchical,hinton2002training}, which allows complicated dependence structure not belonging in the exponential family to be modeled, and is considered a foundational building block for \emph{generative artificial intelligence}, through an architecture known as the \emph{Hopfield Network} or \emph{Boltzmann Machine} \citep{hopfiled82}, whose purpose is to model Hebbian learning in the brain neural circuitry \citep{hebb1949}. The intractable normalizing constant in this case has precluded a practically successful deployment of a full Boltzmann Machine so far, and has led to \emph{reduced form architectures} where training is possible. For example, in Restricted Boltzmann Machines (RBM) or Deep Boltzmann Machines (DBM) \citep{hinton2007boltzmann,fischer2012introduction,salakhutdinov2009deep}, where observed multivariate data is modeled as the marginal distribution over unobserved hidden variables; a key assumption is the independence of the hidden variables given the visible variables, and vice versa, which leads to training by \emph{approximate likelihood} based methods, such as the contrastive divergence (CD) algorithm of \citet{hinton2002training}. Yet, it is known that with high probability the CD solution differs from the maximum likelihood solution \citep{sutskever2010convergence}, and the justification for the artificial bipartite grouping of latent and observed variables in an RBM, or its deep extensions such as the DBM \citep{salakhutdinov2009deep}, is primarily computational, rather than substantive. 

To address these issues, this article develops an array of computationally feasible methods that enable full likelihood-based analysis (both maximum likelihood and fully Bayesian) for both fully and partially observed exponential family graphical models. This is achieved without resorting to computationally expensive \emph{doubly intractable}  Metropolis-Hastings approaches \citep{murray2012mcmc,moller2006efficient}, which allows our method to scale to hundreds of variables in both fully and partially observed cases. We also propose a technique, which we believe to be the first viable method, for training a \emph{full} Boltzmann machine, without any restrictions on its architecture, such as a bipartite graph. The key to our scalability is sampling under the \emph{independence model}, underlying any exponential family graphical model, which accomplishes two things. First, it allows us to use vanilla Monte Carlo instead of running a prohibitive MCMC scheme after each parameter update as in CD. Second, we show that the variance of the importance sampling estimate is well controlled with our choice. We also show that under logarithmic sparsity assumptions, the necessary and sufficient number of importance samples scales linearly in $p$ to get accurate results. Furthermore, we develop theoretical results for likelihood-based inference in both fully and partially observed cases. These have received little attention so far, due to the computational intractability of likelihood-based inference, in the first place.

The remainder of this article is organized as follows. In Section~\ref{sec:prelim} we provide the necessary background material on fully and partially observed exponential family graphical models, where for the latter we focus in particular on Boltzmann machines and their variants.  The main methodological innovation is discussed in Section~\ref{sec:geyer}, where we outline the use of the \emph{tractable independence model} underlying an \emph{intractable graphical model} in the importance sampling method of \citet{geyer1991} for computing the intractable likelihood and its gradient. Applications to fully and partially observed models are discussed in Section~\ref{sec:PEGM_inference} and ~\ref{sec:partial}, followed by theoretical properties of likelihood-based inference in Section~\ref{sec:theory}. Numerical experiments and data analysis results are in~\ref{sec:sims} and~\ref{sec:real_data}, followed by concluding remarks in Section~\ref{sec:conc}.

\section{Preliminaries}\label{sec:prelim}

\subsection{Pairwise \emph{fully observed} exponential family graphical models} Consider a graph $G = (V, E)$ with $|V| = p$ nodes. Associate each node with a random variable $X_j$ with sample space $\mathcal{X}_j$, for $j = 1, \ldots, p$. The edge set consists of pairs $(j,k)$ such that the nodes $j$ and $k$ are connected. We assume the connection is undirected, i.e. if $(j,k) \in E$ then $(k,j) \in E$. A clique $C \subset V$ of a graph is a fully-connected subset of vertices. Define $\mathcal{C}$ to be the set of all cliques of the graph. A probability distribution on the graph can be defined using a product of clique-wise compatibility functions $\psi_C$, i.e., 
\begin{equation}\label{eq:joint_graph_model}
   p(x_1, \ldots, x_p) = \frac{1}{z} \prod_{C \in \mathcal{C}} \psi_C(x_C),
\end{equation}
where $z$ is such that $p(x_1, \ldots, x_p)$ is a valid density or mass function. Valid distributions such as \eqref{eq:joint_graph_model} under certain conditions can be arrived at from another perspective. Indeed, consider the node-conditional distribution $X_j \mid X_{-j}$, i.e. the conditional distribution of $X_j$ given the other variables. Suppose this distribution belongs to the univariate exponential family of models. Then, by the Hammersley--Clifford theorem \citep{lauritzen1996graphical}, the joint distribution produced by these node-conditional distribution is of the form \eqref{eq:joint_graph_model}. This particular class of models is known as exponential family graphical models \citep{yang2015graphical}, for which the compatibility functions for each clique are the dot product of the respective sufficient statistics and the parameter of the distribution. When the clique set is the set of all nodes and the set of all edges, one obtains what is known as the pairwise exponential family graphical models (PEGM), with pdf given by:
\begin{equation}\label{eq:pairwise_graphical_model}
    p_\theta(x_1, \ldots, x_p) = \frac{1}{z(\theta)} \exp\left\lbrace\sum_{j\in V} \theta_j T_j(x_j) + \sum_{(j,k)\in E} \theta_{jk} T_{jk}(x_j, x_k) + \sum_{j \in V} C(x_j)\right\rbrace,
\end{equation}
where $T_j(x_j)$ and $T_{jk}(x_j, x_k)$ are the sufficient statistics, and $\theta \in \Theta \subset \mathbb{R}^{p \times p}$ is the parameter matrix with $\theta_{jk} = \theta_j$ if $j = k$ and $\theta_{jk} = \theta_{kj}$ if $j \neq k$.
Further generalizations of these models when the cliques are not restricted to be the set of vertices and edges are discussed in \citet{yang2015graphical}. Conditional independence structures in \eqref{eq:pairwise_graphical_model} can also be defined in terms of the pairwise parameters $\theta_{jk}$, since $\theta_{jk} = 0$ if and only if $(j,k) \notin E$. In other words, $X_j \ci X_k \mid X_{-(j,k)}$ if and only if $\theta_{jk} = 0$. It is straightforward to observe that if $X$ has density given by \eqref{eq:pairwise_graphical_model}, then the node-conditional distribution of $X_j \mid X_{-j}$ is:
\begin{align*}
    p_{\theta}(x_j \mid x_{-j}) = \frac{1}{z(\theta;x_{-j})} \exp \left\lbrace \theta_jT_j(x_j) + C(x_j) + 2 \underset{k \in N(j)}{\sum}\theta_{jk}T_{jk}(x_j, x_k) \right\rbrace,
\end{align*}
where 
$z(\theta;x_{-j})$ is the normalizing constant belonging to a univariate exponential family distribution for $(X_j\mid X_{-j})$, and $N(j)$ denotes all neighbors of $j$, as encoded by $G$. The parameter space $\Theta$ for which \eqref{eq:pairwise_graphical_model} is a valid probability model is convex \citep{barndorff2014information, wainwright2008graphical} and the sample space is $\mathcal{X} = \mathcal{X}_1 \times \ldots \times \mathcal{X}_p$. When $\Theta$ is open, the model is said to be regular.
Some specific PEGM examples are as follows.
 \begin{example} (Ising). The Ising model \citep{ising1924beitrag} is a PEGM with $C(x_j) = 0$, $T_j(x_j) = x_j$ and $T_{jk}(x_j, x_k) = x_jx_k$. Moreover, $X_j \mid X_{-j}$ has a Bernoulli distribution with natural parameter $\theta_j + 2\sum_{k \in N(j)} \theta_{jk}x_k$. The sample space for this model is $\{0,1\}^p$ and $\Theta = \mathbb{R}^{p \times p}$.
\end{example}
\begin{example} (PGM). The Poisson graphical model (PGM), proposed by \citet{besag1974spatial}, and further developed by \citet{yang2013poisson} is a PEGM with $C(x_j) = \log x_j!$, $T_j(x_j) = x_j$ and $T_{jk}(x_j, x_k) = x_j x_k$. Here $X_j \mid X_{-j} \sim \text{Poi}(\exp\{\theta_j + 2\sum_{k\in N(j)} \theta_{jk} x_k\})$. The sample space is $\{\mathbb{N} \cup \{0\}\}^p$, i.e. the set of all $p$-dimensional count vectors, and $\Theta$ is the set of all $p \times p$ real matrices with non-positive off-diagonal elements.
\end{example}
Other examples include the Potts model \citep{potts1952some}, or the truncated Poisson graphical model \citep{yang2013poisson}. 
In all of the examples above, likelihood-based inference is generally infeasible due to the intractability of $z(\theta)$ and available methods approximate the likelihood $\ell(\theta)$ by $\prod_{j}\ell(X_j \mid X_{-j})$, the \emph{pseudo-likelihood}. There are two key reasons for this: (1) $X_j \mid X_{-j}$ is a member of a univariate exponential family and (2) due to the product form of the pseudo-likelihood, inference can be carried out in parallel for the $p$ nodes.

\subsection{Pairwise \emph{partially observed} exponential family graphical models: product of experts and \emph{approximate} maximum likelihood via contrastive divergence}\label{sec:BM_background}
 Consider a binary vector $\v \in \{0,1\}^p$ which might encode the on and off pixels of a black-and-white image. In generative artificial intelligence (AI), the goal is to learn the underlying features of the distribution of $\v$ so that high-quality synthetic data from the trained AI can be produced for new content. One potential modeling choice could be an exponential family model, e.g., the Ising. For modern complex data such as images, this restriction to exponential family could severely limit  the model's ability to capture more complex dependence not belonging to the exponential family. As an alternative, Boltzmann Machines have emerged as one of the fundamental tools for generative AI. These can be thought of as stochastic neural networks represented by a graph that encodes symmetric connections between observed variables and hidden or latent variables \citep{ackley1985learning}. The joint distribution of the $p$ observed and $m$ hidden variables are assumed to belong to a PEGM (Ising in this case): $p_{\theta}(\v, \h) = z(\theta)^{-1} \exp({-E_{\theta}(\v, \h)})$ where $E_{\theta}(\v, \h)$ is known as the energy function in the literature and $z(\theta) = \sum_{\v, \h} \exp({-E_{\theta}(\v, \h)})$. Specifically,
 \begin{align*}
    \log p_{\theta}(\v, \h)  = \sum_{j} \theta_{jj}v_j + \sum_{k} \theta_{kk} h_k + \sum_{j \neq j'} \theta_{jj'} v_j v_{j'} + \sum_{k \neq k'} \theta_{kk'} h_k h_{k'} + \sum_{j,k} \theta_{jk} v_j h_k - \log z(\theta), 
\end{align*}
for $j,j'=1,\ldots,p;\; k,k'=p+1,\ldots, p+m$ and $\theta \in \mathbb{R}^{(p+m)\times (p+m)}$. The marginal distribution of the visible variables, $p_{\theta}(\v) = \sum_{\h}p_{\theta}(\v,\h)$, does not necessarily belong to the exponential family. Introduction of the hidden variables allows one to describe complex dependency structures in the visible variables using simple conditional distributions. Moreover, the hidden variables act as non-linear feature detectors \citep{hinton2007boltzmann}. It can be seen that for this model, the gradient of the log-likelihood for a visible sample $\v$  is: 
\begin{align}\label{eq:BM_grad}
    \dfrac{\partial \ell (\theta )}{\partial \theta} =  \dfrac{\partial \log \sum_{\h} p_\theta(\v, \h)}{\partial \theta} = \sum_{v, h} \dfrac{\partial E_{\theta}(\v, \h)}{\partial \theta} p_{\theta}(\v, \h) -  \sum_{h} \dfrac{\partial E_{\theta}(\v, \h)}{\partial \theta} p_{\theta}(\h \mid \v),
\end{align}
where \eqref{eq:BM_grad} can be interpreted as the difference of two expectations: expected value of the gradient of the energy function under the joint distribution of $(\v,\h)$ and under the conditional distribution of $\h \mid \v$. For most general versions of this model, existing learning algorithms approximate either both or at least one of these two terms using MCMC sampling, which is too slow for practical deployment, and lacks scalability to larger dimensions. 

Equation \eqref{eq:BM_grad} also illustrates that the most general version of the Boltzmann machine is hard to train and computationally demanding, which has paved the way for simplified versions of Boltzmann machines, specifically the restricted Boltzmann machine (RBM) where the underlying graph is bipartite. RBM \citep{smolensky1986information} is inspired by the biological connectivity network among neurons. RBM's multi-layer extension, or deep Boltzmann machine (DBM) \citep{salakhutdinov2007restricted} has proven to be a popular machine learning tool for unsupervised learning \citep{zhang2019deep} and is a \emph{universal approximator} for any distribution over $\{0,1\}^p$ \citep{le2008representational, montufar2011refinements}. Adjacent layers of a DBM have the same architecture as an RBM. Figure~\ref{fig:BM} shows the graphical model for a BM, RBM and a DBM with three hidden layers.

\begin{figure}

\begin{minipage}[b]{0.5\textwidth}
\centering
\begin{tikzpicture}[remember picture,overlay]

\tikzset{cir/.style={circle,draw=black!80,thick,minimum size=0.6cm},y=0.6cm,font=\sffamily}

\begin{scope}[rotate = 90]
\node[cir] (b1) at (0.5,2) {$v_1$};
\node[cir] (b2) at (0.5,0) {$v_2$};
\node[cir] (b3) at (0.5,-2) {$v_3$};
\node[cir, fill=black!15] (c1) at (2.5,2) {$h_1$};
\node[cir,fill=black!15] (c2) at (2.5,0) {$h_2$};
\node[cir,fill=black!15] (c3) at (2.5,-2) {$h_3$};


\foreach \cnto in {1,2, 3} {
    \foreach \cntt in {1,2, 3} {
        \draw [-] (b\cnto.north)--(c\cntt.south);
    }
}

\draw [-] (b1)--(b2);
\draw [bend right] (b1) to (b3);
\draw [-] (b2)--(b3);
\draw [-] (c1)--(c2);
\draw [bend left] (c1) to (c3);
\draw [-] (c2)--(c3);
\end{scope}
\end{tikzpicture}
\caption*{Fig. 1(a): BM}
\end{minipage}
\hfill
\begin{minipage}[b]{0.5\textwidth}
\centering
\begin{tikzpicture}

\tikzset{cir/.style={circle,draw=black!80,thick,minimum size=0.6cm},y=0.6cm,font=\sffamily}

\begin{scope}[rotate = 90]
\node[cir] (b1) at (0,2) {$v_1$};
\node[cir] (b2) at (0,0) {$v_2$};
\node[cir] (b3) at (0,-2) {$v_3$};
\node[cir, fill=black!15] (c1) at (2,2) {$h_1$};
\node[cir,fill=black!15] (c2) at (2,0) {$h_2$};
\node[cir,fill=black!15] (c3) at (2,-2) {$h_3$};


\foreach \cnto in {1,2, 3} {
    \foreach \cntt in {1,2, 3} {
        \draw [-] (b\cnto.north)--(c\cntt.south);
    }
}

\end{scope}
\end{tikzpicture}
\caption*{Fig. 1(b): RBM}
\end{minipage}

\begin{minipage}[b]{0.5\textwidth}
\centering
\begin{tikzpicture}

\tikzset{cir/.style={circle,draw=black!80,thick,minimum size=0.6cm},y=0.6cm,font=\sffamily}

\begin{scope}
\node[cir] (b1) at (0,0) {$v_1$};
\node[cir] (b2) at (0,-1*2) {$v_2$};
\node[cir] (b3) at (0,-2*2) {$v_3$};
\node[cir, fill=black!15] (c1) at (2,0) {\small $h_{1}^{(1)}$};
\node[cir,fill=black!15] (c2) at (2,-1*2.1) {\small $h_{2}^{(1)}$};
\node[cir,fill=black!15] (c3) at (2,-2*2.1) {\small $h_{3}^{(1)}$};
\node[cir, fill=black!15] (c4) at (4,0) {\small $h_{1}^{(2)}$};
\node[cir,fill=black!15] (c5) at (4,-1*2.1) {\small $h_{2}^{(2)}$};
\node[cir,fill=black!15] (c6) at (4,-2*2.1) {\small $h_{3}^{(2)}$};
\node[cir, fill=black!15] (c7) at (6,0) {\small $h_{1}^{(3)}$};
\node[cir,fill=black!15] (c8) at (6,-1*2.1) {\small $h_{2}^{(3)}$};
\node[cir,fill=black!15] (c9) at (6,-2*2.1) {\small $h_{3}^{(3)}$};


\foreach \cnto in {1,2, 3} {
    \foreach \cntt in {1,2, 3} {
        \draw [-] (b\cnto.east)--(c\cntt.west);
    }
}
\foreach \cnto in {1,2, 3} {
    \foreach \cntt in {4,5,6} {
        \draw [-] (c\cnto.east)--(c\cntt.west);
    }
}
\foreach \cnto in {4,5,6} {
    \foreach \cntt in {7,8,9} {
        \draw [-] (c\cnto.east)--(c\cntt.west);
    }
}

\end{scope}
\end{tikzpicture}
\caption*{Fig. 1(c): DBM}
\end{minipage}
\begin{minipage}[b]{0.5\textwidth}
    \centering
    \begin{tikzpicture}

\tikzset{cir/.style={circle,draw=black!80,thick,minimum size=0.6cm},y=0.6cm,font=\sffamily}

\begin{scope}[rotate = 90]
\node[cir] (b1) at (0,0) {$v_1$};
\node[cir] (b2) at (0,-1*2) {$v_2$};
\node[cir] (b3) at (0,-2*2) {$v_3$};
\node[cir, fill=black!15] (c4) at (0,2) {\small $h_{1}^{(2)}$};
\node[cir,fill=black!15] (c5) at (0, 4) {\small $h_{2}^{(2)}$};
\node[cir,fill=black!15] (c6) at (0, 6) {\small $h_{3}^{(2)}$};
\node[cir, fill=black!15] (c1) at (2,0) {\small $h_{1}^{(1)}$};
\node[cir,fill=black!15] (c2) at (2,-1*2) {\small $h_{2}^{(1)}$};
\node[cir,fill=black!15] (c3) at (2,-2*2) {\small $h_{3}^{(1)}$};

\node[cir, fill=black!15] (c7) at (2,2) {\small $h_{1}^{(3)}$};
\node[cir,fill=black!15] (c8) at (2,4) {\small $h_{2}^{(3)}$};
\node[cir,fill=black!15] (c9) at (2,6) {\small $h_{3}^{(3)}$};


\foreach \cnto in {1,2, 3} {
    \foreach \cntt in {1,2, 3} {
        \draw [-] (b\cnto.north)--(c\cntt.south);
    }
}
\foreach \cnto in {4,5,6} {
    \foreach \cntt in {7,8,9} {
        \draw [-] (c\cnto.north)--(c\cntt.south);
    }
}

\foreach \cnto in {1,2,3} {
    \foreach \cntt in {4,5,6} {
        \draw [-] (c\cnto.south)--(c\cntt.north);
    }
}
\end{scope}
\end{tikzpicture}
\caption*{Fig. 1(d): DBM rearranged as an RBM}
\end{minipage}

\caption{From left to right: BM, RBM, DBM with three visible nodes. DBM has three layers of hidden variables. The notation is: hidden nodes ($h_j \in \{0,1\}$, shaded in gray), visible nodes ($v_k \in \{0,1\}$, transparent). Deep hidden nodes in layer $l$ are denoted by $h^{(l)}$. }
\label{fig:BM}
\end{figure}

Learning in the RBM is facilitated by the conditional independence of the hidden variables given the visible variables and vice versa. Specifically, for RBM: $- E_{\theta}(\v, \h) = \sum_{k} \theta_{kk} h_k  + \sum_{j} \theta_{jj} v_j + \sum_{j}\sum_{k} \theta_{jk}v_j h_k $, which implies $\mathbb{P}(h_k=1\mid \v) = \sigma(\theta_{kk} + \sum_{j} \theta_{jk} v_j)$ and $\mathbb{P}(v_j=1\mid \h) = \sigma(\theta_{jj} + \sum_{k} \theta_{jk} h_k)$ where $\sigma(x) = \{1+ \exp(-x)\}^{-1},\; x\in \mathbb{R}$ is the logistic sigmoid function. The bipartite dependence structure gives a \emph{product of experts} marginal model for the visible variables \citep{fischer2012introduction}:
\begin{equation}\label{eq:RBM_marginal}
    p_{\theta}(\v) = z(\theta)^{-1} \prod_{j=1}^p \exp({\theta_{jj}v_j})\prod_{k=p+1}^{p+m} \left(1 + \exp({{\theta_{kk} + \sum_{j=1}^p \theta_{jk}v_j}})\right).
\end{equation}
The contrastive divergence (CD) algorithm \citep{hinton2002training} is an approximate procedure for finding the MLE under this model via gradient ascent. Suppose $\theta = \theta^{(t)}$ at iteration $t$. Using \eqref{eq:BM_grad} for an RBM one obtains at $\theta=\theta^{(t)}$:
\begin{align*}
 \dfrac{\partial \log  p_{\theta}(\v)}{\partial \theta_{jk}}= \langle \v_j \h_k \rangle_{\mathrm{data}} - \langle \v_j \h_k \rangle_{\mathrm{model}}, \quad j\neq k,
 \end{align*}
\noindent where the angle brackets with subscript ``data'' denote an expectation with respect to $p_{\theta^{(t)}} (\h \mid \v)$ at the observed $\v$, which is analytic; and those with subscript ``model'' denote an {expectation} with respect to $p_{\theta^{(t)}} (\h, \v)$; which is typically not available in closed form and is evaluated using MCMC.  Similarly, for the bias terms, we have at $\theta=\theta^{(t)}$ that: ${\partial \log  p_{\theta}(\v)}/{\partial \theta_{jj}} = \v_j - \sum_{\v} \v_j p_{\theta^{(t)}}(\v) ,$ and $ {\partial \log  p_{\theta}(\v)}/{\partial \theta_{kk}} = \mathbb{P}(\h_k = 1 \mid \v) - \sum_{\v} p_{\theta^{(t)}}(\v)\mathbb{P}(\h_k = 1 \mid \v)$.
However, the difficulty here is that direct and exact simulation from $p_{\theta^{(t)}} (\h, \v)$ is infeasible. Thus, the CD algorithm uses a Gibbs sampler to iteratively sample $(\h \mid \v , \theta^{(t)})$ and $(\v \mid \h, \theta^{(t)})$. The key issues are:
\begin{enumerate}
\item The procedure requires running a new Gibbs sampler every time $\theta^{(t)}$ is updated, until convergence to the target distribution $p_{\theta^{(t)}} (\h, \v)$ is achieved. This is reminiscent of the double MH procedure \citep{liang2010double}, which is computationally prohibitive and is one of the primary motivations behind the bipartite graph of RBM that allows \emph{batch sampling} of $\h \mid \v , \theta^{(t)}$ and $\v \mid \h, \theta^{(t)}$ by drawing independent Bernoulli vectors. The conditional independence structure is lost if one allows within layer connections, precluding the use of more general Boltzmann machines. 

\item Even for an RBM with batch sampling, convergence to the stationary target $p_{\theta^{(t)}} (\h,\v)$ is only achieved if the chain is run for $K\to\infty$ MCMC iterations at every iteration $t$, under standard MCMC theory. In practice, the CD algorithm is only run for a small $K$ (CD-$K$) or even $K=1$ Gibbs iterations for every $t$. 
\end{enumerate}

\vspace{-20pt}
\section{A Monte Carlo Estimate of $z(\theta)$ and Its Gradient for Exponential Family Graphical Models}\label{sec:geyer}
Let $p_{\theta}(x) = {q_{\theta}(x)}/{z(\theta)}, \, x \in \mathbb{R}^p$ be a probability density function such that $z(\theta) = \int q_{\theta}(x) dx = \exp(A(\theta))$ and $\theta \in \Theta$.  Consider the situation where the explicit form of $q_{\theta}(x)$ is known for any $\theta$ in the parameter space and is computationally cheap to evaluate but the analytical form of $z(\theta)$ is intractable for every $\theta \in \Theta$.
Given $n$ independent observations from such a model, form the matrix $\mathbf{X}\in\mathbb{R}^{n\times p}$. For notational convenience, we later use $\mathbf{X}=(X_1,\ldots, X_p)$, where each $X_j$ is understood to be a vector of the $j$th variable over $n$ samples for $j=\{1,\ldots, p\}$. Otherwise, when the indexing is over the $i$th row of $\mathbf{X}$, it is made explicit by writing $X_{i\bullet}$ for $i=\{1,\ldots,n\}$. {\color{red} Let the log-likelihood at $\theta \in \Theta$ be $\ell(\theta)$. Then,
\begin{align*}
    \ell(\theta) = \log p_{\theta}(\mathbf{X}) = \sum_{i=1}^{n}\log q_{\theta}(X_{i\bullet}) - n\log z(\theta), 
\end{align*}
where $z(\theta)$ cannot be explicitly computed and hence, maximum likelihood estimation cannot be performed. We note here that any first-order optimization routine to maximize the above likelihood function requires the gradient $\nabla_\theta \log z(\theta)$. Additionally, for a Bayesian treatment of the problem, one also needs $z(\theta)$ for implementing MCMC  schemes.
Existing literature \citep{chen2013evidence, lou2017dynamic} addresses this problem by considering Monte Carlo, typically by importance sampling estimates of $z(\theta)$ and $\nabla_\theta \log z(\theta)$. Assume $\hat{z}(\theta)$ and $\hat{\nabla}_\theta \log z(\theta)$ are some Monte Carlo estimate of the $z(\theta)$ and $\nabla_\theta \log z(\theta)$, respectively. Then the estimate of the log-likelihood and its gradient are given by: 
$$\hat{\ell}(\theta) = \sum_{i=1}^n \log q_\theta(X_{i\bullet}) - n \log \hat{z}(\theta), \quad \hat{\nabla}_\theta \ell(\theta) = \sum_{i=1}^n \nabla_\theta \log q_\theta(X_{i\bullet}) - n\hat{\nabla}_\theta\log z(\theta). $$
The key insight in this regard was provided by \cite{geyer1991} who noted:}
\begin{align}
    \dfrac{z(\theta)}{z(\phi)} & = \frac{1}{z(\phi)} \int q_{\theta} (x) dx = \dfrac{1}{z(\phi)} \int \dfrac{q_{\theta}(x)}{q_{\phi}(x)} q_{\phi}(x)dx
     = \int \dfrac{q_{\theta}(x)}{q_{\phi}(x)} p_{\phi}(x)dx = \mathbb{E}_{Y \sim p_{\phi}}\left[\dfrac{q_{\theta}(Y)}{q_{\phi}(Y)}\right],\label{eq:geyer}
\end{align}
motivating the Monte Carlo estimate for $z(\theta)/z(\phi)$ as $T_N^{(z)} = \frac{1}{N}\sum_{i = 1}^N \frac{q_{\theta}(Y_{i\bullet})}{q_{\phi}(Y_{i\bullet})}$,
where $Y_{i\bullet}\stackrel{i.i.d.}\sim p_\phi$. 
For any generic $\phi \in \Theta$, the node-conditional distributions can be used to devise a Markov chain Monte Carlo (MCMC) scheme to sample from $p_\phi$. However, if specific choices of $\phi$ are available such that independent samples are possible to simulate cheaply, then these samples could also be utilized to approximate the expectation. Moreover, if for such a $\phi$, the corresponding normalizing constant $z(\phi)$ is available in closed form, one has an estimate of $z(\theta)$, for all $\theta$. The key to the current work is our choice:  $\boxed{\phi = \text{diag}(\theta)}$,  i.e., $\phi$ is a matrix of the same dimension as $\theta$, with diagonal elements equal to the diagonal elements of $\theta$, and off-diagonal elements set to 0. 
With this choice: (a) a sample $Y \sim p_{\phi}$ can be obtained by sampling $Y_j \sim p_{\theta_{jj}}$ independently and setting $Y = (Y_1, \ldots, Y_p)$, and (b) $z(\phi)$ is analytically tractable.  This is in general true for any PEGM since the independence model is simply a distribution over $p$-univariate exponential families, and a few examples are provided next. For the Ising model, $z(\phi) = \prod_{j=1}^p z(\theta_{jj}) = \prod_{j=1}^p (1 + \exp({\theta_{jj}}))$.
For the Poisson graphical model, $z(\phi) = \prod_{j=1}^p z(\theta_{jj}) = \prod_{j=1}^p \exp(\exp({\theta_{jj}}))$.
For observed samples $\mathbf{X}$, the log-likelihood is $\ell(\theta) = \sum_{i=1}^n \log q_\theta(X_{i\bullet}) - n \log z(\theta)$. To obtain the maximum likelihood estimator of $\theta$, one can use the gradient ascent method leading to updates of the form $\theta^{(t+1)} = \theta^{(t)} + \gamma \nabla_{\theta} \ell(\theta) = \theta^{(t)} + \gamma [ \sum_{i=1}^n \nabla_{\theta} \log q_{\theta}(X_{i\bullet}) - n \nabla_{\theta}\log z(\theta)]$, where $\gamma$ is a suitable step-size. Clearly, for such an algorithm to succeed, one needs access to the gradient $\nabla_{\theta}\log z(\theta) = \nabla_{\theta} z(\theta)/z(\theta)$. In the next proposition, an unbiased estimate for $\nabla_{\theta}z(\theta)$ is proposed, similar in spirit to the estimate of $z(\theta)$.
\begin{proposition}\label{prop:geyer_grad}
Define $T_N^{(\nabla_\theta z)} = \frac{1}{N} \sum_{i=1}^N \frac{\nabla_\theta q_\theta(Y_{i\bullet})}{q_\phi(Y_{i\bullet})}$ where $Y_{i\bullet} \overset{iid}{\sim} p_\phi$ for $i = 1, \ldots, N$. Then: $$\dfrac{\nabla_{\theta} z(\theta)}{z(\phi)} = \mathbb{E}_{ p_{\phi}} \left[ T_N^{(\nabla_\theta z)}\right].$$
\end{proposition}
\begin{proof}
In~\eqref{eq:geyer} replace $z(\theta)$ with $\nabla_\theta z(\theta)$ and differentiate under the integral sign.
\end{proof}
Note that $T_N^{(\nabla_\theta z)}$ has the same dimension as $\theta$. While the Monte Carlo version of estimates of $z(\theta)$ and $\nabla_{\theta} z(\theta)$ are unbiased by construction, their practical utility is determined by their variance, which need not be bounded for an arbitrary $\phi$ \citep{geyer1992constrained}. However, for our specific choice of $\phi$, the following key proposition establishes that the variances of  Monte Carlo estimates of $z(\theta)$ and $\nabla_\theta z(\theta)$ are finite under mild conditions.   
    \begin{proposition}\label{prop:exp}
Let $u =\theta - \mathrm{diag}(\theta)$, and $\phi=\mathrm{diag}(\theta)$. Then, (a) Monte Carlo estimate for $z(\theta)$ has bounded variance if $(2u+\phi)\in\Theta$ and, moreover, (b) the Monte Carlo estimate of $\nabla_\theta z(\theta)$ also has bounded variance if $(2(1+\delta) u+\phi)\in\Theta$ for some $\delta>0$.
\end{proposition}
Refer to Supplementary Section~\ref{prf:exp} for a proof of Proposition~\ref{prop:exp}.  The importance of this proposition lies not in deriving an exact expression for the variance. Rather, Proposition~\ref{prop:exp} establishes the conditions for bounded population variance, so that the central limit theorem applies to our importance sampling estimates, and ensures a $\sqrt{N}$ rate of convergence, where $N$ is the number of i.i.d. samples from $p_\phi(\cdot)$, which can be drawn cheaply in batch.  
We proceed to make several remarks on the implications of this proposition. 

\begin{rem}
Note that, $\frac{1}{1+\delta}(2(1+\delta)u+\phi) + \frac{\delta}{1+\delta} \phi = 2u+\phi.$
 Thus, by the convexity of exponential family parameter space \citep[see, e.g., Theorem 9.1 of ][]{barndorff2014information}, the result of Proposition~\ref{prop:exp}, Part (b) implies that of Part (a).
\end{rem}

\begin{rem}
Note that if $\phi=\phi_0$ is an arbitrary reference point not depending on $\theta$, the Monte Carlo variance of the importance ratio $q_\theta(X)/q_{\phi_0}(X)$ can easily become unbounded, even if $\phi_0\in\Theta$. This is observed, for example, in a Poisson graphical model, by choosing a $\phi_0$, such that $\mathrm{offdiag}(\theta-\phi_0)>0$. This shows the importance of setting $\phi=\mathrm{diag}(\theta)$, and the technique of \citet{geyer1991} is not very useful under an arbitrary $\phi$, possibly explaining why the technique has not been widely attempted in this model. Setting $\phi=\mathrm{diag}(\theta)$ removes the first order $T_j(X_j)$ terms from the ratio $q_\theta(X)/q_{\phi}(X)$, which is helpful for bounding the variance of what remains.
\end{rem}
\begin{rem}\label{rem:diag}
Another benefit of setting $\phi=\mathrm{diag}(\theta)$, is that one is (exactly) sampling from the independence model $p_\phi(\cdot)$ for the purpose of Monte Carlo, which allows batch sampling, and is computationally far more efficient than (approximately) sampling from $p_\theta(\cdot)$ at every iteration, for which an iterative technique such as MCMC must be deployed. Variations of this latter approach, i.e., sampling auxiliary data $Y_1,\ldots, Y_N$ under $p_\theta(\cdot)$ given the current $\theta$ to approximate $z(\theta)$ and $\nabla_\theta \log z(\theta)$, are used in the double Metropolis--Hastings of \citet{liang2010double}, the contrastive divergence approach of \citet{hinton2002training} that is popular for fitting a restricted Boltzmann machine \citep{salakhutdinov2007restricted}, and the Hamiltonian Monte Carlo approach under intractable likelihood as in \citet{stoehr2019noisy}. All of these are based on the identity 
$\nabla_\theta \log z(\theta) = \mathbb{E}_{Y\sim p_\theta} \{\nabla_\theta (\log q_\theta(Y))\}$, for which the RHS can be estimated by $T_N^{(\nabla_\theta \log z)}=N^{-1}\sum_{i=1}^{N} \nabla_\theta (\log q_\theta(Y_{i\bullet}))$,
where $Y_{1\bullet},\ldots,Y_{N\bullet}\stackrel{i.i.d.}\sim p_\theta$. The need to sample iteratively from $p_\theta(\cdot)$ lies at the heart of the lack of scalability for MCMC based methods, in both fully Bayesian as well as approximate maximum likelihood calculations for doubly intractable models. The choice of $\phi=\mathrm{diag}(\theta)$, and subsequent sampling from $p_\phi$, eliminates this difficulty, resulting in far greater scalability. From a computational point of view, the choice $\phi = c\mathrm{diag}(\theta),$ for some $c>0$ to be determined, is also appealing. For this choice, sampling is equally cheap, and $z(\phi)$ is also available in closed form. One may be tempted to optimize over $c$ to minimize the variance of the importance estimate, although, due to the intractability of the variance of $z(\theta)/z(\phi)$, this may not be straightforward. We choose to avoid this complication and find $c=1$ works well in practice.
\end{rem}
\begin{rem}
    Generally, the variance of the estimator is a function of $\theta$ and $p$. But for PGMs, the variance of the proposed estimator admits a universal bound. Indeed, 
    $$\mathbb{E}_{\phi}[q_\theta^2(X)/q_\phi^2(X)]  = \sum_{\mathcal{X}} \exp\left\lbrace 2\sum_j \sum_k \theta_{jk} X_jX_k \right\rbrace p_\phi(X) \leq \sum_{\mathcal{X}} p_\phi(X) = 1$$ 
    since for PGMs, $\theta_{jk} \leq 0$ for $j \neq k$, and $X_j, X_k \geq 0$.
\end{rem}
\begin{rem}
Although not our main focus, the same construction may potentially be used to design a simple MCMC-free sampler for  $p_\theta(\cdot)$. See Supplementary Section~\ref{app:sampler} for details.
\end{rem}

The following result provides further support on the well-behaved nature of our Monte Carlo estimates, providing exponential concentration bounds.
\begin{proposition}\label{lemma:sub-gaussian}
   Consider the class of PEGMs in \eqref{eq:pairwise_graphical_model} with $T_j(x_j) = x_j$ and $T_{jk}(x_j, x_k) = x_j x_k$ and suppose the $j$th sample space is bounded in the interval $[a_j, b_j]$ for $j=1,\ldots,p$. Set $U = q_{\theta}(Y)/q_{\phi}(Y)$ and $V = \nabla_{\theta_{jk}}q_{\theta}(Y)/q_{\phi}(Y)$ with $ Y \sim p_{\phi}$. Then for every $\theta$, both $U$ and $V$ as a function of $Y$ have bounded differences in the sense of \citet[Corollary 1]{boucheron2003concentration}. Moreover, there exists $C>0$ such that for every $t>0$, $\mathbb{P}[|U - \mathbb{E}(U)| > t] \leq 2 \exp\{ - t^2/4C\}$, and $\mathbb{P}[|V - \mathbb{E}(V)| > t] \leq 2 \exp\{ - t^2/4C\}$.
\end{proposition}
The proof is given in Supplementary Section~\ref{sec:lemma_proof}. Proposition \ref{lemma:sub-gaussian} establishes that for PEGMs with finite support, the specific choice of $\phi = \text{diag}(\theta)$ not only leads to finite variance Monte Carlo estimates of $z(\theta)$ and $\nabla_{\theta}z(\theta)$, but for this specific class, the estimates are sub-Gaussian, which, in turn, implies exponential concentration of the sample mean of $U, V$. The class of PEGMs with finite support is quite rich. Indeed, the Ising model, Potts model, and the truncated Poisson graphical model are all members of this family.

The number of samples to draw in an importance sampling procedure typically depends on the problem and the dimension. In the next result, we make the connection between $N$ and $p$ explicit. The result is based on the work by \cite{chatterjee2018sample} who characterized $N$ in terms of the Kullback-Leibler divergence between $p_\theta$ and $p_\phi$, which we write as $D(p_\theta, p_\phi)$. For PEGMs, $D(p_\theta, p_\phi) = \log z(\theta) - \log z(\phi) + (\text{vech}(\theta) - \text{vech}(\phi))^\T \mu$ where $\mu = \nabla_\theta \log z(\theta)$ and $\text{vech}(A)$ denotes the half-vectorization of a symmetric matrix $A$. 
\begin{proposition}\label{prop:importance_sample_size}
    Suppose $S_\theta = \{\theta: \mathrm{vech}(\theta) \neq 0\}$, where half-vectorization is on the off-diagonal terms, and let $|S_\theta| = \log p $. Set $\phi = \mathrm{diag}(\theta)$. If $N = Cp$ for some large positive constant $C$, then,
    $$\mathbb{E}\left| \frac{z(\phi) T_N^{(z)}}{z(\theta)} - 1 \right| \leq e^{-C \sigma/4} + 4/C,$$
    where $\sigma = \{(\mathrm{vech}(\theta - \phi))^\T \nabla^2 \log z(\theta)(\mathrm{vech}(\theta - \phi))\}^{1/2}$.
\end{proposition}
See Supplementary Section~\ref{sec:importance_sample_size_proof} for a proof. The above sample size under these settings is also necessary to ensure good results \citep[][Section 3]{chatterjee2018sample}.  As a result, for sparse graphical models, the proposed importance sampling estimate is extremely beneficial in that the computational complexity of $T_N^{(z)}$ for $N = O(p)$ is linear in $p$. We also note here that the sparsity assumption arises naturally in high-dimensional graphical models, and existing methods often assume logarithmic sparsity \citep{ravikumar2010high}. Coincidentally, the same degree of sparsity allows for a practically useful importance sampler. In Supplementary Section \ref{sup:mc_imp}, we consider a Gaussian graphical model as a canonical example of a PEGM with analytically tractable $z(\theta)$. Our results indicate that the proposed estimator works really well up to $p = 300$ when the underlying $\theta$ matrix is sparse.
\begin{rem}
   One should exercise caution in interpreting the result of Proposition~\ref{prop:importance_sample_size}. It does not say the partition function of PEGMs with arbitrary dense topology can be probabilistically approximated to arbitrary accuracy in linear worst case time complexity. Indeed, it is known such a universal estimator cannot exist for the Ising model \citep{ChandrasekaranSrebroHarsha2008}. Consistent with this finding, for dense graphs with $O(p^2)$ edges and with $\phi =\text{diag}(\theta)$, the result of \cite{chatterjee2018sample} shows that the number of importance samples $N$ would scale as $O(\exp({p^2}))$, making it impossible to evaluate $T_N^{(z)}$ in polynomial time. Hence, for our proposed estimator to work well in practice, we require $|S_\theta| = O(\log p)$.
\end{rem}

\section{Applications to Fully Observed Intractable Pairwise Exponential Family Graphical Models}\label{sec:PEGM_inference}
\subsection{MLE and Bayesian inference in low dimensions}
Suppose we observe i.i.d. data $\mathbf{X} = (X_{1\bullet}, \ldots, X_{n\bullet})$ where $X_{i\bullet}$'s are observations from a PEGM with density $p_\theta(x) = q_{\theta}(x)/z(\theta)$. The maximum likelihood estimate satisfies $\hat{\theta} = \argmax_{\theta \in \Theta} \ell(\theta) $, where $\ell(\theta)= \sum_{i=1}^n \log q_\theta(X_{i\bullet}) - n\log z(\theta)$. Denote $\log q_{\theta}(\mathbf{X}) = \sum_{i=1}^n \log q_\theta(X_{i\bullet})$.  The estimate can be obtained by implementing a gradient ascent algorithm subject to the constraint $\theta \in \Theta$. We achieve this by the method of projected gradient ascent, which is a special case of proximal methods under box type constraints such as ours \citep[p.~149,][]{parikh2014proximal}. The projection operation is defined as $\mathcal{P}_{\Theta}(\theta^*) = \argmin_{\theta \in \Theta} \norm{\theta - \theta^*}_2^2$. Thus, the updates of the gradient ascent algorithm are:
\begin{align}\label{eq:projected_gradient_descent}
    \theta^{(t+1)}  = \mathcal{P}_{\Theta}(\tilde \theta^{(t+1)})&= \mathcal{P}_{\Theta}\left(\theta^{(t)} + \gamma \nabla_{\theta} \ell(\theta^{(t)})\right)
=\mathcal{P}_{\Theta}\left( \theta^{(t)} + \gamma \dfrac{\nabla_{\theta}\, q_{\theta^{(t)}}(\mathbf{X})}{q_{\theta^{(t)}}(\mathbf{X})} - \gamma n\dfrac{\nabla_{\theta} z(\theta^{(t)})}{z(\theta^{(t)})}\right),
\end{align}
where $\gamma>0$ is the step size.
For the Ising model, $\Theta = \mathbb{R}^{p \times p}$, hence $\mathcal{P}_{\Theta}$ is simply the identity map. However, for the Poisson graphical model, $\Theta = \{\theta: \theta_{jj} \in \mathbb{R}, \theta_{jk}\leq 0, \, j\neq k = 1, \ldots, p\},$ for which $\mathcal{P}_{\Theta}$ simply amounts to thresholding any positive $\{\tilde\theta^{(t+1)}\}_{jk}, j\ne k$, computed via the gradient ascent step, to zero.

For Bayesian inference, consider proper priors $\pi(\theta)$ such that $\int_\Theta \pi(d\theta) = 1$ and assume $\nabla_\theta \pi(\theta)$ exists. Then, $U(\theta) = -\log \pi(\theta \mid \mathbf{X}) = -\ell(\theta) - \log \pi(\theta) + \log m(\mathbf{X})$ where $m(\mathbf{X}) = \int_\Theta \exp(\ell(\theta)) \pi(\theta) d\theta$. Posterior sampling is usually achieved by constructing a random walk Markov kernel that transitions from the current state $\theta$ to a new candidate state $\theta'$.
Existing algorithms for doing this \citep{liang2010double, moller2006efficient, murray2012mcmc} are based on generating auxiliary samples from the candidate $\theta'$. However, it is well known that the Hamiltonian Monte Carlo sampler could be more efficient in exploring a complicated likelihood surface compared to random walk Metropolis kernels, since the former uses gradient information for the MCMC updates \citep{neal2011mcmc}. This often leads to better exploration of the target density, and the advantages could be substantial in high dimensions. One potential drawback is the gradient may not always be available. Since we have an estimate of $\nabla _\theta \ell(\theta)$ following Proposition \ref{prop:geyer_grad}, and hence of $\nabla_\theta U(\theta)$, we consider a candidate $\theta'$ constructed using reflected Hamiltonian dynamics \citep{betancourt2011nested}, so that $\theta' \in \Theta$. This proposal is then finally accepted with a Metropolis-Hastings correction \citep[Equation 3.6]{neal2011mcmc}. To complete the algorithmic specifications, in low-dimensional situations, we use a flat Gaussian prior (mean 0 and variance 100) on the parameters for Ising model. For PGM, we use Gaussian priors on the diagonal elements of $\theta$ (same parameters as Ising) and a flat exponential prior on the off-diagonal elements (rate 0.01). 
\subsection{Penalized MLE and Bayesian inference in high dimensions}\label{sec:high_dimensional_inference}
The more practically useful and statistically challenging situation arises when $p$ is moderate to high-dimensional. We consider the problem of structure learning under the assumption of a sparse underlying graph $G$. This problem was considered by \citet{meinshausen2006variable} within the context of Gaussian graphical models, \citet{ravikumar2010high} addressed Ising models, and \citet{yang2012graphical} for general PEGMs. All these approaches are based on the idea of neighborhood selection, wherein neighbors of each node are estimated separately. Essentially, these procedures rely on the pseudo-likelihood $\prod_{j=1}^{p}p_{\theta}(X_j \mid X_{-j})$ with a suitable penalty. To our knowledge, joint structure learning for general PEGMs has not been considered before. Let $\norm{M}_1 = \sum_j \sum_k |M_{jk}| $ denote the $\ell_1$-norm of the matrix $M$. We propose the following $\ell_1$-penalized estimator for $\theta$:
\begin{equation}\label{eq:l1_penalized_estimator}
    \hat{\theta}_\lambda = \argmin_{\theta \in \Theta} -\ell(\theta) + \lambda\norm{\theta}_1 = \argmin_{\theta \in \Theta} Q(\theta).
\end{equation}
For computing the estimator, a proximal optimization  algorithm \citep{parikh2014proximal} can be used with minor modifications from the updates in \eqref{eq:projected_gradient_descent}. Indeed, the objective function is equivalent to $Q(\theta) + \ind_{\theta \in \Theta}$. Since both the $\ell_1$-penalty and the parameter space $\Theta$ are convex, proximal maps of the sum of the two functions amount to the composition of the proximal maps of the functions \citep{yu2013decomposing}. This leads to batch updates of the form:
\begin{equation}\label{eq:high_dim_projected_gradient_descent}
    \theta^{(t+1)} = \mathcal{P}_\Theta \, [\text{prox}_{\ell_1}(\theta^{(t)} + \gamma \nabla_{\theta} \ell(\theta^{(t)}))],
\end{equation}
where $\text{prox}_{\ell_1}(x)$ is the well-known soft-thresholding operator applied to each element of $x$ \citep{parikh2014proximal}. The tuning parameter $\lambda$ is selected via cross-validation wherein the out-of-sample log-likelihood is maximized. This usually results in better Frobenius norms but poor structure learning. When structure learning is the main focus, we estimate the underlying graph by averaging estimates $\hat{G}_\lambda$ obtained over a grid of values for $\lambda$.  
A more detailed discussion is deferred to Supplementary Section \ref{app:computational_details}.

Fully Bayesian inference is also possible in the high-dimensional case. Prior choice for sparse parameters is a well-developed field. One popular choice is the family of scale mixtures of normal prior which is elicited as $\theta_{jk} \mid \rho^2_{jk}, \tau^2 \overset{iid}{\sim}\Gauss(0,\rho^2_{jk} \tau^2), \, \rho_{jk}\overset{iid}{\sim} f, \, \tau \sim g$ for some choice of mixing distribution $f$ and $g$. These are computationally beneficial \citep{bhattacharya2016fast} compared to the discrete mixture priors \citep{mitchell1988bayesian, george1993variable}, without compromising statistical accuracy \citep{bhattacharya2015dirichlet, chakraborty2020bayesian, bhadra2015horseshoe+}. Furthermore, posterior modes under these priors can be thought of as a penalized estimator of $\theta$ under the penalty $-\log \pi(\theta)$. In particular, when $\tau = 1$ and $\rho_{jk}^2\mid \lambda \sim \text{Exp}(1/2\lambda^2) $ \citep{park_bayesian_2008}, then $-\log \pi(\theta) = \lambda \norm{\theta}_1$, the posterior mode of which is $\hat{\theta}_\lambda$ in \eqref{eq:l1_penalized_estimator}. In our work, we choose this prior augmented by the hyperprior $\lambda \sim \text{Gamma}(a_\lambda, b_\lambda)$, i.e. if $\pi_L(\theta)$ denotes the joint prior over $\theta$, it can be specified hierarchically as $\theta_{jk}\mid \rho^2_{jk}, \lambda \overset{iid}{\sim} \Gauss(0, \rho_{jk}^2 \lambda^{-1}), \, \rho_{jk}^2 \overset{iid}{\sim}\text{Exp}(1/2)$ and $\lambda \sim \text{Gamma}(a_\lambda, b_\lambda)$. Additionally, to respect the restriction that $\theta \in \Theta$, we truncate the prior $\pi(\theta)$ to $\{\theta: \theta \in \Theta\}$ so that our working prior is $\pi(\theta) \propto  \pi_L(\theta) \mathbb{I}_{\theta \in \Theta}$. An HMC sampler can be designed using the scale mixture representation of $\pi(\theta)$. Indeed, the conditional log-posterior $U(\theta \mid \{\rho_{jk}\})$ is differentiable, and the HMC sampler can be implemented conditional on the latent $\{\rho_{jk}\}$. The latent $\{\rho_{jk}\}$'s can be updated given $\theta_{jk}$ using an inverse-Gaussian distribution. Details are given in Supplementary Section~\ref{app:computational_details}.

\section{Applications to Partially Observed Graphical Models: The Restricted and Full Boltzmann Machines}\label{sec:partial}
First consider the training of RBMs. The complete data model is Ising, which is a PEGM, and thus the proposed methodology can be used without much effort. To see this, consider the Expectation Maximization (EM) algorithm for computing the maximum likelihood estimator. Standard calculations yield the following updates for a gradient ascent algorithm for the weights between the hidden and visible variables:
\begin{align}\label{eq:RBM_grad}
\theta^{(t+1)} & = \theta^{(t)} + \gamma \left\{\mathbb{E}(\h=\mathbf{1} \mid \v, \theta = \theta^{(t)}) \v^{T} - \nabla_{\theta} \log z(\theta^{(t)})\right\} \\
           \nonumber = \theta^{(t)} + &\gamma \left\{\mathbb{E}(\h=\mathbf{1} \mid \v, \theta = \theta^{(t)}) \v^{T} - \mathbb{E}_{(\h, \v) \sim p_{\phi}} \left[ \dfrac{\nabla_{\theta}  e^{-E_{\theta}(\v, \h)}}{e^{-E_{\phi}(\v, \h)}}\right]\middle/\mathbb{E}_{(\h, \v) \sim p_{\phi}} \left[ \dfrac{e^{-E_{\theta}(\v, \h)}}{e^{-E_{\phi}(\v, \h)}} \right]\right\},
\end{align}
where the second equality follows from Proposition \ref{prop:geyer_grad}. For RBMs, $\mathbb{E}(\h=\mathbf{1} \mid \v, \theta = \theta^{(t)})$ is a vector with $k$-th element $\mathbb{P}(h_k = 1\mid \v,\theta = \theta^{(t)})$, which is a product of sigmoid functions. Updates for the biases are similar. Similar to the CD algorithm described in Section \ref{sec:BM_background}, the proposed algorithm here also makes use of Monte Carlo sampling to estimate $\nabla_{\theta} \log z(\theta)$. But, unlike CD, in this case samples $(\h, \v) \sim p_{\phi}$ are independent, and estimates the correct gradient. A Gibbs sampler is not needed.

Remarkably, the proposed framework allows a feasible training method for a full BM as well; a problem long considered intractable. Define the parameter matrix $\theta \in \mathbb{R}^{(p+m)\times (p+m)}$ as 
$\theta = \begin{pmatrix}
    \theta_{\mathbf{pp}} & \theta_{\mathbf{pm}} \\
    \theta_{\mathbf{pm}}^\T & \theta_{\mathbf{mm}} 
\end{pmatrix}$, where $\theta_{\mathbf{pp}}$ contains bias and interaction terms between the visible variables, $\theta_{\mathbf{pm}}$ contains the interaction terms between the hidden and visible variables, and $\theta_{\mathbf{mm}}$ contains bias and interaction terms between the hidden variables. At $\theta = \theta^{(t)}$, the expected value of the complete data log-likelihood given the observed $\v$ is:
\begin{align*}
    \mathbb{E}_{\theta^{(t)}, \v}\{\log p_{\theta}(\v, \h)\} =&  \sum_{j} \theta_{jj}v_j + \sum_{k} \theta_{kk}\mathbb{E}_{\theta^{(t)}, \v} (h_k) + \sum_{j \neq j'} \theta_{jj'} v_j v_{j'}  \\
    &+\sum_{k \neq k'} \theta_{kk'} \mathbb{E}_{\theta^{(t)}, \v}(h_k h_{k'})
     + \sum_{j,k} \theta_{jk} v_j \mathbb{E}_{\theta^{(t)}, \v}(h_k) - \log z(\theta),
\end{align*}
for $j,j'=1,\ldots,p;\; k,k'=p+1,\ldots, p+m$, where $\mathbb{E}_{\theta^{(t)}, \v}$ stands for the expectation taken with respect to the density $p_{\theta_t}(\cdot \mid \v)$.
Thus, to train this model using EM, one needs  $\{\mathbb{E}_{\theta_t, \v}(h_k)$,\,  $\mathbb{E}_{\theta^{(t)}, \v}(h_k h_{k'})\}$, or more generally, $\mathbb{E}_{\theta^{(t)}, \v}\{g(\h)\}$, where $g(\h) = g(h_1, \ldots, h_m)$ is a function of the hidden variables  with $g(\h) = h_k \text{ or } h_k h_{k'}$.
Fortunately, the conditional distribution of $\h \mid \v$, i.e. $p_{\theta}(\cdot \mid \v)$ here is also Ising. Indeed, 
$$p_{\theta_{\h \mid \v}}(\h \mid \v) \propto \exp\left\lbrace \sum_{k} 
(\theta_{kk} + \sum_j \theta_{jk} v_j) h_k  + \sum \sum_{k \neq k'} \theta_{kk'} h_k h_{k'}\right\rbrace = \exp({-E_{\theta_{\h\mid \v}}(\v, \h)}),$$
where the notation $\theta_{\h \mid \v}$ stands for the parameter of the conditional Ising model for $\h \mid \v$. A self-normalized importance sampling method to estimate $\mathbb{E}_{\theta, \v}\{g(\h)\}$ is now presented.
\begin{proposition}  Define $w(\v, \h) = \exp\{-E_{\theta_{\h\mid \v}}(\v, \h) + E_{\phi_{\h \mid \v}}(\v, \h)\}$. Then,  
    \begin{align*} 
    \mathbb{E}_{\theta, \v}\{g(h)\} = \mathbb{E}_{(\h \mid \v)\sim \phi_{\h \mid \v}} \left[ g(\h) \dfrac{p_{\theta_{\h \mid \v}}(\h \mid \v)}{p_{\phi_{\h \mid \v}}(\h \mid \v)}\right] = \dfrac{\mathbb{E}_{(\h \mid \v)\sim \phi_{\h \mid \v}}\left[ g(\h) w(\v, \h)\right]}{\mathbb{E}_{(\h \mid \v)\sim \phi_{\h \mid \v}}\left[ w(\v, \h)\right]}, 
    \end{align*}
where a natural choice for $\phi_{\h \mid \v}$ is $\text{diag}(\theta_{\h \mid \v})$. 
\end{proposition}
\begin{proof}
    Clearly, 
    \begin{align*}
    \mathbb{E}_{(\h \mid \v)\sim \phi_{\h \mid \v}} \left[ g(\h) \dfrac{p_{\theta_{\h \mid \v}}(\h \mid \v)}{p_{\phi_{\h \mid \v}}(\h \mid \v)}\right] = \mathbb{E}_{(\h \mid \v)\sim \phi_{\h \mid \v}} \left[ g(\h) w(\v, \h)\right] \left[\dfrac{z(\theta_{\h \mid \v})}{z(\phi_{\h \mid \v})}\right]^{-1}.
    \end{align*}
    The proof follows by applying \eqref{eq:geyer} to the second term on the right-hand side of the above. 
       \end{proof}
       Thus, the corresponding estimator is: $T_N^{(g(\h))} = \frac{\sum_{i=1}^N g(\h_i)w(\v, \h_i)}{\sum_{i=1}^N w(\v, \h_i)}$, where $\h_i \overset{iid}{\sim} p_{\phi_{\h\mid \v}}$, and the following remarks are in order.

\begin{rem}
    Variance of this self-normalized importance sampling estimate is finite when $\phi_{\h \mid \v} = \text{diag}(\theta_{\h \mid \v})$. Indeed, the estimator $Z$ is of the form $Z = X/Y $. Thus, the Cauchy-Schwarz inequality gives $\text{Var}_{\phi_{\h \mid \v}}(Z) \leq \mathbb{E}_{\phi_{\h \mid v}} (X^2/Y^2) \leq \mathbb{E}_{\phi_{\h \mid v}}(X^2) \mathbb{E}_{\phi_{\h \mid v}}(1/Y^2)$. Since $g(\h)$ is bounded for $g(\h) = h_k$ and $g(\h) = h_k h_{k'};\, k, k'=p+1,\ldots,p+m$; another application of Cauchy--Schwarz inequality along with Proposition  \ref{prop:exp} proves that $\mathbb{E}_{\phi_{\h \mid v}}(X^2) < \infty$. The random denominator $1/Y$ is of the form $\exp\{E_{\theta_{\h\mid \v}}(\v, \h) - E_{\phi_{\h \mid \v}}(\v, \h)\}$. That the variance of this quantity is finite follows from Proposition \ref{prop:exp}.
\end{rem}

\begin{rem}
    The key difference between training a BM and an RBM is that when using \eqref{eq:RBM_grad} for RBMs, one Monte Carlo estimate of $z(\theta)$ is sufficient. However, for BMs, Monte Carlo estimates of $\mathbb{E}_{\theta, \v}\{g(\h)\}$ are needed. Moreover, the parameters of $\h \mid \v$ are dependent on $\v$. A CD-type algorithm can in principle be implemented, but would be immensely expensive since one needs to run an $m$-dimensional Gibbs sampler for each $\v$. On the other hand, for us, batch sampling is possible using the specific choice $\phi_{\h \mid \v} = \mathrm{diag}(\theta_{\h  \mid \v})$.
\end{rem}
\begin{rem}\label{rm:BM_marginal}
    The training of BM also allows one to compute an estimate of the marginal likelihood of the visible variables. Indeed, for all $\h \in \{0,1\}^m$, the following relation holds: $\log p_{\hat{\theta}}(\v) = \log p_{\hat{\theta}}(\v, \h) - \log p_{\hat{\theta}}(\h \mid \v)$. Since both $(\v, \h)$ and $(\h \mid \v)$ are Ising, this quantity is computable using techniques developed here. The marginal likelihood  may then be used for model selection or hyperparameter tuning purposes.
\end{rem}

The proposed framework also encompasses the DBM where there are multiple layers of hidden variables but the underlying graph can be rearranged in a two-layer graph \citep{Aggarwal2018}, with all odd layer variables constituting one half of the graph and even layers constituting the other half, with the restriction that connections between variables in the same layer are not allowed; see Figure \ref{fig:BM}. Thus, the DBM can be thought of as a relaxation of the BM and a similar training algorithm can be implemented in principle to train the model. We omit the details. In the context of DBM,  \citet{salakhutdinov2008quantitative} proposed a related annealed importance sampling scheme \citep{neal2001annealed} where a similar choice of $\phi = \text{diag}(\theta)$ is proposed for computing marginal probabilities. But they do not discuss full likelihood inference, and still train their model via CD. 


\section{Theoretical Properties}\label{sec:theory}
In this section we investigate the properties of the likelihood-based estimates, including the Bayesian posterior, specifically in high-dimensional settings, since the properties of these procedures in low-dimensional problems are fairly well-understood. We first focus on the $\ell_1$-penalized estimator in \eqref{eq:l1_penalized_estimator}. Since its introduction in the regression context \citep{tibshirani1996regression}, $\ell_1$-penalized estimators have been used in many other settings; for instance in sparse covariance estimation problems \citep{rothman2010sparse}, reduced-rank problems \citep{yuan2007dimension}, generalized linear models \citep{ravikumar2010high, yang2012graphical}, among others. These estimators have been shown to be consistent with the optimum rate of convergence for sparse problems. We establish these properties of the estimator in \eqref{eq:l1_penalized_estimator} for PEGMs under fairly general assumptions, drawing parallels with the existing literature. In the sequel, we write $\theta_0 \in \Theta$ as the true data-generating parameter, and  $\mathbf{S} = \{(j,k): \theta_{0,jk} \neq 0, j,k = 1, \ldots, p\}$,\; $|\mathbf{S}| = s,$ where $|\cdot|$ denotes the cardinality. For matrices $A$, $B$, we write $A \otimes B$ as their Kronecker product. For two sequences $a_n , b_n$, if $a_n = C b_n$ for some $C\in(0,\infty)$, then we write $a_n \asymp b_n$.
We start by stating the assumptions.
\begin{assumption}[Dimension]\label{ass:dimension}
    We assume that $\log p = o(n)$, $s = O(\log p)$ and $s\log p = o(n)$ where $a_n = o(b_n)$ denotes $a_n/b_n \to 0$ as $n \to \infty$.
\end{assumption}

\begin{assumption}[Regularity]\label{ass:regularity}
The space $\Theta$ of feasible parameters is open and $\theta_0$ is in the interior of $\Theta$. Write $\frac{\partial^2 A(\theta)}{\partial \theta \partial \theta^\T} = G(\xi) = (G_{jk})_{j,k=1}^{p^2}$ where $\xi = (\theta \otimes \theta)$. Set $\xi_0 = (\theta_0 \otimes \theta_0)$. We also assume the PEGM is sufficiently regular, i.e., for some $\beta, \bar{\beta} > 0,$ and for any $\tilde{\Delta}^{p^2 \times p^2} $ with $\sigma_{\max}(\tilde{\Delta}) = \alpha_n$ such that $\alpha_n \to 0$, we have
    $$\lambda_{\min}\left[ G(\xi_0)\right]\geq \beta, \quad \sigma_{\max}(\theta_0) \leq \bar{\beta}, \quad 
    \sigma_{\max}\left[ D_{\xi_0^\T}\left\lbrace \mathrm{tr}(\tilde{\Delta} G(\xi))\right\rbrace\right] \asymp \sigma_{\max}(\tilde{\Delta}) ,
    $$
    where $\lambda_{\min}$ denotes the minimum eigenvalue, $\sigma_{\max}$ denotes the maximum singular value, and the matrix $D_{X_0}[f(X)]$ has elements $\frac{\partial f(X)}{\partial X_{jk}}$, evaluated at $X_0$.
\end{assumption}

\begin{assumption}[Bounded moments]\label{ass:moments}
     For all $j, k = 1, \ldots, p,$ assume, $\mathbb{E}_{\theta_0}[|T(X_j)|] \leq \kappa_1$, $\mathbb{E}_{\theta_0}[T^2(X_j)] \leq \kappa_2$, and  $\mathbb{E}_{\theta_0}[|T(X_j,X_k)|] \leq \kappa_3$. Also, assume that,
    $\max_{u: |u|<1} \dfrac{\partial^2}{\partial \theta_j^2} A(\theta_0 + ue_j) \leq \kappa_4,$
    where $e_j \in \mathbb{R}^{p^2}$ is the unit vector corresponding to the index $j$. Assume, furthermore,
    $\max_{\eta: |\eta|<1} \dfrac{\partial^2}{\partial \eta^2} \bar{A}_{jk}(\theta_0;\eta) \leq \kappa_5,$
    where, 
    $$\bar{A}_{jk}(\theta;\eta) \coloneqq \log \int_{\mathcal{X}} \exp\left\{ \eta T(X_j, X_k) + \sum_{j=1}^p  \theta_j T(X_j) + \sum_{\substack{j \neq k = 1}}^{p} \theta_{jk}T(X_j, X_k) + \sum_{j=1}^p C(X_j)\right\} \mathrm{d}x.$$
\end{assumption}
Assumption \ref{ass:regularity} imposes the condition that the covariance matrix of the model is not singular.  The condition $\sigma_{\max}(\theta_0) \leq \bar{\beta}$ together with $ \sigma_{\max}[ D_{\xi_0^\T}\{\text{tr}(\tilde{\Delta} G(\xi))\}] \leq C ||\tilde{\Delta}||_2$ ensures that the covariance matrix at $\theta_0$ is close to covariance matrices at parameters $(\theta_0 + \Delta)$, when $\norm{\Delta}_F$ is sufficiently small. For example, consider a zero mean Gaussian graphical model, which is a PEGM with tractable normalizing constant. The parameter of interest is the precision matrix. Covariance of the model at, say $\Omega_0$, is $\Omega_0^{-1} \otimes \Omega_0^{-1}$, and at $\Omega_0 + \Delta$ is $(\Omega_0 + \Delta)^{-1} \otimes (\Omega_0 + \Delta)^{-1}$. This difference is controlled by assuming $\lambda_{\max}(\Omega_0^{-1}) = \sigma_{\max}^2(\Omega_0^{-1}) < \rho$, for some $\rho>0$. The advantage here is that the canonical parameter and the mean parameter are the same \citep[Chapter 3]{wainwright2008graphical}. Hence, assuming $\lambda_{\max}(\Omega_0^{-1})< \rho$ is sufficient to ensure $|| D_{\xi_0^\T}\{\text{tr}(\tilde{\Delta} G(\xi))\}||_2$ is controlled.
Assumption \ref{ass:moments} controls the size of the deviations of the sample moments from the population moments. We have the following result, with a proof in Supplementary Section~\ref{sec:proof_l1_consistency}.
\begin{theorem}\label{thm:l1_consistency}
  Suppose we observe $n$ i.i.d copies $\mathbf{X} = (X_{1\bullet}, \ldots, X_{n\bullet})\sim PEGM(\theta)$. Let $r_n = \sqrt{\frac{s\log p}{n}}$ and $\hat{\theta}_\lambda$ be a solution to \eqref{eq:l1_penalized_estimator}. Under Assumptions \ref{ass:dimension}--\ref{ass:moments}, for $\lambda \asymp \sqrt{\frac{\log p}{n}}$, there exists sufficiently large $M>0$ such that   as $n \to \infty$, one has:
  $$\mathbb{P}_{\theta_0}\left[ \norm{\hat{\theta_\lambda} - \theta_0}_F^2 > Mr_n\right] \to 0,$$
  where $\norm{\cdot}_F$ denotes the Frobenius norm.
\end{theorem}
 \noindent The rate of convergence of the $\ell_1$-penalized estimator matches with existing rates in sparse parameter estimation theory; see for example \citet{rothman2010sparse}. {\color{red} Also see Supplementary Section~\ref{sec:conv} for further exposition on this theorem.}

Frequentist evaluation of Bayesian posteriors is typically studied through the lens of posterior consistency and contraction rates \citep{ghosal2017fundamentals}. In models with increasing number of parameters, a standard condition for posterior consistency is positive prior support on shrinking Kullback--Leibler balls around the true distribution. Suppose $p_{\theta_0}$ is the true density of the data and $p_{\theta}$ is the fitted density. Let $\pi(\theta)$ be the prior density on $\theta$ and $\text{KL}(p_{\theta_0}, p_{\theta})$ denote the Kullback--Leibler divergence between $p_{\theta_0}$ and $p_{\theta}$. In Proposition \ref{prop:KL_set}, we show when $p_{\theta_0}$ and $p_{\theta}$ are PEGMs, then the set $\{\theta: \text{KL}(p_{\theta_0}, p_{\theta})< \epsilon \}$ for any $\epsilon>0$ can be characterized in terms of $\norm{\text{vech}(\theta_{0})-\text{vech}(\theta)}_2^2$, where $\text{vech}(\theta)$ stands for the half-vectorization of a symmetric matrix $\theta$. As a result, any prior distribution that has positive mass on the set $\{\theta: \norm{\text{vech}(\theta_{0})-\text{vech}(\theta)}_2^2< \epsilon\}$ for $\epsilon >0$ yields a consistent posterior. 
Properties of global-local shrinkage priors around small Euclidean neighborhoods of a sparse parameter are well studied; see, e.g., \citet{bhattacharya2016suboptimality}. We next establish posterior concentration for PEGMs under prior $\pi(\theta) \propto \pi_L(\theta)\mathbb{I}_{\theta \in \Theta}$ where $\pi_L(\theta)$ is defined at the end of Section \ref{sec:high_dimensional_inference}. Let $H^2(f_1, f_2) = (1/2)\int (\sqrt{f_1(x)} - \sqrt{f_2(x)})^2 dx$ be the squared Hellinger distance between densities $f_1, f_2$ with respect to the Lebesgue measure. We have the following result on the rate of posterior contraction.
\begin{theorem}\label{thm:posterior_consistency}
    Suppose we observe $n$ i.i.d copies $\mathbf{X} = (X_{1\bullet}, \ldots, X_{n\bullet})\sim PEGM(\theta)$ and $\theta$ is endowed with prior $\pi(\theta) \propto \pi_L(\theta)\mathbb{I}_{\theta \in \Theta}$ where $\pi_L(\theta)$ is the marginal prior obtained under hierarchy $\theta_j \mid \lambda \overset{iid}{\sim} \text{Laplace}(\lambda^{-1})$ and $\lambda \sim \text{Gamma}(a_\lambda, b_\lambda)$. Suppose, $a_\lambda = O(s\log p)$ and $b_\lambda = O\left(s\sqrt{\frac{s\log p}{np}}\right)$. Let Assumptions \ref{ass:dimension}--\ref{ass:moments} hold, and assume that $\lambda_{\max}[G(\xi_0)] \leq \tilde{\beta}$ for some $\tilde{\beta}>0$. Set $\bar{\epsilon}_n = \{(s \log^2 p)/n\}^{1/2}$. Then, if $\pi_L(\Theta) \geq e^{-K\log p}$ for some $K>0$, we have in $\mathbb{P}_{\theta_0}-$probability: 
    $$\pi\left[ \theta : H^2(p_{\theta_0}, p_\theta) > M_n\bar{\epsilon}_n^2 \mid \mathbf{X}\right] \to 0.$$
\end{theorem}
See Supplementary Section~\ref{app:consistency_proof} for a proof. The posterior contraction rate is minimax up to a logarithmic term in the dimension $p$. The condition $\pi_L(\Theta) \geq e^{-K\log p}$ implies that the set of permissible parameters $\Theta$ is not too small relative to $\pi_L$. This is trivially satisfied for the Ising model as $\Theta = \mathbb{R}^{p(p+1)/2}$. For PGM, the corresponding probability is $e^{-2\log \{p(p+1)/2\}}$, which satisfies the condition.

We next turn our attention to partially observed models, specifically the RBM. Suppose we observe $n$ i.i.d variables $\mathbf{V} = (\v_{1\bullet}, \ldots, \v_{n\bullet})^\T$ with support on $\{0,1\}^p$.
 When an RBM is fitted to the observed data, maximum likelihood inference involves optimizing the marginal likelihood \eqref{eq:RBM_marginal}. The iterative scheme \eqref{eq:RBM_grad} essentially performs Monte Carlo EM algorithm on the complete-data log-likelihood. Let $Q(\theta; \psi) = \mathbb{E}_{\psi} [\log p_\theta (\v, \h) \mid \h]$, and $Q^\star(\theta; \psi)$ be its Monte Carlo estimate where we use $T_N(\theta) = T_N^{(\nabla_\theta z)(\theta)}/T_N^{(z)}(\theta)$ to estimate $\nabla_\theta \log z(\theta)$. Our next result shows that the sequence $\theta^{(t)}$ approaches a stationary point of the observed log-likelihood with probability 1. In spirit, our result parallels the convergence result of CD learning developed in \cite{jiang2018convergence}.
The proof is in Supplementary Section~\ref{proof:em_rbm}.
\begin{theorem}\label{thm:EM_RBM}
Consider the iterations $\theta^{(t+1)} = \theta^{(t)} + \gamma_t \nabla Q^\star(\theta^{(t)}; \theta^{(t)})$. Suppose $\gamma_t$ is a sequence of step-size satisfying $0\leq \gamma_t\leq 1$ for all $t\geq 1$ and $\sum_{t=1}^\infty \gamma_t = \infty $ and $\sum_{t=1}^\infty \gamma_t^2 < \infty$. Let $\mathcal{L} = \{\theta: \nabla_\theta \log p_\theta(\v) = 0\}$ be the set of stationary points of the observed likelihood function and $\{N_t\}$ be a sequence of increasing Monte Carlo sample sizes. Then $\lim_{t \to \infty} \inf \{\norm{\theta^{(t)} - \theta}_2: \theta \in \mathcal{L}\} \to 0$ with probability 1. 
\end{theorem}

\vspace{-0.2in}
\section{Numerical Experiments}\label{sec:sims}
We compare the performances of (maximum) likelihood-based methods (MLE), including a full Bayes method, to that of (maximum) pseudo-likelihood (MPLE) based methods for parameter inference in PEGMs. {\color{red} Additional numerical experiments summarizing results from other methods are provided in Supplementary Section \ref{sec:new_experiments}.} We consider the Ising model and the PGM as examples of PEGMs. We design our experiment in a low-dimensional (LD), high-dimensional (HD), and ultra high-dimensional (UHD) setting for both models. In the low-dimensional setting, our goal is to demonstrate the efficiency in terms of uncertainty quantification of a full likelihood-based analysis since both MLE and PMLE are consistent.  The focus in the (ultra) high-dimensional setting is to understand the performance of \emph{penalized} versions of these methods, PMLE and PMPLE, respectively, in terms of graph structure learning. We set the dimension $p = 3, 5 $ (low-dimensional) and $p = 20, 30, 40, 50$ (high-dimensional) and $p = 100, 200, 300$ (ultra high-dimensional). For both low and high-dimensional settings, the sample size is set as $n = 100$. In the ultra high-dimensional setting we set $n = 2p$. The true parameter $\theta_0$ for all the cases is generated randomly, first by sampling $z_{jk} \sim \text{Bernoulli}(\omega)$ and then setting $\theta_{0,jk}\mid [z_{jk} = 0] = 0$ and $\theta_{0,jk}\mid [z_{jk} = 1] = \eta $. In the low-dimensional setting we set $\omega = 0.9$, $\eta = -0.8$, whereas for the other two cases we set $\omega = 0.05$, $\eta = -3$. We note here that $\theta_0$ generated in this way ensures that $\theta_0 \in \Theta$, both for Ising and PGM, and the conditions of Proposition \ref{prop:exp} are satisfied. 

In the LD setting, we compare the methods using the Frobenius norm $||\hat{\theta} - \theta_0||_F^2$, average coverage probabilities of 95\% confidence intervals, and their corresponding widths. 
For the Bayes method, we choose $\hat{\theta}$ as the posterior mean using the Laplace scale mixture prior described previously. In the HD setting, we focus on the Matthews correlation coefficient (MCC) \citep{chicco2021matthews} between the estimated and the true graph and $||\hat{\theta} - \theta_0||_F^2$. The MCC values ranges between $-1$ and $+1$, with $-1$ indicating the worst performance and $+1$ the best. The advantage of MCC as a measure is that its value is close to $+1$ if and only if the estimated graph estimates zero and non-zero edges between nodes correctly at the same time. We only consider PMLE and PMPLE here, since the Bayesian posterior mean with a continuous shrinkage prior is not sparse. For both of these methods, the estimated graph is produced following the procedure described in Supplementary Section~\ref{app:computational_details}. Additionally, we report $||\hat{\theta} - \theta_0||_F^2$, where for PMLE, the estimate $\hat{\theta}$ correspond to the one obtained using cross-validation. Finally, in the UHD setting, we only consider MCC of PMLE and PMPLE. 
We summarize the Frobenius norm and MCC for the Ising model and PGM in Table \ref{tab:Ising_PGM_table} for 10 replications. The average coverage probabilities of $95\%$ confidence intervals and their corresponding widths for both models in the LD settings are in Supplementary Table~\ref{tab:CI_width}. Full likelihood-based methods achieve lower Frobenius norms in all the settings considered, for both  Ising and  PGM, compared to their PMPLE counterparts. In terms of graph reconstruction, the PMLE estimator \eqref{eq:l1_penalized_estimator} outperforms the PMPLE-based methods with significantly higher MCC values. Figure \ref{fig:Ising_stable_selection} shows the heatmap comparing the data generating parameter $\theta_0$ and the estimates obtained from PMPLE and PMLE for the Ising model with $p = 20$; the PMLE estimate corresponds to $\lambda = 7$. 
Clearly, $\hat{\theta}_{PMLE}$ effectively captures the underlying structure and also yields estimates that are closer to the truth. Finally, we note that penalization was not used in LD, and the fully Bayesian procedure was computationally prohibitive in UHD. These cases are marked with a `--' in Table~\ref{tab:Ising_PGM_table}. Additional simulation results on the partially observed case, comparing the representational powers of BM and RBM, are presented in Supplementary Section~\ref{sec:supp_bm}.

\begin{table}[!t]
    \centering
    \scalebox{.5}{
    \begin{NiceTabular}{ccc||ccc|cc||ccc|cc}
    \toprule
    & & & \multicolumn{5}{c||}{Ising} & \multicolumn{5}{c}{PGM}\\
    \midrule \addlinespace
        Setting & $p$ & $n$ & \multicolumn{3}{c|}{ $||\theta_0 - \hat{\theta}||_F^2$} & \multicolumn{2}{c||}{MCC} & \multicolumn{3}{c|}{$||\theta_0 - \hat{\theta}||_F^2$} & \multicolumn{2}{c}{MCC}  \\ \addlinespace
        \toprule
        & & & PMLE & PMPLE & Bayes & PMLE & PMPLE  & PMLE & PMPLE & Bayes & PMLE & PMPLE  \\ \addlinespace
       \midrule
        \multirow{2}{*}{LD} & 3 & \multirow{2}{*}{100} & \textbf{1.24 (0.52)} & 6.48 (6.05)  & 2.09 (1.62) & -- & -- & \textbf{0.38 (0.09)} & 5.83 (5.65) &  0.70 (0.02)  & --& --\\
         & 5 &  & 2.07 (0.85) & 16.23 (6.42)  & \textbf{1.68 (0.28)} & -- & -- &  \textbf{0.98 (0.27)} & 19.57 (5.21) & 1.37 (0.30)   & --& --\\
         \midrule \addlinespace
         \multirow{2}{*}{HD} & 20 &\multirow{4}{*}{100} & \textbf{9.18 (0.61)}  & 9.56 (0.57)  & 10.58 (0.27)& \textbf{0.67 (0.05)} & 0.48 (0.01) & \textbf{11.60 (0.43)} & 13.62 (0.17) & 13.91 (0.05)& \textbf{0.83 (0.07)} & 0.24 (0.02)\\
          & 30 &  & \textbf{12.85 (0.99)} & 13.94 (0.61) & 16.24 (0.38) & \textbf{0.78 (0.02)} & 0.49 (0.02) & \textbf{15.27 (1.25)} & 20.43 (0.12) & 20.94 (0.01)& \textbf{0.71 (0.05)} & 0.22 (0.02)\\
          & 40 & & \textbf{18.61 (1.31)} & 27.38 (0.11) & 24.33 (0.29)  & \textbf{0.71 (0.05)} & 0.46 (0.03) & \textbf{21.00 (1.03)} & 27.06 (0.07) & 27.26 (0.01) & \textbf{0.73 (0.04)} & 0.25 (0.02)\\
          & 50 & & \textbf{26.89 (0.50)} & 35.38 (0.09) & 32.10 (0.21) & \textbf{0.71 (0.03)} & 0.39 (0.03)& \textbf{30.57 (1.32)} & 34.36 (0.15) & 33.85 (0.01)& \textbf{0.61 (0.04)} & 0.22 (0.02) \\
          \midrule \addlinespace
         \multirow{2}{*}{UHD} & 100 & 200 & \textbf{54.35 (0.93)} & 69.71 (0.04)  & -- & \textbf{0.80 (0.02)} & 0.44 (0.01) & \textbf{59.26 (0.62)} & 67.43 (0.07) & -- & \textbf{0.52 (0.02)} & 0.27 (0.01) \\
          & 200 & 400& \textbf{128.97 (0.14)} & 137.16 (0.02) & -- & \textbf{0.70 (0.01)} & 0.41 (0.01) & \textbf{131.84 (0.03)} & 134.70 (0.05) & -- & \textbf{0.47 (0.01)} & 0.34 (0.01) \\
         & 300 & 600 & \textbf{195.97 (0.06)} & 202.21 (0.02)& -- & \textbf{0.56 (0.01)} & 0.24 (0.01) & \textbf{201.11 (0.03)} & 203.73 (0.05) & -- & \textbf{0.50 (0.01)} & 0.33 (0.01)\\
         \bottomrule
    \end{NiceTabular}
    
    }
    \caption{Mean (sd) of MSE = $||\hat{\theta} - \theta_0||_F^2$ and Matthews correlation coefficient (MCC) across 10 datasets for penalized maximum likelihood estimation (PMLE), penalized maximum pseudo-likelihood estimation (PMPLE) and the Bayes methods for Ising and PGM.  Best performers are bolded.}
    \label{tab:Ising_PGM_table}
\end{table}
\begin{figure}[!h]
\centering
\includegraphics[width=0.23\textwidth, angle = 270]{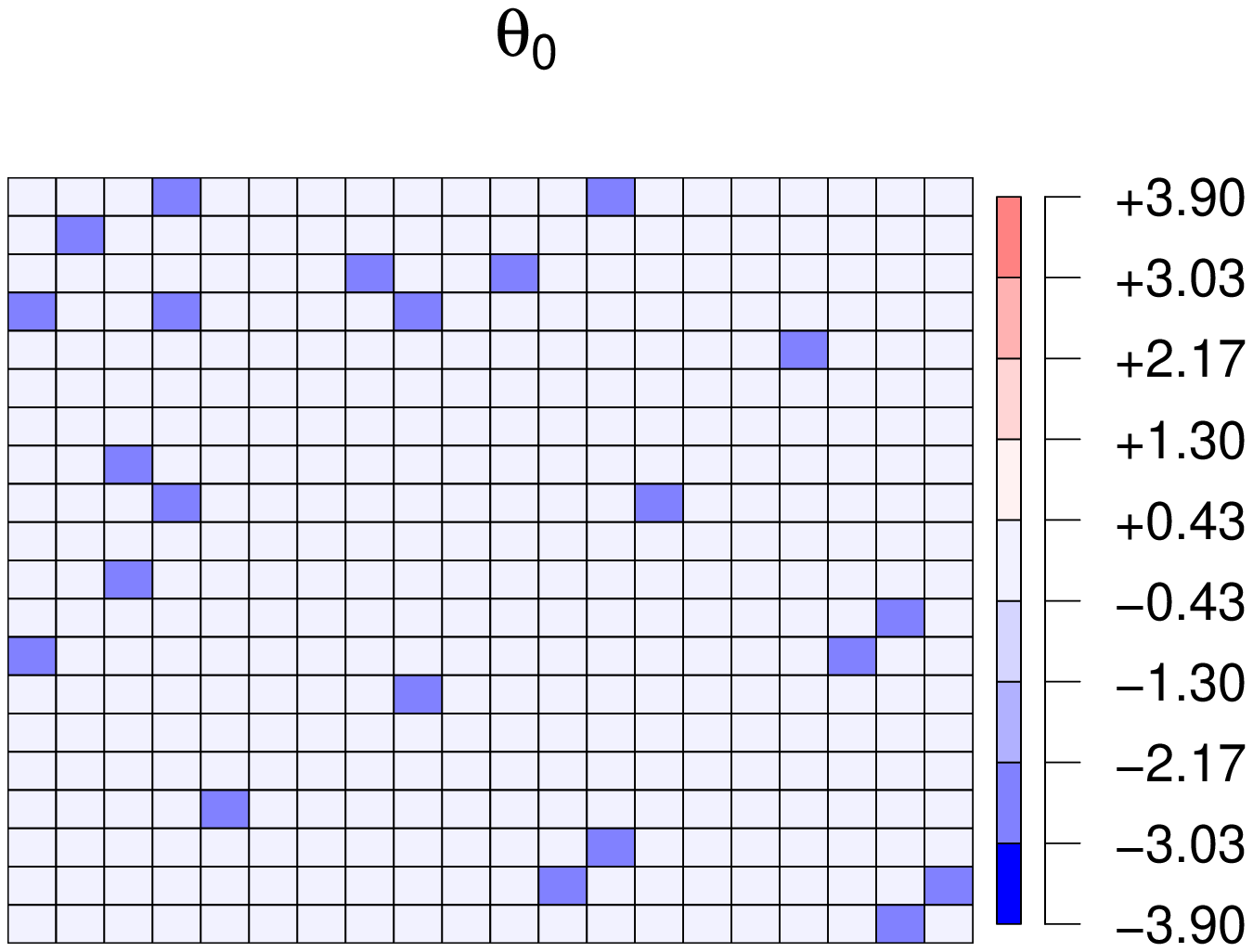} \quad\includegraphics[width=0.23\textwidth, angle = 270]{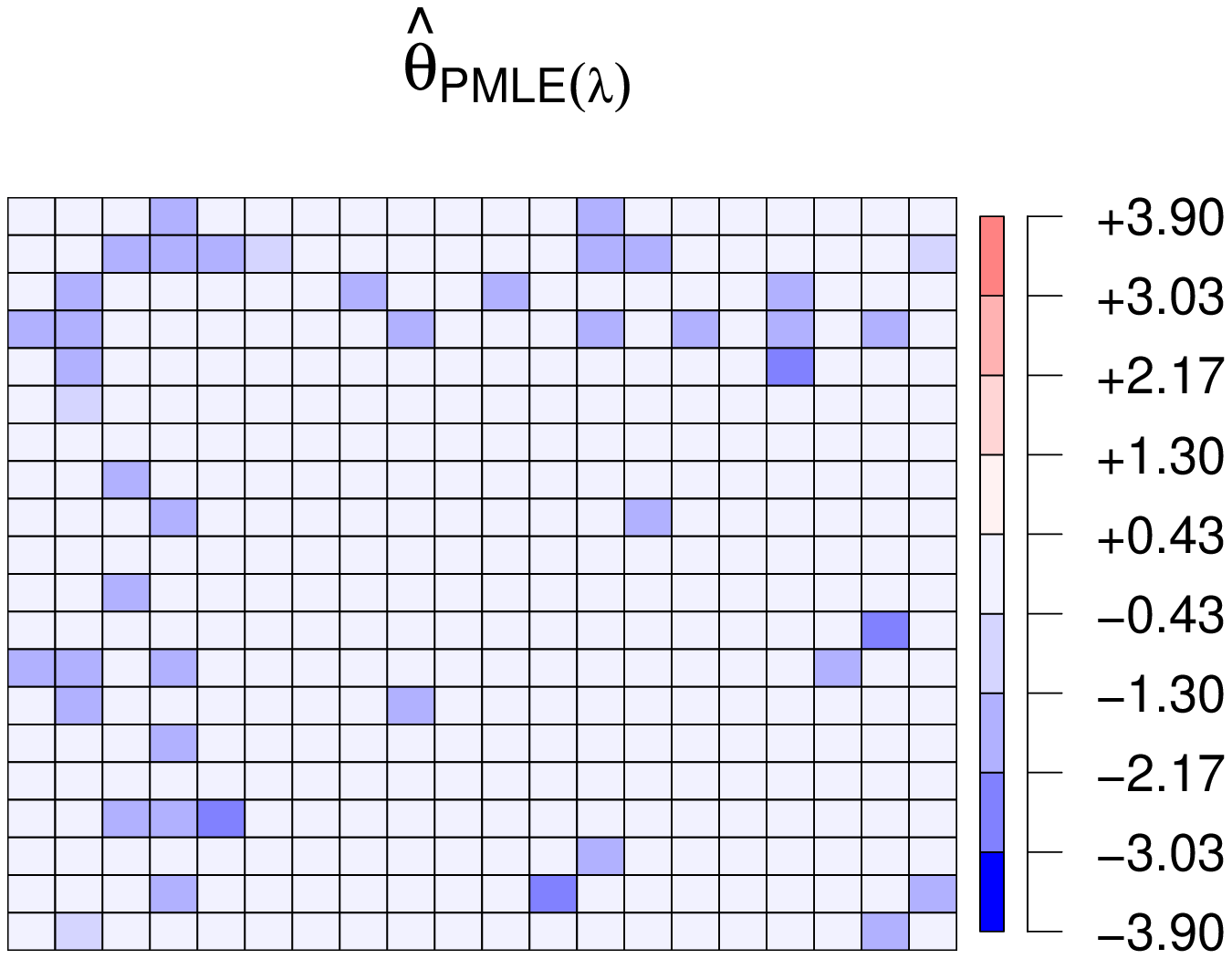}\quad \includegraphics[width=0.23\textwidth, angle = 270]{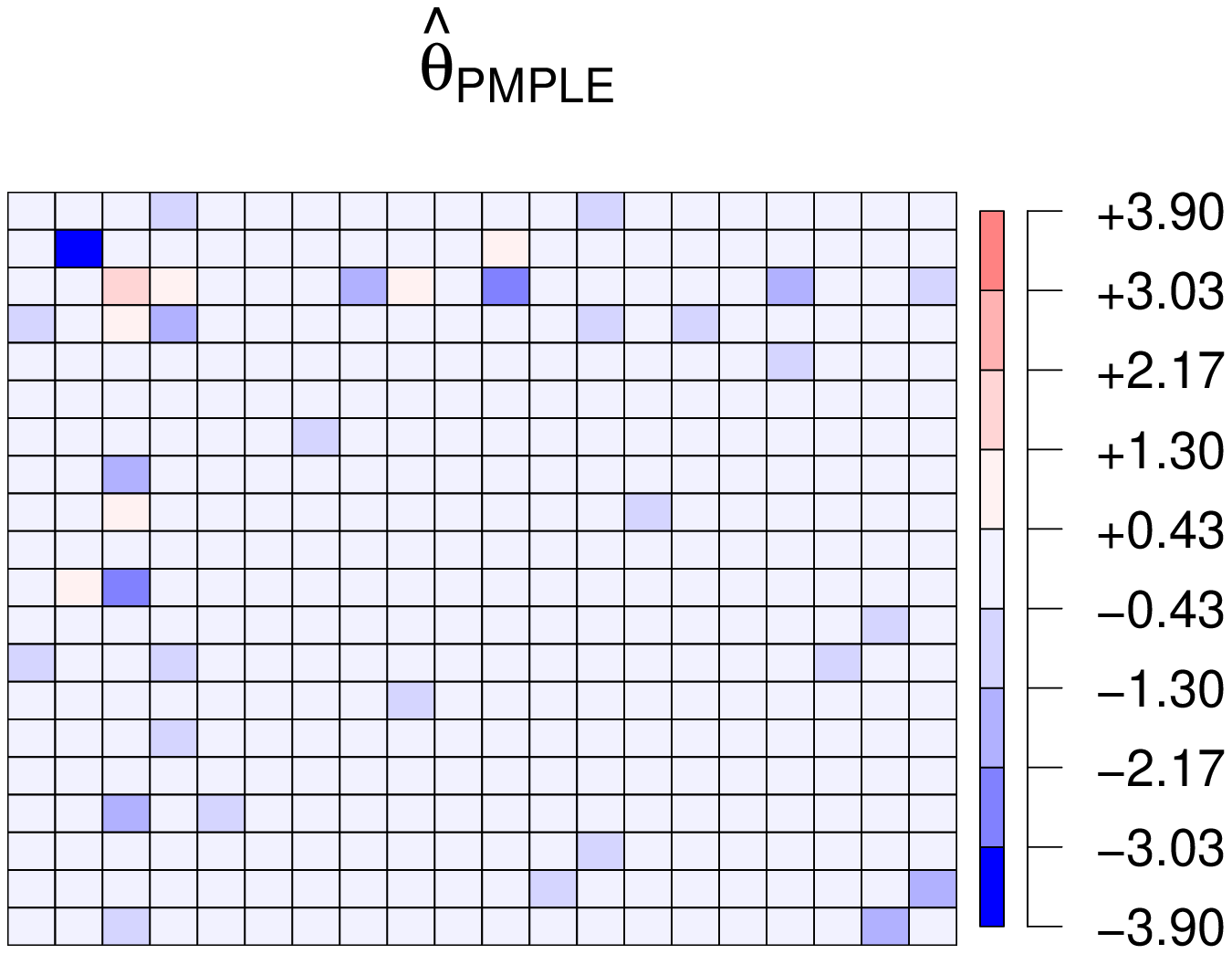}
    \caption{True parameter $\theta_0$ and the corresponding PMLE and PMPLE  estimates $\hat{\theta}$ for the Ising model with $(n,p)=(100,20)$ and $\lambda = 7$.}
    \label{fig:Ising_stable_selection}
\end{figure}

\section{Data Analysis}\label{sec:real_data}
We provide demonstrations on data sets arising from disparate applications, intentionally so chosen, to demonstrate the scope of the methodology.

\subsection{Movie ratings network via Ising model}\label{sec:Ising_movie}

We demonstrate the merits of a likelihood-based analysis for the Ising Model by analyzing a real-world movie rating dataset. The original $Movielens \ 25M$ dataset collected by Grouplens consists of movie ratings, ranging from 0 to 5 with increments of 0.5,  for 62,000 movies by 162,000 users (\href{https://grouplens.org/datasets/movielens/}{https://grouplens.org/datasets/movielens/}). We select the most frequently rated 50 movies ($p = 50$) from the same group of $n = 303$ viewers. We code observations that received a rating of 5 as `1', signifying complete satisfaction with the film, and assign `0' to ratings of 4.5 or below \citep{roach2023graphical}. Supplementary Table~\ref{tab:movie_id} provides the legend of the movie titles, genres, and years of release.
The binary variables $X_{ij}$ indicates the preference of the $i$th viewer for the $j$th movie, where $i = 1, \ldots, 303$ and $j = 1, \ldots, 50$. Specifically, $X_{ij} = X_{ik} = 1$, where $j \neq k$, implies that the $i$th viewer has strong preference for both movies $j$ and $k$. Other combinations can also be interpreted similarly. We model the distribution of $X$ using $\text{Ising}(\theta)$ where $\theta \in \mathbb{R}^{p \times p}$. Diagonal elements of $\theta$ model the propensity of a movie being rated 1, whereas off-diagonal elements capture the direction of association between movies $(j,k)$ conditional on other movies.

We estimate $\theta$ using (\ref{eq:l1_penalized_estimator}) and build the graph structure using the stable selection algorithm described in the Supplementary Section~\ref{app:computational_details}. For stable selection, we select a grid of values of $\lambda = 5, 10, \dots, 100 $ and set threshold $\pi_{thr} = 0.6$. The resulting movie rating network with abbreviations of the movie titles is plotted in Supplementary Figure \ref{fig:Ising_movie_sw_lr}, where the edge widths are proportional to the stable selection probability, and the node sizes are proportional to the degree. Our analysis reveals that \textit{Star Wars IV} and \textit{Star Wars V} have the most connections, partly due to their appeal to a broader audience, see Supplementary Table~\ref{tab:movies_top}. \textit{Lord of the Rings: The Fellowship of the Ring} has the second most connections. These two movie franchise also form two separate cliques within the graph, see Supplementary Figure~\ref{fig:Ising_movie_sw_lr}. A more detailed summary of the results, in terms of the most positive and negative connections, is provided in Table \ref{tab:movie_top_pos_neg}. Movies within the \textit{Lord of the Rings} and \textit{Star Wars} franchises have strong positive connections, whereas \textit{Inception (2010)} and  \textit{Batman (1989)} have the strongest negative connection. We remark here that \textit{Batman (1989)} is directed by Tim Burton, unlike the popular Batman movie franchise of the 2000s directed by Christopher Nolan, who also directed~\textit{Inception (2010)}. A heatmap of $\hat{\theta}_\lambda$ for $\lambda = 10$ is in Supplementary Figure \ref{fig:Ising_movie_estimated_theta}.

\begin{table}[!t]
    \centering
    \scalebox{0.6}{
    \begin{tabular}{c|c||c|c}
    \hline
\textbf{Positive Edge} & \textbf{$\hat{\theta_{jk}}$} & \textbf{Negative Edge} & \textbf{$\hat{\theta_{jk}}$} \\

\hline
Lord of the Rings (2003)  -  Lord of the Rings (2001) & 0.77 & Inception (2010) - Batman (1989) &  -0.95  \\        
Lord of the Rings (2002) - Lord of the Rings (2001)    & 0.72 & Terminator (1984) - Dances with Wolves (1990) & -0.80 \\    
Lord of the Rings (2003) - Lord of the Rings (2002) & 0.61 & True Lies (1994) - Star Wars: Episode IV (1977) & -0.53 \\
Star Wars: Episode V (1980) - Star Wars: Episode IV (1977) & 0.56 & Memento (2000) - Fugitive (1993) & -0.49\\ 
Terminator (1984) - Terminator 2: Judgment Day (1991) & 0.55 & Inception (2010) - Terminator (1984) &  -0.48 \\
Star Wars: Episode VI (1983) - Star Wars: Episode IV (1977) & 0.53 & Dances with Wolves (1990) - Twelve Monkeys (1995) &  -0.44 \\ 
Star Wars: Episode VI (1983) - Star Wars: Episode V (1980) & 0.48 & Independence Day (1996) - Batman (1989)  &  -0.44  \\ 
Raiders of the Lost Ark  (1981) - Star Wars: Episode V  (1980) &0.40 & Godfather (1972) - Independence Day (1996) & -0.42  \\
Godfather (1972) - Schindler's List (1993) &  0.25 & Inception (2010)-Independence Day (1996) &  -0.41\\ 
Raiders of the Lost Ark (1981) - Star Wars: Episode IV (1977) & 0.24 & Braveheart (1995) - Toy Story (1995) & -0.39 \\
\hline
\end{tabular}
}
\caption{Top 10 positive and negative interactions in the Movie Ratings Network.}
\label{tab:movie_top_pos_neg}
\end{table}

\subsection{Breast cancer microRNA network via Poisson graphical model}

 We next apply a full-likelihood ($\hat\theta_{PMLE}$) and pseudo-likelihood ($\hat\theta_{PMPLE}$) analysis to the task of inferring dependency among count-valued variables that represent microRNA (miRNA) expression profiles of the $353$ genes. This level III RNA-Seq data is collected from 445 breast invasive carcinoma (BRCA) patients ($n = 445$) from the Cancer Genome Atlas (TCGA) project \citep{cancer2012comprehensive} and is available in the \texttt{R} package \texttt{XMRF}. The observed values are counts of sequencing reads mapped back to a reference genome. Raw counts are highly skewed and zero-inflated.  Following the processing pipeline in \citet{allen2012log}, we use the \texttt{processSeq} function in the \texttt{XMRF} package to adjust for sequencing depth, mitigate overdispersion, and normalize the counts. This process includes a quantile matching step to preserve the integer-valued nature of the data. Since the PGM can only model negative conditional associations, we select a subset of $p = 101$ genes with negative Pearson correlation between their counts.

We select a grid of $\lambda_{PMLE} = [200, 240, 280, \dots, 1000]$, and set the threshold $\pi_{thr} = 0.6$ to implement the stable selection algorithm. For pseudo-likelihood, we perform node-wise $\ell_1$ penalized Poisson generalized linear model (GLM) using \texttt{glmnet} \citep{friedman2010regularization} in \texttt{R} with a negative constraint on coefficients and consider a grid of 20 evenly spaced values in the interval $[5\sqrt{{\log p}/{n}}, \; 7\sqrt{{\log p}/{n}}]$ for the penalty parameter $\lambda_{PMPLE}$. We then construct the final estimator $\hat{\theta}_{\lambda, PMPLE}$ for a given $\lambda$ by setting $\hat{\theta}_{\lambda, PMPLE} = (\hat{\theta}_{\lambda} + \hat{\theta}^\T_{\lambda})/2$, where $\hat{\theta}_{\lambda}$ represents the matrix formed by coefficients from Poisson GLMs. The sequence of estimates $\hat{\theta}_{\lambda, PMPLE}$ are then used for model selection following the stable selection procedure with $\pi_{thr} = 0.6$.

Estimated networks from a full-likelihood and pseudo likelihood-based analysis are comparable. Supplementary Figures~\ref{fig:PGM_breast_cancer_graph_MLE_neg} and~\ref{fig:PGM_breast_cancer_graph_MPLE_neg} show the estimated networks from a likelihood and pseudo-likelihood analysis, respectively, where blue edges are common to both networks. 
Let $S_{PMLE}:=\{ (j,k): \hat{\theta}_{PMLE, jk} \neq  0, 1\leq j \leq k \leq p \}$ and  $S_{PMPLE}:=\{ (j,k): \hat{\theta}_{PMPLE, jk} \neq  0, 1\leq j \leq k \leq p \}$. The cardinality of the intersection $|S_{PMLE} \cap S_{PMPLE}| = 372$ is about $70\%$ of $|S_{PMLE}|$. Moreover, the genes with the top numbers of connections in the PMLE-based network also have high degrees of rank in the PMPLE-based network, as shown in Table~\ref{tab:cancer_nodes}. Both networks successfully identify the edge between the CDH1 (located on chromosome 16q) and the BCL9 (located on chromosome 1q). Existing studies reported that the CDH1 is often underexpressed, while the BCL9 is typically overexpressed, caused by frequent chromosomal aberrations involved in breast cancer pathogenesis \citep{privitera2021aberrations}. Additionally, the top connected genes in Table~\ref{tab:cancer_nodes}, such as WHSC1L1, TMPRSS2, CDH1, SLC34A2 and EP300, are all known key players in breast cancer \citep{xiao2022tmprss2, kim2021high}. 


\begin{table}[h!]
\centering
    \scalebox{0.65}{
    \begin{tabular}{c|cccccccccc}
    
    \hline
    \textbf{Gene} & WHSC1L1 & TMPRSS2  & CDH1  & SLC34A2 & FGFR2  & EP300 & RNF43 & NUP98  & ACSL3 & MAP2K1  \\
    \hline
    $d_{PMLE}$ (\textbf{Ranking})& 52 (1) & 36 (2) & 24 (3) & 21 (4) & 21 (5) & 20 (6) & 19 (7) & 19 (8) & 19 (9)& 19 (10)\\
     $d_{PMPLE}$ (\textbf{Ranking}) & 33 (3) & 15 (43) & 34 (2) & 26 (7) & 19 (25) & 17 (32) & 20 (21) & 17 (32) & 14 (49) & 7 (84)\\

        \hline
    \end{tabular}
    }
    \caption{Top 10 genes ranked by degrees from the full likelihood-based ($d_{PMLE}$) and pseudo likelihood-based ($d_{PMPLE}$) estimates for Breast Cancer data.}
    \label{tab:cancer_nodes}
\end{table}

\subsection{Analysis of MNIST data via restricted and full Boltzmann machines} \label{sec:PGM_breast}

We analyze the MNIST data \citep{deng2012mnist} using RBM and BM where we train the model with the proposed algorithm separately for each digit. {\color{red} Results for RBM trained with annealed importance sampling \citep{neal2001annealed}, which is also a likelihood-based approach, are also included.}
The data consists of images of handwritten digits ($0$--$9$). The number of training samples for the $10$ digits vary within the range $n=3000$ to $n=4000$. We consider images of 15$\times$15 resolution. Each image is then converted to a binary vector of length $p = 225$, with on pixels are coded as 1 and off pixels are coded as 0. A fitted RBM or BM model on these images approximates the distribution of this binary vector over $\{0,1\}^{225}$. Furthermore, since both RBM and BM are generative models, the trained models can be used to produce synthetic data based on held out test images of these digits. For example, suppose an RBM is trained for the digit $0$ yielding an estimate $\hat{\theta}$. Given a new test image $\v_{new}$ of the digit $0$, we can generate $\h_{new}$ by simulating from $\h_{new}\mid \v_{new}, \hat{\theta}$. The generated $h_{new}$ can then be used to generate $\v_{recon}$ using $(\v_{recon}\mid \h_{new},\hat{\theta}).$ 
This is known as reconstruction. A well-trained model will generally capture the global distribution of the pixels for each digit, and be able to reproduce subject-specific variations within this global distribution. In Figure \ref{fig:reconstruction}, we show the reconstructed images of test images for all the digits $0$--$9$ with $m = 50$ hidden variables for both RBM and BM. Additionally, we compare the performance of reconstructed test images in terms of Brier loss. The Brier loss measures how well-calibrated the reconstructed probabilities are. This is done on test images of all the digits. Given a test image $\v$, and reconstructed probabilities $\hat{\v}$, we compute the Brier loss as 
$$B = (n_t p)^{-1} \sum_{i=1}^{n_t} \sum_{j =1}^p (\v_{ij} - \hat{\v}_{ij})^2,$$
where $n_t$ is the number of test images. The results are reported in Table \ref{tab:MNIST_brier} {\color{red} along with rumtimes of all methods averaged over all digits. Performance of the likelihood-based RBM (RBM-MLE, RBM-AIS) is better than CD-trained RBM for all digits while RBM-MLE, which is the proposed approach, fares much better when it comes to runtime. Also, with the same number of hidden variables, BM-MLE achieves a (marginally) lower reconstruction error compared to RBM-MLE in most cases.

\begin{figure}[!t]
   \begin{subfigure}{0.10\textwidth}
       \includegraphics[width = 2cm, height = 2cm]{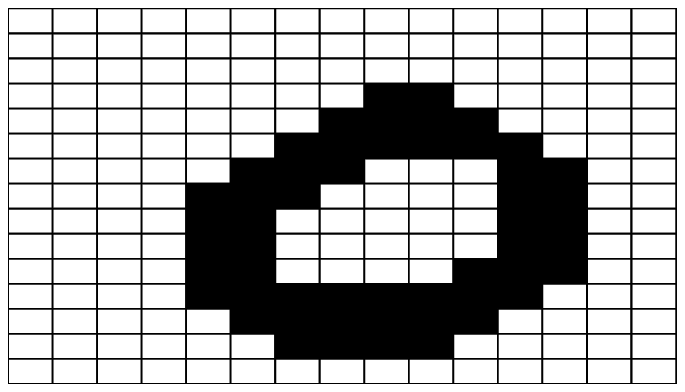}
   \end{subfigure}
   \hspace{-0.1in}
   \begin{subfigure}{0.10\textwidth}
       \includegraphics[width = 2cm, height = 2cm]{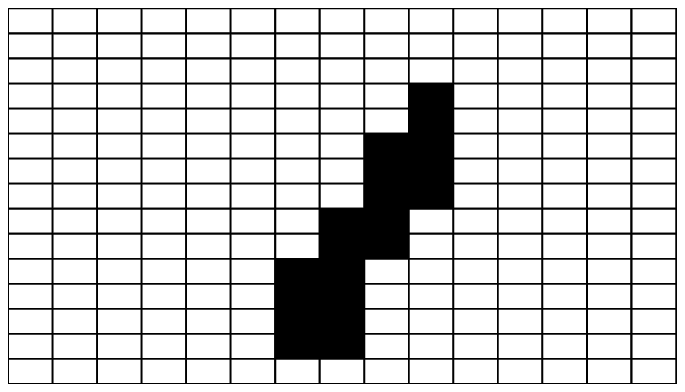}
   \end{subfigure}
   \hspace{-0.1in}
   \begin{subfigure}{0.10\textwidth}
       \includegraphics[width = 2cm, height = 2cm]{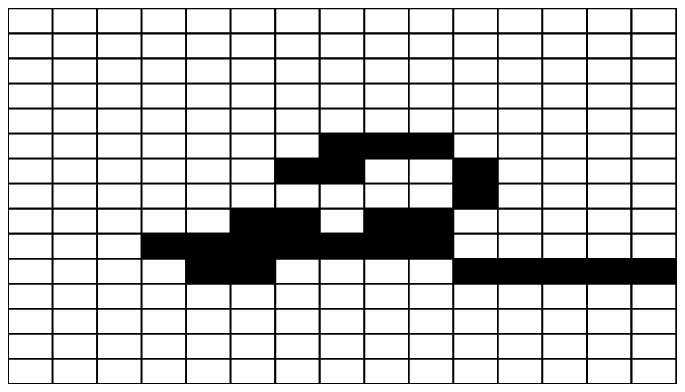}
   \end{subfigure}
   \hspace{-0.1in}
   \begin{subfigure}{0.10\textwidth}
       \includegraphics[width = 2cm, height = 2cm]{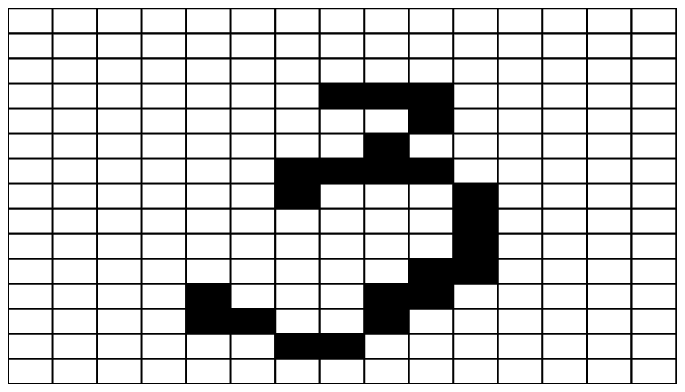}
   \end{subfigure}
  \hspace{-0.1in}
   \begin{subfigure}{0.10\textwidth}
       \includegraphics[width = 2cm, height = 2cm]{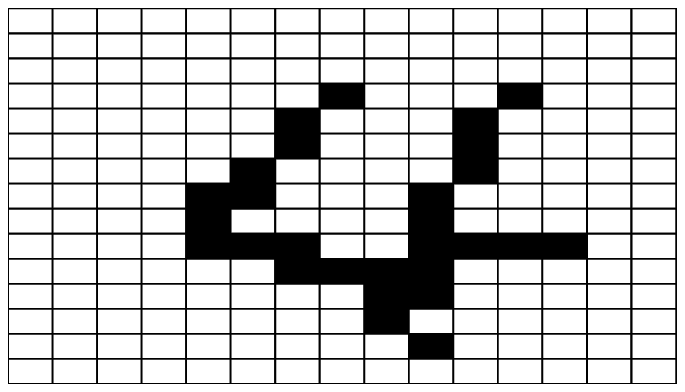}
   \end{subfigure}
   \hspace{-0.1in}
   \begin{subfigure}{0.10\textwidth}
       \includegraphics[width = 2cm, height = 2cm]{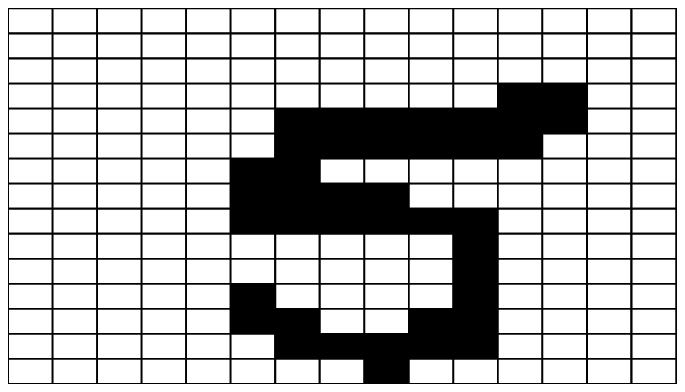}
   \end{subfigure}
   \hspace{-0.1in}
   \begin{subfigure}{0.10\textwidth}
       \includegraphics[width = 2cm, height = 2cm]{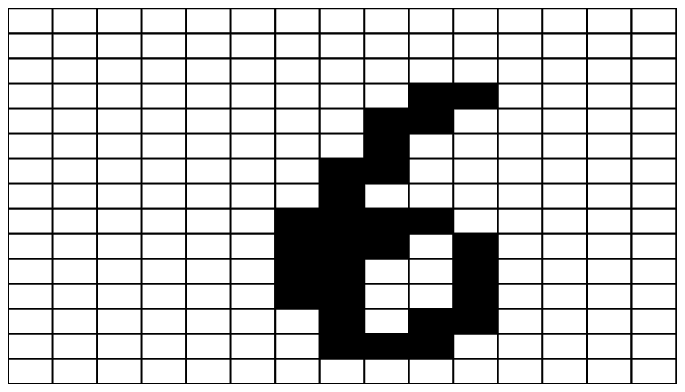}
   \end{subfigure}
   \hspace{-0.1in}
   \begin{subfigure}{0.10\textwidth}
       \includegraphics[width = 2cm, height = 2cm]{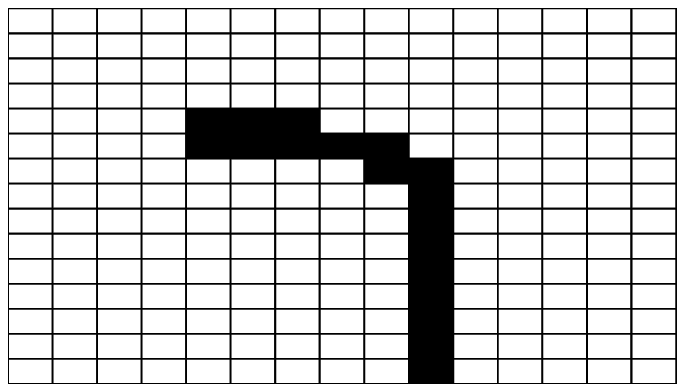}
   \end{subfigure}
   \hspace{-0.1in}
   \begin{subfigure}{0.10\textwidth}
       \includegraphics[width = 2cm, height = 2cm]{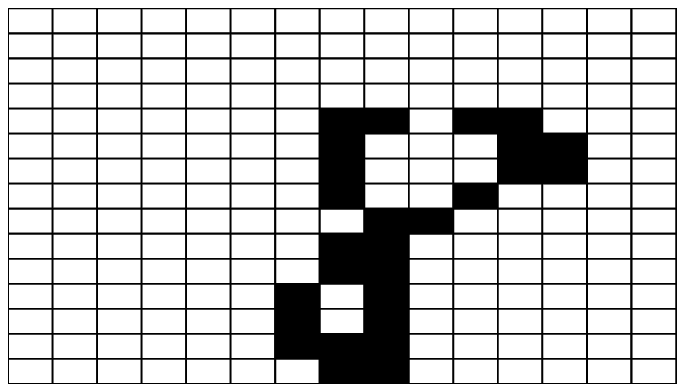}
   \end{subfigure}
   \hspace{-0.1in}
   \begin{subfigure}{0.10\textwidth}
       \includegraphics[width = 2cm, height = 2cm]{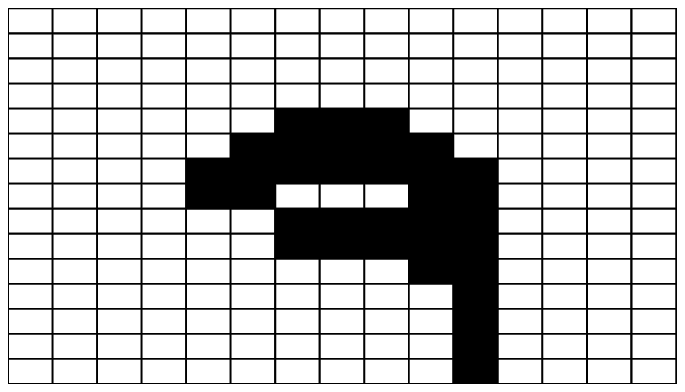}
   \end{subfigure}
   \begin{subfigure}{0.10\textwidth}
       \includegraphics[width = 2cm, height = 2cm]{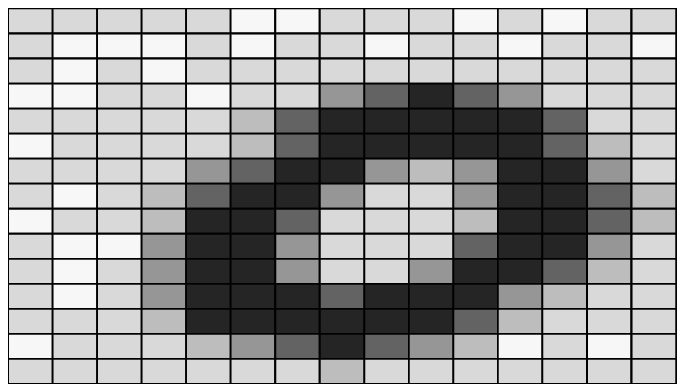}
   \end{subfigure}
   \hspace{-0.1in}
   \begin{subfigure}{0.10\textwidth}
       \includegraphics[width = 2cm, height = 2cm]{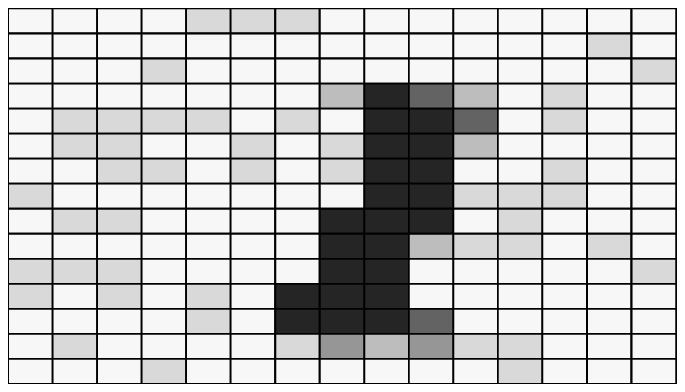}
   \end{subfigure}
   \hspace{-0.1in}
   \begin{subfigure}{0.10\textwidth}
       \includegraphics[width = 2cm, height = 2cm]{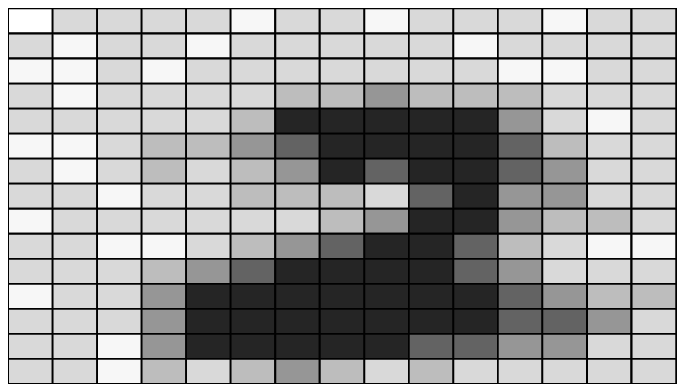}
   \end{subfigure}
   \hspace{-0.1in}
   \begin{subfigure}{0.10\textwidth}
       \includegraphics[width = 2cm, height = 2cm]{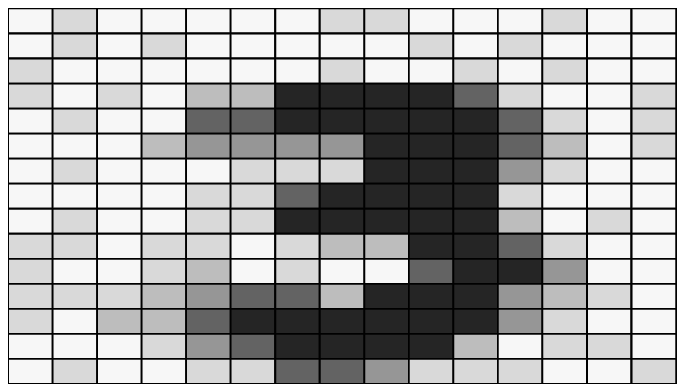}
   \end{subfigure}
  \hspace{-0.1in}
   \begin{subfigure}{0.10\textwidth}
       \includegraphics[width = 2cm, height = 2cm]{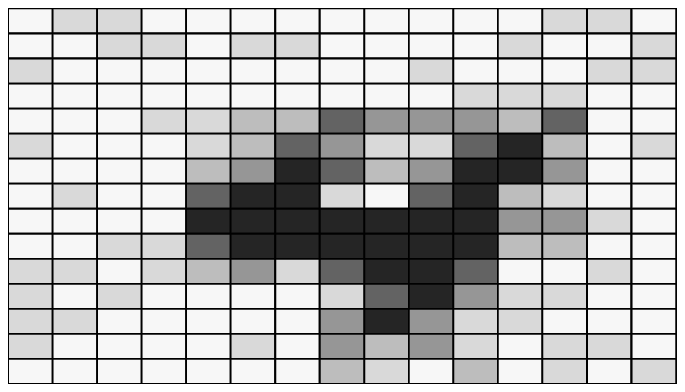}
   \end{subfigure}
   \hspace{-0.1in}
   \begin{subfigure}{0.10\textwidth}
       \includegraphics[width = 2cm, height = 2cm]{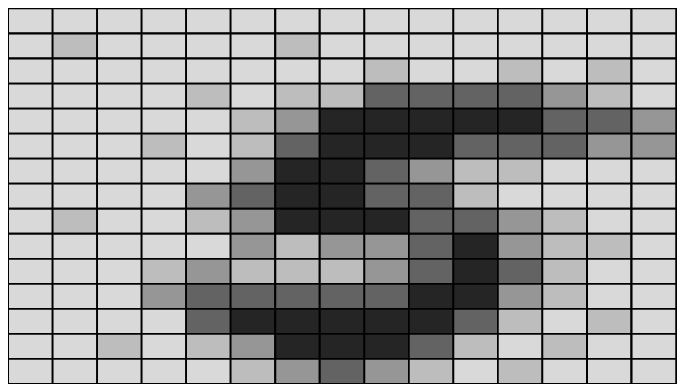}
   \end{subfigure}
   \hspace{-0.1in}
   \begin{subfigure}{0.10\textwidth}
       \includegraphics[width = 2cm, height = 2cm]{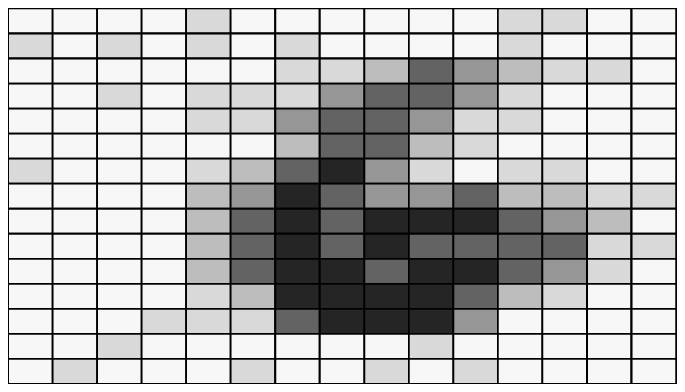}
   \end{subfigure}
   \hspace{-0.1in}
   \begin{subfigure}{0.10\textwidth}
       \includegraphics[width = 2cm, height = 2cm]{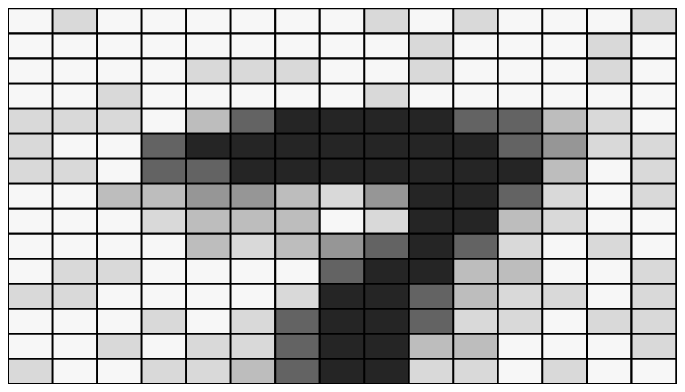}
   \end{subfigure}
   \hspace{-0.1in}
   \begin{subfigure}{0.10\textwidth}
       \includegraphics[width = 2cm, height = 2cm]{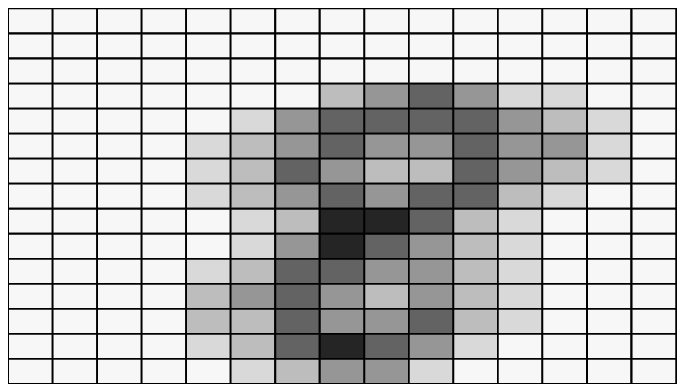}
   \end{subfigure}
   \hspace{-0.1in}
   \begin{subfigure}{0.10\textwidth}
       \includegraphics[width = 2cm, height = 2cm]{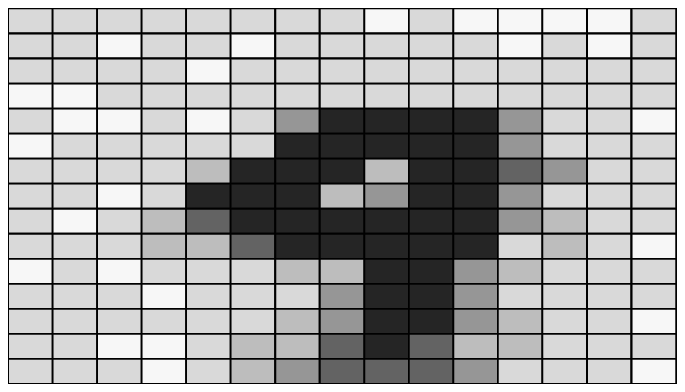}
   \end{subfigure}
    \begin{subfigure}{0.10\textwidth}
       \includegraphics[width = 1.85cm, height = 1.85cm]{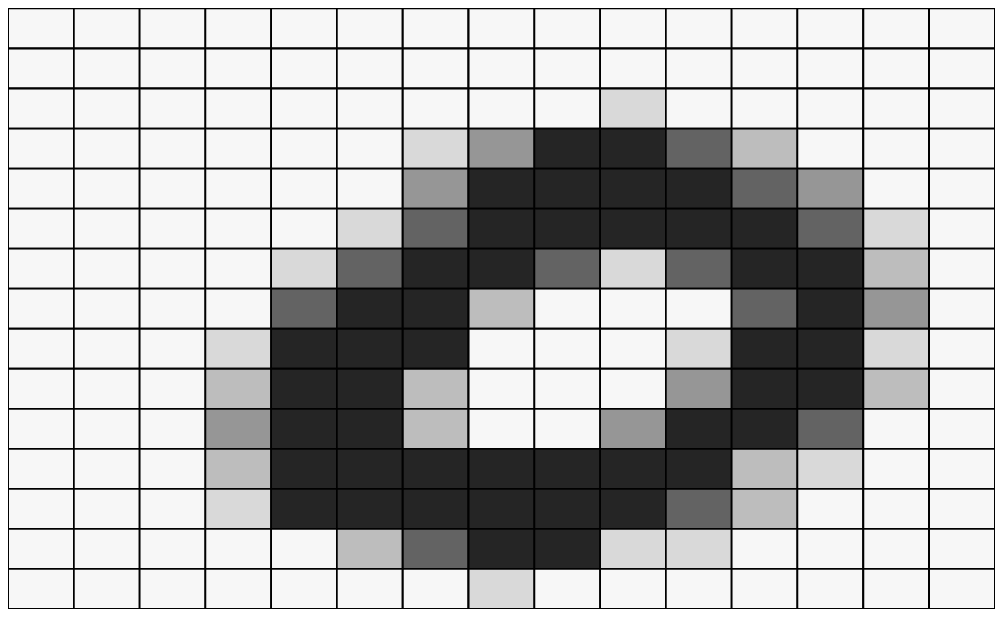}
   \end{subfigure}
   \hspace{-0.1in}
   \begin{subfigure}{0.10\textwidth}
       \includegraphics[width = 1.85cm, height = 1.85cm]{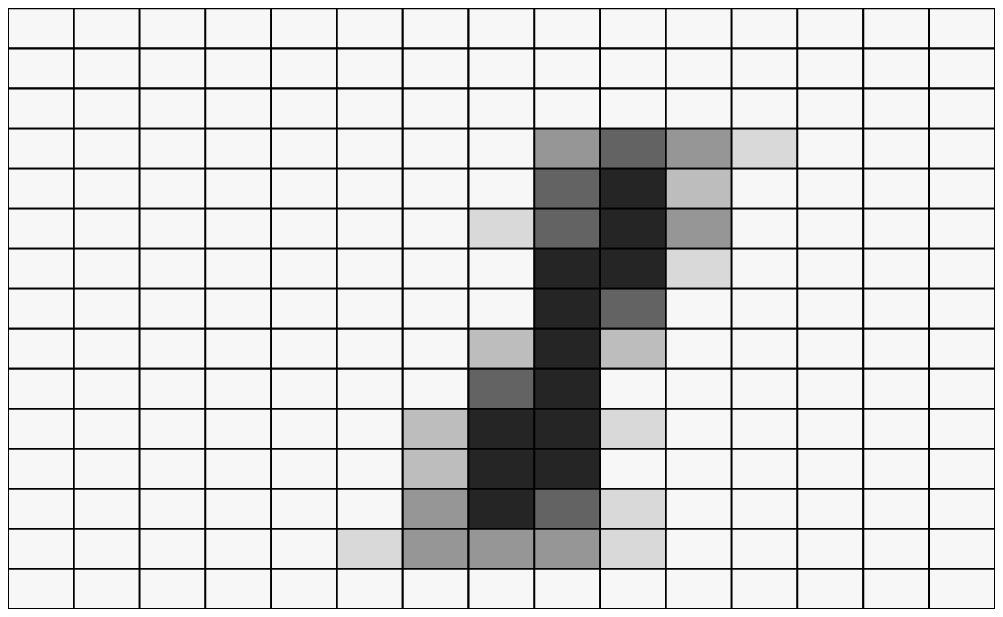}
   \end{subfigure}
   \hspace{-0.1in}
   \begin{subfigure}{0.10\textwidth}
       \includegraphics[width = 1.85cm, height = 1.85cm]{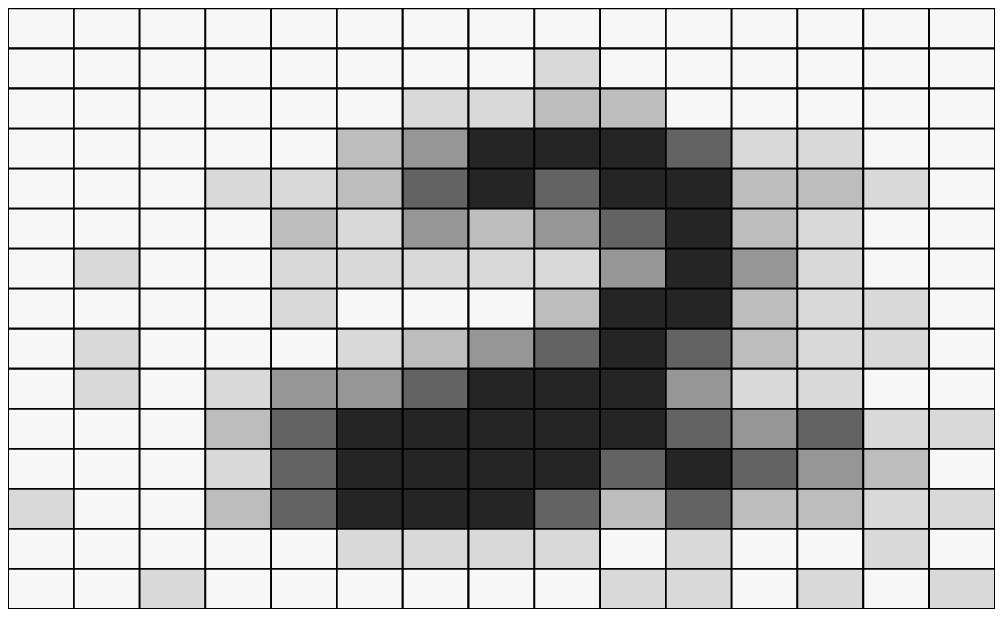}
   \end{subfigure}
   \hspace{-0.1in}
   \begin{subfigure}{0.10\textwidth}
       \includegraphics[width = 1.85cm, height = 1.85cm]{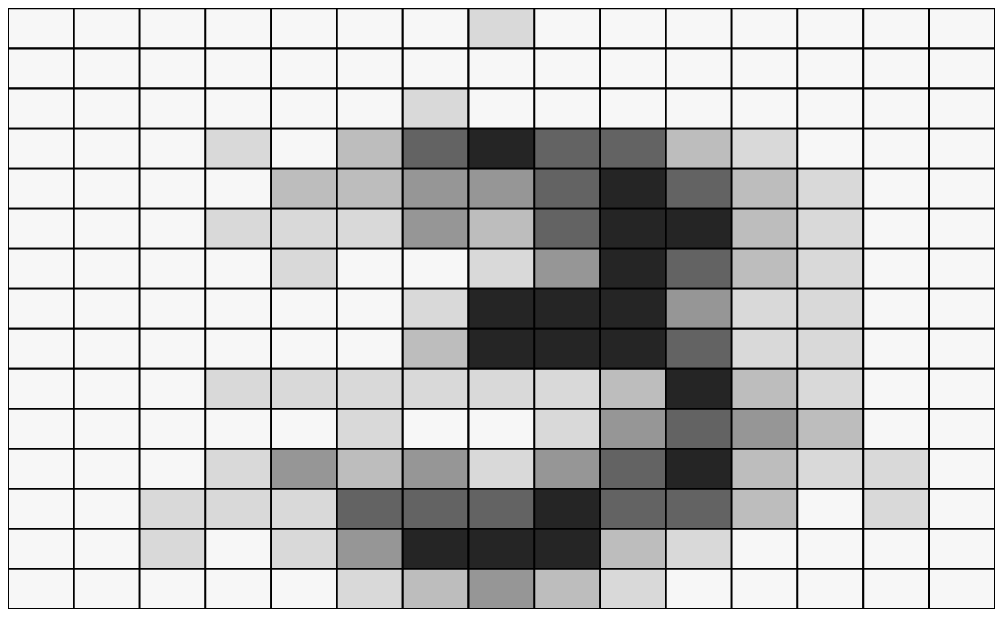}
   \end{subfigure}
  \hspace{-0.1in}
   \begin{subfigure}{0.10\textwidth}
       \includegraphics[width = 1.85cm, height = 1.85cm]{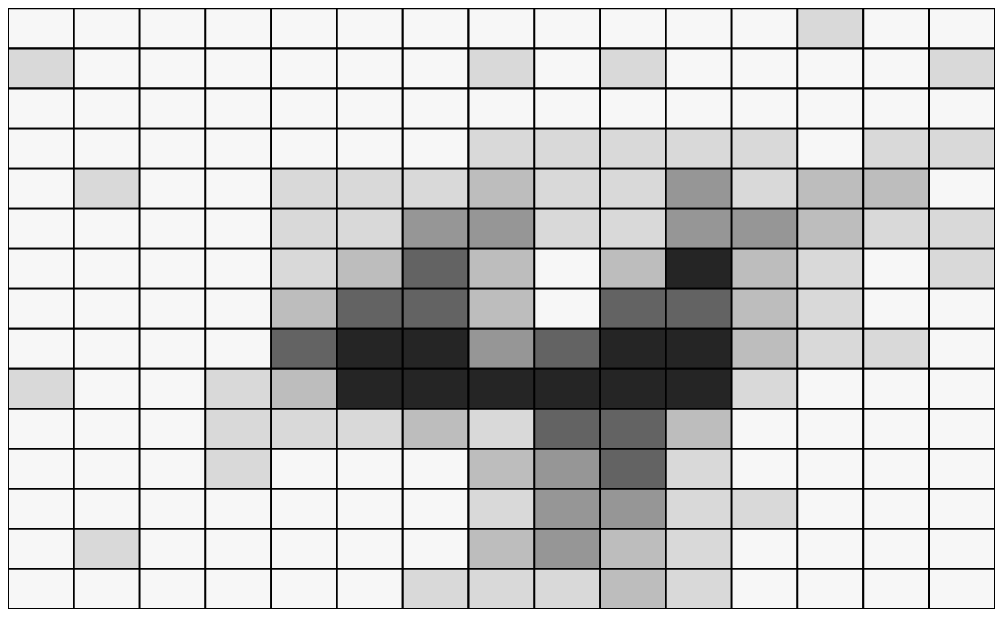}
   \end{subfigure}
   \hspace{-0.1in}
   \begin{subfigure}{0.10\textwidth}
       \includegraphics[width = 1.85cm, height = 1.85cm]{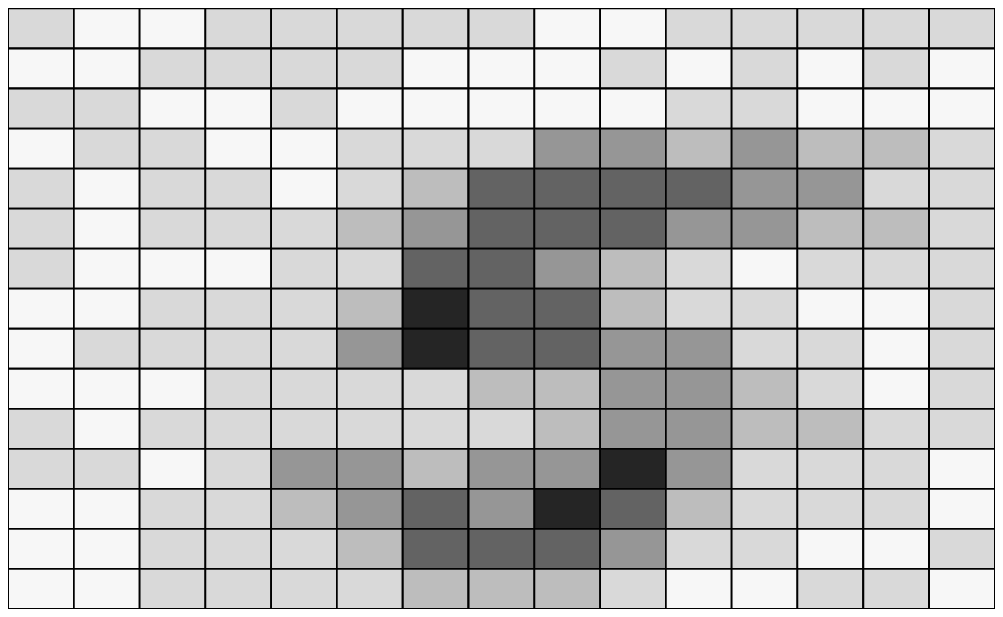}
   \end{subfigure}
   \hspace{-0.1in}
   \begin{subfigure}{0.10\textwidth}
       \includegraphics[width = 1.85cm, height = 1.85cm]{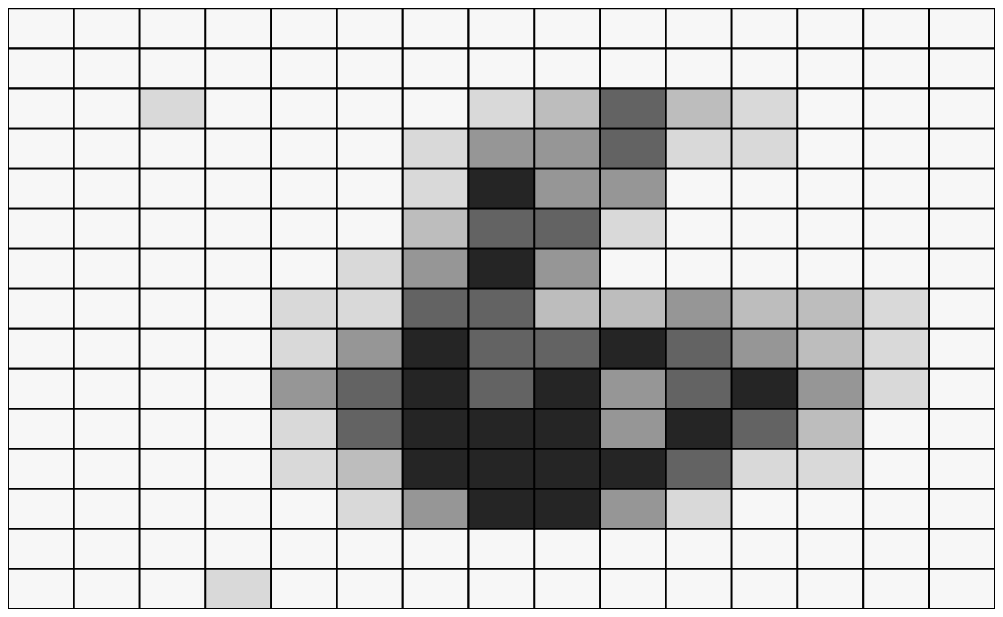}
   \end{subfigure}
   \hspace{-0.1in}
   \begin{subfigure}{0.10\textwidth}
       \includegraphics[width = 1.85cm, height = 1.85cm]{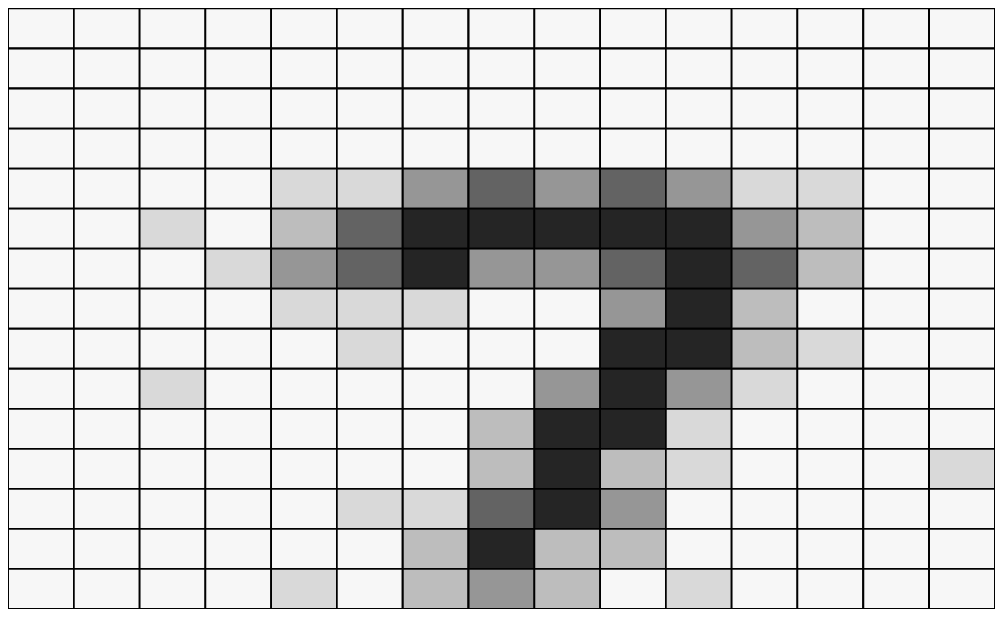}
   \end{subfigure}
   \hspace{-0.1in}
   \begin{subfigure}{0.10\textwidth}
       \includegraphics[width = 1.85cm, height = 1.85cm]{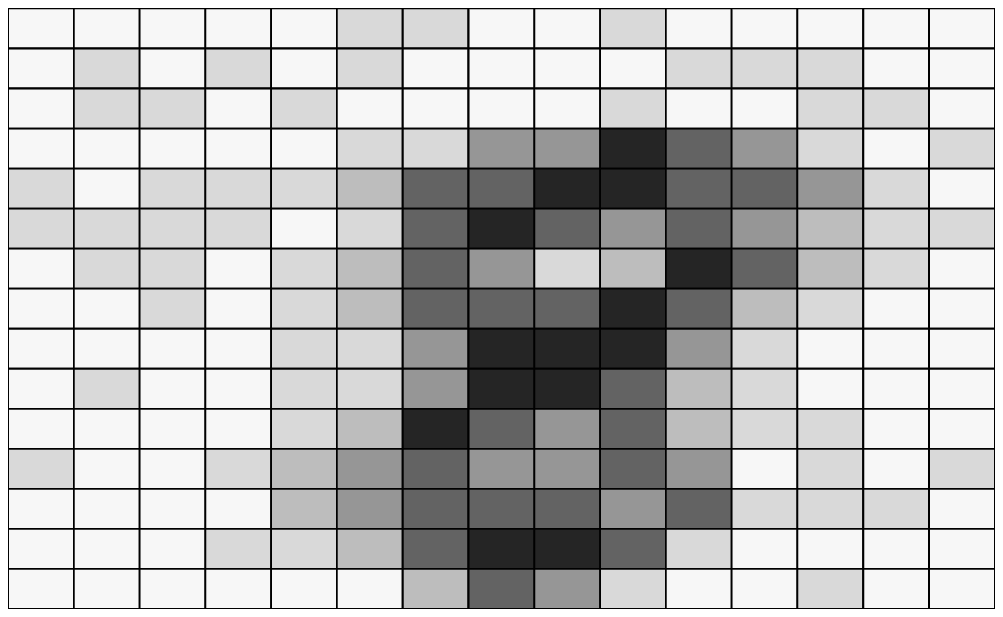}
   \end{subfigure}
   \hspace{-0.1in}
   \begin{subfigure}{0.10\textwidth}
       \includegraphics[width = 1.85cm, height = 1.85cm]{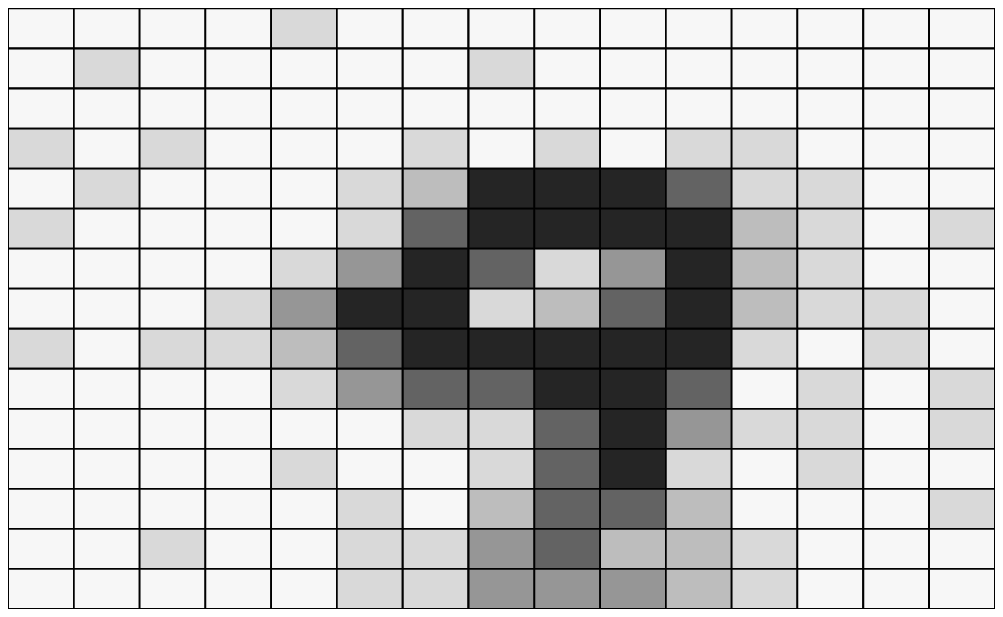}
   \end{subfigure}
   \caption{Observed and reconstructed images of handwritten digits. The first row corresponds to observed of test images of digits $0$--$9$. The second and third rows show the reconstructed images from the trained RBM-MLE and BM-MLE model. All models are fitted with $m = 50$ hidden variables. }
   \label{fig:reconstruction}
\end{figure}


   \begin{table}
\centering
    \scalebox{0.7}{
    \begin{tabular}{c|cccccccccc|c}
    \hline
    Digit & $\mathbf{0}$ & $\mathbf{1}$ & $\mathbf{2}$ & $\mathbf{3}$ & $\mathbf{4}$ & $\mathbf{5}$ & $\mathbf{6}$ & $\mathbf{7}$ & $\mathbf{8}$ & $\mathbf{9}$ & Run Time (s.)\\
    \hline
    RBM-CD & 0.44 & 0.36 & 0.42 & 0.48 & 0.43 & 0.43 & 0.47 & 0.39 & 0.41 & 0.40 & 3.24\\
    RBM-AIS & 0.12 & 0.09 & 0.13 & 0.12 & 0.13 & 0.15 & 0.11 & 0.11 & 0.13 & 0.11 & 63.48  \\
    RBM-MLE & 0.14 & 0.06 & 0.14 & 0.14 & 0.11 & 0.12 & 0.12 & 0.10 & 0.14 & 0.11 & 0.73  \\
    BM-MLE &  0.14 & 0.06 & 0.14 & 0.14 & 0.10 & 0.12 & 0.12 & 0.10 & 0.14 & 0.10  & 139.21  \\
        \hline
    \end{tabular}
    }
    \caption{Brier loss for RBM and BM trained with proposed algorithm, AIS and RBM trained with CD on MNIST data. The results demonstrate the superiority of full likelihood-based training over CD-based training. Moreover, the runtimes show scalability of the proposed approach compared to AIS.}
    \label{tab:MNIST_brier}
\end{table}
}

\section{Conclusions and Future Works}\label{sec:conc}
The overarching focus of the current paper has been on the development of computationally feasible likelihood-based inference procedures in intractable graphical models, and in particular for models that were long considered impermeable to such inference, e.g., Boltzmann machine-based architectures. Our theoretical results on likelihood-based inference complement the methodological contributions. A number of future investigations can be considered. First, we only considered first order gradient-based methods for optimization. However, since we have a Monte Carlo estimate of the gradient, a quasi-Newton approach such as BFGS \citep[see, e.g.,][]{broyden1970convergence} appears straightforward to develop and may have better convergence speed in practice compared to a first order gradient descent scheme. Second, the intractability of the normalizing constant also affects the computation of evidence or marginal likelihood, and makes model selection and comparison difficult. Although not explicitly considered in this work, it is a natural future step to consider model selection problems in exponential family graphical models and to consider multiple connected graphs, where we envisage that the techniques developed in this paper will prove fruitful.




\section*{Acknowledgments}
The authors thank the AE and four anonymous referees for constructive suggestions. Bhadra and Chakraborty were supported by U.S. National Science Foundation Grant SES-2448704.

\section*{Supplementary Material}
The Supplementary Material contains proofs, algorithmic details, and additional results from simulations and data analysis.  Computer code with complete documentation is publicly available on \texttt{github} at: \href{https://github.com/chenyujie1104/ExponentialGM}{https://github.com/chenyujie1104/ExponentialGM}.

\bibliography{references.bib}

\clearpage
\begin{center}
	{\large{\bf Supplementary Material to\\
	{\it Likelihood-based Inference in Fully and Partially Observed Exponential Family Graphical Models with Intractable Normalizing Constants}}
	}
\end{center}
\setcounter{equation}{0}
\setcounter{table}{0}
\setcounter{section}{0}
\setcounter{subsection}{0}
\setcounter{figure}{0}
\renewcommand{\theequation}{S.\arabic{equation}}
\renewcommand{\thesection}{S.\arabic{section}}
\renewcommand{\thealgorithm}{S.\arabic{algorithm}}
\renewcommand{\thetable}{S.\arabic{table}}
\renewcommand{\thefigure}{S.\arabic{figure}}

\begingroup
  \hypersetup{linktoc=none}
\parttoc 
\endgroup

\section{Proofs}
\subsection{Proof of Proposition \ref{prop:exp}}\label{prf:exp}
(a). (Proof for $z_\theta$). For brevity, define $T_j = T_j(X_j)$ and $T_{jk} = T_{jk}(X_j, X_k)$. Recall $\phi=\mathrm{diag}(\theta)$. Thus, from Equation \eqref{eq:pairwise_graphical_model}, $q_\theta (X)/q_\phi (X) = \exp \left\lbrace \underset{(j,k) \in E}{\sum}\theta_{jk}T_{jk} \right\rbrace$.
    However, the model of Equation~\eqref{eq:pairwise_graphical_model} is of the form $p_{\eta}(X) = \exp \{ \eta T(X) + H(X) - A(\eta)\}$
    with $T(X) = (T_j, T_{jk})$ as the sufficient statistic. That is, $T(X)$ is a $p\times p$ symmetric matrix with $T_j$ as the $j$th diagonal entry and $T_{jk}$ as the $(j,k)$th off-diagonal entry.  The moment generating function of exponential family distribution satisfies $M_{T,\eta}(u) =  \mathbb{E}_\eta(\exp(u^T T(X))  = \exp(A(u + \eta) - A(\eta)).$
Setting $u =\theta - \mathrm{diag}(\theta)$, and $\phi=\mathrm{diag}(\theta)$, {one has}, $u + \phi = \theta \in \Theta$. Similarly, $\phi \in \Theta$, because this corresponds to the special case where the $p$ variables are independent. Using the fact that for any valid parameter $\gamma \in \Theta$ we have $A(\gamma)< \infty$ in an exponential family model, shows that: 
\begin{align*}
\mathbb{E}_\phi (q_\theta(X)/q_\phi(X)) &= M_{T,\phi}(u) = \exp(A(u+\phi) - A(\phi)) = \exp(A(\theta) - A(\phi)) < \infty.
\end{align*}
Similarly, if, $2u + \phi \in \Theta$, then, $\mathbb{E}_\phi\{(q_\theta(X)/q_\phi(X))^2\} = M_{T,\phi}(2u) = \exp(A(2u+\phi) - A(\phi)) < \infty.$
Thus, $
\mathrm{Var}_\phi (q_\theta(X)/q_\phi(X)) < \infty.$ 

(b). (Proof for $\nabla_\theta z(\theta)$). Next consider $\nabla_\theta z(\theta)$. For $(j,k)\in E$, $\frac{\nabla_{\theta_{jk}} q_\theta (X)}{q_\phi (X)} = \frac{ q_\theta (X)}{q_\phi (X)} T_{jk}$.
Thus, 
$$\mathbb{E}_\phi\left\{\left(\frac{\nabla_{\theta_{jk}} q_\theta (X)}{q_\phi (X)}\right)^2\right\} =\mathbb{E}_\phi\left\{\left(\frac{ q_\theta (X)}{q_\phi (X)}\right)^2 T_{jk}^2\right\} \le \left\{\mathbb{E}_\phi\left(\frac{ q^{2a}_\theta (X)}{q^{2a}_\phi (X)} \right)\right\}^{1/a}\left[\mathbb{E}_\phi\left(T_{jk}^{2b}\right)\right]^{1/b} <\infty,$$ 
where $a>1,\; b>1,\; (1/a+1/b)=1$ and the claim in the previous display is a consequence of the H\"older inequality. Similar to Part (a), and taking $a=1+\delta$, we have: 
$
\mathbb{E}_\phi\left(\frac{ q^{2a}_\theta (X)}{q^{2a}_\phi (X)} \right) =M_{T,\phi}(2au)=M_{T,\phi}(2(1+\delta)u),
$
which is bounded if $(2(1+\delta)u+\phi)\in\Theta$, for some  $\delta>0$. The boundedness of ${E}_\phi\left(T_{jk}^{2b}\right)$ follows from the fact that for an exponential family distribution, the sufficient statistic also belongs to an exponential family, for which moments of all orders exist \citep[Theorem 8.1,][]{barndorff2014information}. Repeating the calculation for all $(j,k) \in E$, the result follows. If $(j,k)\notin E$, the variance is zero and the result is trivial. 
\subsection{Proof of Proposition \ref{lemma:sub-gaussian}}\label{sec:lemma_proof}
Since $\phi = \text{diag}(\theta)$ and $X$ is distributed according to a PEGM, we have $U = q_{\theta}(X)/q_{\phi}(X) = \exp\{\sum_{(j,k) \in E} \theta_{jk}X_jX_k\}$. Next, recall that a function $g:\mathbb{R}^p \to \mathbb{R}$ is said to have the bounded difference property if for some non-negative constants $c_1, \ldots, c_p$,
\begin{align}\label{eq:bdd_difference}
    \sup_{x_1,\ldots, x_p, x_{j'} \in \mathbb{R}^p} |g(x_1, \ldots, x_j, \ldots, x_p) - g(x_1, \ldots, x_{j'}, \ldots, x_p)| \leq c_j.
\end{align}
In the current context, let $U = U(X_1,\ldots, X_j, \ldots, X_p)$ and $U' = U(X_1, \ldots, X_{j'}, \ldots, X_p)$. By assumption, each $X_j$ has support in $[a_j, b_j]$.  Without loss of generality, let $|b_j| \geq |a_j|$ for all $j = 1, \ldots, p$. Clearly, $U \leq \exp\{\sum_{(j,k) \in E}\theta_{jk} |X_j X_k|\} \leq \exp\{\sum_{(j,k) \in E}\theta_{jk} |b_j b_k|\}$. Similarly, $U'$ is also bounded. Then by the triangle inequality, $|U - U'| = |U| + |U'|  \leq 2\exp\{\sum_{(j,k) \in E}\theta_{jk} |b_j b_k|\}$. Thus, $U$ is of bounded difference according to the definition \eqref{eq:bdd_difference}. The second assertion then follows from \citet[][Theorem 12]{boucheron2003concentration}. Noting that $V = X_jX_k U$, the same argument for $V$ then applies.

\subsection{Proof of Proposition \ref{prop:importance_sample_size}}\label{sec:importance_sample_size_proof}
The result is a direct consequence of Theorem 3.1 of \cite{chatterjee2018sample}. In particular, when $|S_\theta| = s = O(\log p)$, it is straightforward to observe that $D(p_\theta, p_\phi) \sim s$. 

\subsection{Auxiliary results for Theorems~\ref{thm:l1_consistency} and~\ref{thm:posterior_consistency}}
In this subsection, we record a few auxiliary results which will be used to prove Theorems in Section \ref{sec:theory} of the main draft. The first result is an exponential concentration result, which is used to prove Theorem \ref{thm:l1_consistency}. The second result is needed for the proof of Theorem~\ref{thm:posterior_consistency}, and characterizes the prior Kullback-Leibler support set in terms of the Euclidean distance between the true parameter and prior realizations.

\begin{proposition}\label{prop:tail_bound}
     Suppose we observe $X_1, \ldots, X_n \overset{iid}{\sim} \text{PEGM}(\theta_0)$ where the model satisfies Assumption 2.2 and Assumption 2.3. Let $\mu_{0} = (\mu_{0,jk})$ be the matrix of moments, i.e. $\mu_{0,jk} = \mathbb{E}_{\theta_0}[T(X_j, X_k)]$ when $j \neq k$, and $\mu_{0, jj} = \mathbb{E}_{\theta_0}[T(X_j)] $. Define $S_n = \sum_{i=1}^n T(X_{ij}, X_{ik})$. Then for some $c \geq 0$
     $$\mathbb{P}[|S_n - n\mu_{0, jk}| > n\nu] \leq 2 \exp(-cn \nu^2),$$
     for $\nu \geq 4\kappa_3/3$.
 \end{proposition}
 \begin{proof}
     Define $S_n = \sum_{i=1}^n T(X_{ij}, X_{ik})$. We have for any $s\geq 0$,
 \begin{align*}
     \mathbb{P}(S_n - n \mu_{0,jk} > n\nu )  &\leq \mathbb{E}[e^{s(S_n - n\mu_{0,jk})}]e^{-sn\nu}\\
     & = \prod_{i=1}^n \mathbb{E}[e^{sT(X_{ij}, X_{ik})}]e^{-ns(\mu_{0,jk}+ \nu)}.
 \end{align*}
 Write $T_j = T(X_j)$ and $T_{jk} = T(X_j, X_k)$. Then,
 \begin{align*}
 \log \mathbb{E}[\exp(\alpha T_{jk})] &=  \log \int_{\mathcal{X}} \exp\left\lbrace \alpha T_{jk}  +\sum_{j=1}^p  \theta_{0j} T_j + \sum_{\substack{j,k=1 \\j \neq k}}^{p} \theta_{0,jk}T_{jk} + \sum_{j=1}^p C(X_j) - A(\theta_0) \right\rbrace \mathrm{d}x\\
 & = \bar{A}_{jk}(\theta_0;\alpha) - \bar{A}_{jk}(\theta_0;0). 
 \end{align*}
 Letting $\alpha \leq 1$ and using a Taylor expansion we get for some $\gamma \in (0,1)$,
 \begin{align*}
     \bar{A}_{jk}(\theta_0;\alpha) - \bar{A}_{jk}(\theta_0;0) &= \alpha \dfrac{\partial}{\partial \eta} \bar{A}_{jk}(\theta_0;\eta)\rvert_{\eta = 0} + \frac{\alpha^2}{2} \dfrac{\partial^2}{\partial \eta^2} \bar{A}_{jk}(\theta_0;\eta)\rvert_{\eta = \gamma\alpha}\\
     & \leq \alpha \kappa_3 + \frac{\alpha^2}{2} \kappa_5.
 \end{align*}
 Thus, $\mathbb{E}[\exp(s T_{jk})]\leq \exp(s \kappa_3 + \frac{s^2}{2} \kappa_5)$ for $s\leq 1$. Hence, for $0\leq s \leq 1$, and then minimizing over $s$ and since $\mu_{0,jk} \geq -\kappa_3$ we get,
 \begin{align*}
     \mathbb{P}(S_n - n \mu_{0,jk} > n\nu )& \leq e^{-ns(\nu + \mu_{0,jk}) +  ns \kappa_3 + ns^2 \frac{\kappa_5}{2}} \leq e^{-n\frac{(\nu + \mu_{0,jk} - \kappa_3)^2}{2\kappa_5}}, \quad s = (\nu + \mu_{0,jk} - \kappa_3)/\kappa_5\\
     & \leq e^{-n \frac{\nu^2}{8 \kappa_5}}, 
 \end{align*}
 when $\nu \geq 4\kappa_3/3$.  Applying the same result on $-(S_n - n \mu_{0,jk})$ and using the fact that $\mathbb{P}(|X|>t) = \mathbb{P}(X>t) + \mathbb{P }(-X>t)$, we obtain the result.
 A similar analysis leads to $\mathbb{P}(|S_n - \mu_{0,jj}|>n\nu) \leq 2 \exp(-cn\nu^2)$ for $0\leq \nu \leq \kappa_2/2$.
 \end{proof}

For notational convenience, we write $\text{vech}(\theta) = \theta \in \mathbb{R}^{p(p+1)/2}$. Let $\mu \in \mathbb{R}^{p\times p}$ denote the corresponding symmetric matrix of mean parameters, i.e. $\mu_{jk} = \mathbb{E}[T(X_j, X_k)]$ for $j \neq k$ and $\mu_{jj} = \mathbb{E}[T(X_j)]$ for $j = 1, \ldots, p$. Similarly, $\mu_0$ is the mean matrix when the parameter is $\theta_0$. Again, we write $\text{vech}(\mu) = \mu$.
\begin{proposition}\label{prop:KL_set}
   Suppose $\pi$ is some distribution on $\Theta$.  Assume that $\lambda_{\max}\left(\dfrac{\partial^2 A(\theta_0)}{\partial \theta \partial \theta^\T}\right)\leq \tilde{\beta}$ for some $\tilde{\beta}>0$. Fix $\epsilon>0$. Define the set $B(p_{\theta_0}, \epsilon ) = \{p_\theta: KL(p_{\theta_0},p_\theta) < \epsilon, V(p_{\theta_0}, p_\theta)< \epsilon\}$, where $V(p_{\theta_0}, p _\theta) = \text{Var}_{p_{\theta_0}}\left[ \log p_{\theta_0}/p_\theta\right]$. Then for every $\epsilon >0$ the exists $\delta = \delta(\epsilon)>0$ such that,
   $$\pi\left\lbrace B(p_\theta, \epsilon)\right\rbrace = \pi\left\lbrace \theta \in \Theta: \norm{\theta - \theta_0}_2 < \delta \right\rbrace.$$
\end{proposition}
\begin{proof}
    It is straightforward to establish that $KL(p_{\theta_0}, p_{\theta}) = \langle \theta_0- \theta, \mu_0\rangle + \log z(\theta) - \log z(\theta_0)$. Indeed, letting $T \in \mathbb{R}^{p(p+1)/2}$ to be the vector of sufficient statistics, we have
    \begin{align*}
    KL(p_{\theta_0}, p_\theta) =& \log z(\theta) - \log z(\theta_0) +  \int \langle \theta_0 - \theta, T \rangle p_{\theta_0}(x) dx\\
    =& \langle \theta_0 - \theta, \mu_0\rangle + \log z(\theta) - \log z(\theta_0), 
    \end{align*}
    where $\mu_{0}  = \mathbb{E}_{\theta_0}[ T]$.
    Furthermore, 
    \begin{align*}
    V(p_{\theta_0}, p_\theta) &= \int \langle \theta_0 - \theta, T - \mu_0 \rangle^2 p_{\theta_0}(x) dx\\
    & = \mathbb{E}_{\theta_0}\left[ \langle \theta_0 - \theta, T - \mu_0 \rangle^2 \right] \\
    & = (\theta_0 - \theta)^\T  \dfrac{\partial^2\log z(\theta_0)}{\partial \theta \partial \theta^\T} (\theta_0 - \theta).
\end{align*}
    Write $A(\theta) = \log z(\theta)$ and $A(\theta_0) = \log z(\theta_0)$. It is well known that $A(\theta)$ is convex and has derivatives of all orders \citep[Proposition 3.1]{wainwright2008graphical}.  Also, by assumption $\lambda_{\max}\left(\dfrac{\partial^2 A(\theta_0)}{\partial \theta \partial \theta^\T}\right)\leq \tilde{\beta}$ which implies there exists $L>0$ such that $\norm{\nabla A(\theta_1) - \nabla A(\theta_2)}_2\leq L \norm{\theta_1 - \theta_2}_2$ for all $\theta_1, \theta_2 \in \Theta$.
    Thus, standard convex analysis \citep{boyd2004convex} implies that: 
    \begin{align*}
        \lvert A(\theta_1) - A(\theta_2) - \nabla A(\theta_2)^\T (\theta_1 - \theta_2) \rvert \leq \frac{L}{2} \norm{\theta_1 - \theta_2}_2^2.
    \end{align*}
   Since, $KL(p_{\theta_0}, p_{\theta}) = \langle \theta_0- \theta, \mu_0\rangle + A(\theta) - A(\theta_0)$ and $\lambda_{\min}\left(\dfrac{\partial^2 A(\theta_0)}{\partial \theta \partial \theta^\T}\right)\geq \beta,$ we have that: 
    \begin{align*}
        \frac{\beta}{2} \norm{\theta_1 - \theta_2}_2^2 \leq KL(p_{\theta_0}, p_{\theta}) \leq \frac{L}{2} \norm{\theta_1 - \theta_2}_2^2,
    \end{align*}
which proves the result. A similar argument applies to $V(p_{\theta_0}, p_\theta)$.

\end{proof}


\subsection{Proof of Theorem \ref{thm:l1_consistency}}\label{sec:proof_l1_consistency}

Define the matrix $S = (S_{jk})$ where $S_{jk} = \sum_{i = 1}^n T(X_{ij},X_{ik})$ if $j \neq k$ and $S_{jj}= \sum_{i=1}^n T(X_{ij})$. Then $\ell(\theta) = \sum_{i=1}^n q_{\theta}(X_i) - n A(\theta) $ $ = \tr(\theta S) - n A(\theta)$. The $\ell_1$-penalized estimator is defined as the minimizer of $-\ell(\theta) + \lambda \norm{\theta}_1 = nA(\theta) - \tr(\theta S) + \lambda \norm{\theta}_1$. Let
\begin{align*}
    Q(\theta) &= nA(\theta) - \tr(\theta S) + \lambda \norm{\theta}_1 - nA(\theta_0) + \tr(\theta_0 S) - \lambda\norm{\theta_0}_1 \\
    &= n [A(\theta) - A(\theta_0)]   -\tr[(\theta - \theta_0)(S-\mu_0)] - \tr[(\theta - \theta_0)\mu_0] + \lambda [\norm{\theta}_1 - \norm{\theta_0}_1].
\end{align*}
Then, the $\ell_1$ penalized estimate minimizes $Q(\theta)$. Define $\Delta = \theta - \theta_0$ and $G(\Delta) = Q(\theta_0 + \Delta)$. Equivalently, we want to minimize $G(\Delta)$ with respect to $\Delta$. First, note that for any $\alpha \in (0,1)$ and $\theta_1, \theta_2 \in \Theta$,
\begin{align*}
\small
    \alpha Q(\theta_1) =& n\alpha [A(\theta_1) - A(\theta_0)] - \tr[\alpha(\theta_1 - \theta_0)(S- \mu_0)] - \tr[\alpha(\theta_1 - \theta_0)\mu_0] \\ 
    & +\lambda \alpha [\norm{\theta_1}_1 - \norm{\theta_0}_1], \\
    (1-\alpha) Q(\theta_2) =& n(1-\alpha) [A(\theta_2) - A(\theta_0)] - \tr[(1 - \alpha)(\theta_2 - \theta_0)(S - \mu_0)] - \tr[(1-\alpha)(\theta_2 - \theta_0)\mu_0]\\
    &+ \lambda (1-\alpha) [\norm{\theta_1}_1 - \norm{\theta_0}_1].
\end{align*}
This implies that if $\theta^\star = \alpha \theta_1 + (1-\alpha)\theta_2$ then,
\begin{align*}
    \alpha Q(\theta_1) + (1-\alpha) Q(\theta_2) &\geq n[A(\theta^\star) - A (\theta)] - \tr((\theta^\star - \theta_0) (S - \mu_0)) -\tr((\theta^\star - \theta_0) \mu_0) \\
    & + \lambda [\norm{\theta^\star}_1 - \norm{\theta_0}_1]\\
    &= Q(\theta^\star),
\end{align*}
by triangle inequality and since $A(\theta)$ is convex. Similarly, if $\Delta_1 = \theta_1 - \theta_0$ and $\Delta_2 = \theta_2 - \theta_0$, then $\alpha G(\Delta_1) + (1-\alpha)G(\Delta_2) \geq G(\Delta^\star)$ where $\Delta^\star = \theta^\star - \theta_0$. Thus, $G(\Delta)$ is convex and $G(0) = 0$. Hence, if $\hat{\Delta} = \inf_{\Delta} G(\Delta)$, then $G(\hat{\Delta}) \leq 0$.

Now, $G(\Delta) = n[A(\theta_0 + \Delta) - A(\theta_0)] - \tr[\Delta(S - \mu_0)] - \tr(\Delta \mu_0) + \lambda [\norm{\theta_0 + \Delta}_1 - \norm{\theta_0}_1]$. A Taylor expansion of $A(\theta_0 + \Delta) - A(\theta_0)$ yields 
\begin{align*}
   A(\theta_0 + \Delta) - A(\theta_0) & = \, \tr(\nabla A(\theta_0)\Delta)  \\
   & + \Tilde{\Delta}^\T\left[ \int_{0}^1 (1-v)\left\lbrace \frac{\partial^2}{\partial v^2} A(\theta_0 + v \Delta) \right\rbrace  dv \right] \Tilde{\Delta},
\end{align*}
where $\Tilde{\Delta} = \text{vec}(\Delta) \in \mathbb{R}^{p^2}$. Here, we adopt the convention that for a matrix-valued function (with matrix arguments) $f: \mathbb{R}^{d_1 \times d_2} \to \mathbb{R}^{d_3 \times d_4}$, $df(A)$ is a matrix of order $d_1d_3 \times d_2d_4$ with elements
$$df(A) = \begin{pmatrix}
    \frac{df}{A_{11}} & \ldots & \frac{df}{A_{1, d_2}}\\
    \vdots & \vdots & \vdots \\
     \frac{df}{A_{d_1, 1}} & \ldots & \frac{df}{A_{d_1, d_2}}
     \end{pmatrix}.
$$
In particular, since $\nabla A(\theta) \in \mathbb{R}^{p \times p}$, we have that $\frac{\partial^2 A(\theta)}{\partial \theta \partial \theta^\T } \in \mathbb{R}^{p^2 \times p^2}$.

Thus $G(\Delta) = n\Tilde{\Delta}^\T\left[ \int_{0}^1 (1-v)\frac{\partial^2 }{\partial v^2}A(\theta_0 + v \Delta)dv \right] \Tilde{\Delta} - \tr[\Delta(S-\mu_0)] + \lambda [\norm{\theta_0 + \Delta}_1 - \norm{\theta_0}_1]$. That is, $G(\Delta) = T_1 - T_2 + T_3$. Let $\mathbf{S} = \{j\leq k: \theta_{0,jk} \neq 0, j, k = 1, \ldots, p\}$.
We next work with each term separately. We have 
\begin{align*}
    |T_2| = |\tr[\Delta(S-\mu_0)]|& \leq \left| \sum\sum_{j\leq k} (S_{jk} - \mu_{0,jk})\Delta_{jk}\right|  \\
    & \leq \left| \sum_{j \leq k:}\sum_{\theta_{0, jk} = 0} (S_{jk} - \mu_{0,jk})\Delta_{jk}\right| + \left| \sum_{j \leq k:}\sum_{\theta_{0, jk} \neq 0} (S_{jk} - \mu_{0,jk})\Delta_{jk}\right|\\
    & \leq \max_{\mathbf{S}^c} |S_{jk} - \mu_{0,jk}|\norm{\Delta_{\mathbf{S}^c}}_1 + \sqrt{s} \max_{\mathbf{S}}|S_{jk} - \mu_{0,jk}|\norm{\Delta_\mathbf{S}}_F,
\end{align*}
where the second bound comes from Cauchy-Schwarz inequality. 
From Proposition \ref{prop:tail_bound}, and the union bound we have that with probability at least $1 - 2e^{-C_1\log p}$, $\max_{j\leq k} |S_{jk} - \mu_{0, jk}|\leq C_1 \sqrt{\log p/n}$ for some $C, C_1>0$. Thus we have with probability tending to 1,
\begin{align*}
    |T_2| \leq C_1\sqrt{\frac{\log p}{n}} \norm{\Delta_{\mathbf{S}^c}}_1 + C_1 \sqrt{\frac{s\log p}{n}} \norm{\Delta_{\mathbf{S}}}_F.
\end{align*}
Next, we work with $T_3 = \lambda [\norm{\theta_0 + \Delta}_1 - \norm{\theta_0}_1]$. Clearly, $\norm{\theta_0}_1 = \norm{\theta_{0\mathbf{S}}}_1$ and $\norm{\theta_0 + \Delta}_1 = \norm{\theta_{0\mathbf{S}} + \Delta_S}_1 + \norm{\Delta_{\mathbf{S}^c}}_1$. Thus, by triangle inequality 
$$T_3 \geq \lambda (\norm{\Delta_{\mathbf{S}^c}}_1 - \norm{\Delta_\mathbf{S}}_1).$$
Set $\lambda = \frac{C_1}{\epsilon} \sqrt{\frac{\log p}{n}}$ for some $0< \epsilon < 1 $. Then since $\norm{\Delta_\mathbf{S}}_1 \leq \sqrt{s}\norm{\Delta_\mathbf{S}}_F$, we have,
\begin{align*}
    T_3 &\geq \dfrac{C_1}{\epsilon} \sqrt{\dfrac{\log p}{n}}\norm{\Delta_{\mathbf{S}^c}}_1 - \dfrac{C_1}{\epsilon}\sqrt{\dfrac{s\log p}{n}}\norm{\Delta_{\mathbf{S}}}_F.
\end{align*}
Finally, for $T_1$ we have for $\Delta \in \Theta_n(M) = \{\Delta: \Delta = \Delta^\T , \norm{\Delta}_F = Mr_n\}$, 
\begin{align*}
   T_1 &= n \Tilde{\Delta}^\T\left[ \int_{0}^1 (1-v)\frac{\partial^2}{\partial v^2} A(\theta_0 + v \Delta)dv \right] \Tilde{\Delta} \\
   &\geq \lambda_{\min}\left(\int_{0}^1 (1-v)\frac{\partial^2}{\partial v^2} A(\theta_0 + v \Delta)dv\right) \norm{\Delta}_F^2 \\
   & \geq \norm{\Delta}_F^2\int_0^1(1- v)\lambda_{\min}\left(\frac{\partial^2}{\partial v^2} A(\theta_0 + v \Delta) \right)dv \\
   &\geq \norm{\Delta}_F^2 \min_{0\leq v \leq 1} \lambda_{\min}\left(\frac{\partial^2}{\partial \theta \partial \theta^\T} A(\theta_0 + v \Delta) \right) \int_0^1 (1- v)dv \\
    &\geq \frac{1}{2} \norm{\Delta}_F^2 \min \left[\lambda_{\min}\left(\frac{\partial^2}{\partial \theta \partial \theta^\T} A(\theta_0 + \bar{\Delta}) \right): \norm{\bar{\Delta}}_F \leq Mr_n\right].
\end{align*}
In order to work with $\lambda_{\min}\left(\frac{\partial^2}{\partial \theta \partial \theta^\T} A(\theta_0 + \bar{\Delta}) \right)$, we next consider a Taylor expansion of $\frac{\partial^2}{\partial \theta \partial \theta^\T} A(\theta)$. For a matrix $P$, $\norm{P}_2 = \sup_{\norm{x}_2 = 1} \norm{Px}_2$. We  note here that $\frac{\partial^2 A(\theta)}{\partial  \theta \partial \theta^\T}: \Xi  \to \mathbb{R}^{p^2 \times p^2}$, where $\Xi = \Theta \times \Theta$. Elements of $\Xi$ are of the form $\theta \otimes \theta$. Let $\xi_0 = \theta_0 \otimes \theta_0$. Then, a $\bar{\Delta}$ perturbation in $\theta_0$ leads to the following perturbation in $\xi_0$: $ \xi_{\bar{\Delta}} = (\theta_0 + \bar{\Delta}) \otimes (\theta_0 + \bar{\Delta}) = \xi_0 + \theta_0 \otimes \bar{\Delta} + \bar{\Delta} \otimes \theta_0 + \bar{\Delta} \otimes \bar{\Delta}  = \xi_0 + \underline{\Delta}$, say. Moreover, $\norm{\xi_{\bar{\Delta}} - \xi_0}_2 = \norm{\underline{\Delta}}_2 \leq 2 \norm{\theta_0}_2 \norm{\bar{\Delta}}_2 + \norm{\bar{\Delta}}_2^2$, where we have used the fact that $\norm{A \otimes B}_2 = \norm{A}_2 \norm{B}_2$. Similarly, $\norm{\xi_{\bar{\Delta}} - \xi_0}_F = \norm{\underline{\Delta}}_F \leq 2 \norm{\theta_0}_2 \norm{\bar{\Delta}}_F + \norm{\bar{\Delta}}_F^2$, since $\norm{AB}_F \leq \norm{A}_2 \norm{B}_F$.
Hence, $\norm{\xi_{\bar{\Delta}} - \xi_0}_2 = O(r_n )$ and $\norm{\xi_{\bar{\Delta}} - \xi_0}_F = O(r_n )$  since by assumption $\norm{\theta_0}_2 < \bar{\beta}$. 
 Then, a Taylor expansion of $G(\xi_0 + \underline{\Delta})$ around $\xi_0$  \citep[Equation (24)]{vetter1973matrix} yields:
     \begin{align}
       G(\xi_0 + \underline{\Delta} ) &= G(\xi_0) + D_{\xi_0^\T}[\mathrm{tr}(\underline{\Delta}G(\xi) ) ]+ R_{\xi_0},
     \end{align}
     where $R_{\xi_0}^{p^2 \times p^2}$ is a matrix such that $\norm{R_{\xi_0}}_F = o(\norm{\bar{\Delta}}_F)$, which follows from the property of Taylor's theorem and analytical nature of $A(\theta)$ for all $\theta \in \Theta$. Thus, we have that $G(\xi_0 + \underline{\Delta}) = G(\xi_0) + H$, where $H = D_{\xi_0^\T}[\mathrm{tr}(\underline{\Delta}G(\xi) ) ] + R_{\xi_0}$, and $\norm{H}_2 \leq \norm{D_{\xi_0^\T}[\mathrm{tr}(\underline{\Delta}G(\xi) ) ]}_2 + \norm{R_{\xi_0}}_2 \leq C \norm{\bar{\Delta}}_2 + \norm{R_{\xi_0}}_2$ for some positive constant $C>0$ by the regularity assumption. Hence, $\norm{H}_2 = O(r_n) $. Due to Weyl's inequality, we then have:
     \begin{align*} 
     &\left| \lambda_{\min} \left( \frac{\partial^2}{\partial \theta \partial \theta^\T } A(\theta_0 + \bar{\Delta}) \right) - \lambda_{\min} \left( \frac{\partial^2}{\partial \theta \partial \theta^\T} A(\theta_0) \right) \right| \\ &\leq \norm{\frac{\partial^2}{\partial \theta \partial \theta^\T} A(\theta_0 + \bar{\Delta}) -\frac{\partial^2}{\partial \theta \partial \theta^\T} A(\theta_0 )}_2\\
     & \leq Cr_n,
     \end{align*}
     for some $C>0$.
Since $\norm{\bar{\Delta}}_F \leq Mr_n$ and $Mr_n \to 0$ as $n \to \infty$ we have that,
$$\lambda_{\min}\left(\frac{\partial^2}{\partial \theta \partial \theta^\T} A(\theta_0  + \bar{\Delta}) : \norm{\bar{\Delta}}_F \leq Mr_n\right)\geq \frac{\beta}{2}.$$ 
Putting all the pieces together, we obtain,

\begin{align*}
    G(\Delta) &\geq \frac{n\beta}{4} \norm{\Delta}_F^2 + C_1 \sqrt{\dfrac{\log p}{n}}\norm{\Delta_{S^c}}_1 \left(\frac{1}{\epsilon} - 1\right) - C_1\sqrt{\dfrac{s\log p}{n}}\norm{\Delta_{S}}_F  \left(1 + \frac{1}{\epsilon}\right)\\
    & \geq  \frac{n\beta}{4} \norm{\Delta}_F^2 - C_1\sqrt{\dfrac{s\log p}{n}}\norm{\Delta_{S}}_F  \left(1 + \frac{1}{\epsilon}\right) > 0,
\end{align*}
since $0<\epsilon <1$, and the second term in the last line of the above display is $O(s\log p/n)$.

Hence, we have argued that if $\hat{\Delta}$ is the minimizer of $G(\Delta)$ then $G(\hat{\Delta})\leq 0$ and $\inf \{G(\Delta): \Delta \in \Theta_n(M)\}>0$, then $\hat{\Delta}$ must be inside the ball $\Theta_n(M)$, i.e. $\norm{\hat{\Delta}}_F \leq Mr_n$.

\subsection{Proof of Theorem \ref{thm:posterior_consistency}}\label{app:consistency_proof}
To establish the posterior concentration rate result, we follow the standard route of sufficient conditions established by \citet[][Chapter 8]{ghosal2017fundamentals}. Define $B(p_{\theta_0}, \epsilon ) = \{p_\theta: KL(p_{\theta_0},p_\theta) < \epsilon, V(p_{\theta_0}, p_\theta)< \epsilon\}$ and $\mathcal{P}= \{p_{\theta}: \theta \in \Theta\}$. Equip $\mathcal{P}$ with the Hellinger metric which we write as $H$ (the choice of the metric can is flexible). To be precise, define the squared Hellinger distance between two probability distributions having densities $p_1$ and $p_2$ with respect to the Lebesgue measure as: $H^2(p_1, p_2) = \frac{1}{2} \int (\sqrt{p_1(x)} - \sqrt{p_2(x)})^2 dx $. Now suppose $p_1(x) = q_{\theta_1}(x)/z(\theta_1)$ and $p_2(x) = q_{\theta_2}(x)/z(\theta_2)$ where $\theta_1, \theta_2 \in \Theta$. Then, straightforward calculations yield:
$$ H^2(p_1, p_2) = 1 - \dfrac{z\left( \frac{\theta_1 + \theta_2}{2}\right)}{\sqrt{z(\theta_1) z(\theta_2)}}.$$
Furthermore, by Pinsker's inequality:
$$H^2(p_{\theta_0}, p_{\theta}) \leq \frac{1}{\sqrt{2}} \sqrt{KL(p_{\theta_0}, p_\theta)} \leq C \norm{\theta_0 - \theta}_2, $$
for some $C>0$, in the light of Proposition \ref{prop:KL_set}. The conditions for the posterior distribution $\pi(\theta \mid \mathbf{X})$ to contract at the rate $\epsilon_n$ are the following:
\begin{enumerate}
    \item Sufficient prior mass: $ \pi\{B(p_{\theta_0, \epsilon_n})\} \geq e^{-Cn \epsilon_n^2}$.
    \item Metric entropy: there exists a partition of $\mathcal{P} = \mathcal{P}_1 \cup \mathcal{P}_2$ such that $\log (\xi \epsilon_n, \mathcal{P}_1, d) \leq n\epsilon_n^2$.
    \item Small support on $\mathcal{P}_2$: prior support on $\mathcal{P}_2$ is such that $\pi(\mathcal{P}_2) \leq e^{-(C+4)n\epsilon_n^2}$.
\end{enumerate}
The above three conditions imply that $\pi(\theta: H^2(p_{\theta_0}, p_\theta) \geq M_n \epsilon_n \mid \mathbf{X}) \to 0$ in $P_{\theta_0}$-probability as $n \to \infty$ for every $M_n \to \infty$. We show that in our case these conditions hold with $\epsilon_n = \bar{\epsilon}_n = \sqrt{s \log^2 p/n}$.

\noindent \circled{1} {\bf Prior mass:}  The vectorized data generating parameter $\text{vech}(\theta_0) \in \Theta \subset \mathbb{R}^{p(p+1)/2}$. By Assumption \ref{ass:regularity}, $\Theta$ is open, hence there exists $\epsilon >0$ such that $\{\theta: \norm{\theta - \theta_0}_2 < \epsilon\} \subset \Theta$. Moreover, from Proposition \ref{prop:KL_set}, we have:
 $$\pi\left\lbrace B(p_{\theta_0}, \epsilon)\right\rbrace =\pi\left\lbrace \theta \in \Theta: \norm{\theta - \theta_0}_2 < \delta(\epsilon)\right\rbrace.$$
 Since the prior distribution is a truncation of the prior on $\mathbb{R}^{p(p+1)/2}$, we also have: 
 \begin{align}\label{eq:trun_prior_prob}
 \nonumber \pi\left\lbrace \theta : \norm{\theta - \theta_0}_2 < \epsilon \right\rbrace & = \dfrac{\pi_L\left\lbrace \theta: \norm{\theta - \theta_0}_2 < \epsilon \right\rbrace}{\pi_L(\Theta)} \\
 & \geq \pi_L\left\lbrace \theta: \norm{\theta - \theta_0}_2 < \epsilon \right\rbrace,
 \end{align}
 since $\pi_L(\Theta) \leq 1$. Next, we obtain a lower bound of $\pi_L\left\lbrace \theta : \norm{\theta - \theta_0}_2 < \bar{\epsilon}_n \right\rbrace$ where $\bar{\epsilon}_n = \{(s \log^2 p)/n\}^{1/2}$. We will show for $\theta, \theta_0 \in \mathbb{R}^p$, and $\theta_0$ having $s$ non-zero entries,
 $$ \pi_L\left\lbrace \theta: \norm{\theta - \theta_0}_2 < \sqrt{\dfrac{s\log^2 p}{n}}  \right\rbrace \geq \pi_L\left\lbrace \theta: \norm{\theta - \theta_0}_2 < \sqrt{\dfrac{s\log p}{n}} \right\rbrace \geq e^{-Cn \bar{\epsilon}_n^2}.$$
The result then applies to our context where $\theta_0, \theta \in \mathbb{R}^{p(p+1)/2}$.
Define $S = \{j: |\theta_{0j}| \neq 0\}$. Set $\delta = \sqrt{(s\log p)/n}$. We then have:
    \begin{align*}
\mathbb{P}[\norm{\theta - \theta_0}_2 < \delta]& \geq \mathbb{P}\left[\norm{\theta_S - \theta_{0S}}_2 < \frac{\delta}{2}\right]\mathbb{P}\left[\norm{\theta_{S^c} }_2 < \frac{\delta}{2}\right]\\
&\geq \prod_{j \in S}\mathbb{P}\left[|\theta_j - \theta_{0j}|< \frac{\delta}{2\sqrt{s}}\right] \int_{I_\lambda} \mathbb{P}\left[\norm{\theta_{S^c}}_2 < \frac{\delta}{2} \mid \lambda \in I_\lambda \right] \pi(\lambda) d\lambda,
\end{align*}
where $I_{\lambda} = [\lambda_1, \lambda_2]$ with $\lambda_1 = \sqrt{\frac{np\log p}{s}}$ and $\lambda_2 = 2\lambda_1$. We next focus on the conditional probability:
\begin{align*}
    \mathbb{P}\left[\norm{\theta_{0S^c}}_2 \leq \frac{\delta}{2} \mid \lambda \in I_\lambda \right]\geq \prod_{j \notin S}\mathbb{P}\left[|\theta_j|< \frac{\delta}{2\sqrt{p}} \mid \lambda \in I_\lambda \right].
\end{align*}
First, we handle the term $\mathbb{P}\left[|\theta_j|< \frac{\delta}{2\sqrt{p}} \mid \lambda \in I_\lambda \right]$. Recall, if the random variable $X$ is sub-exponential distribution with parameters $(\nu^2, \rho),$ then, 
\begin{equation}\label{eq:sub_exponential}
    \mathbb{P}\left[ |X - \mu| \geq t \right] \leq \begin{cases}
        &2e^{-\frac{t^2}{2\nu^2}}, \quad \text{if } 0<t \leq \frac{\nu^2}{\rho},\\
        & 2e^{-\frac{t}{2\rho}}, \quad \text{if }t > \frac{\nu^2}{\rho}.
    \end{cases}
\end{equation}
Here $\theta_j \mid \lambda \sim \text{Laplace}(1/\lambda)$, and the Laplace distribution is sub-exponential with $\nu = 2/\lambda$ and $\rho = 2/\lambda$. Set $t = \frac{\delta}{2\sqrt{p}} = \frac{1}{2}\sqrt{\frac{s\log p }{np}}$ in \eqref{eq:sub_exponential} and note that $\nu^2/\rho = 2/\lambda$. Thus, for any $\lambda> 4\sqrt{\frac{np }{s \log p}}$
we have:
\begin{align*}
    \mathbb{P}\left[|\theta_j|< \frac{\delta}{2\sqrt{p}} \mid \lambda \right] &\geq 1 - 2e^{-\frac{\delta \lambda}{4 \sqrt{p}}}.
\end{align*}
Specifically, for any $\lambda \in I_\lambda$ we get, 
\begin{align*}
    \mathbb{P}\left[|\theta_j|< \frac{\delta}{2\sqrt{p}} \mid \lambda \in I_\lambda\right] &\geq 1 - 2e^{-\frac{\delta \lambda}{4 \sqrt{p}}} \geq 1 - 2e^{\frac{\delta \lambda_1}{4 \sqrt{p}}} = 1 - 2e^{-\frac{1}{4}\log p}.
\end{align*}
Thus, for some $C>0$, $\mathbb{P}\left[|\theta_j|< \frac{\delta}{2\sqrt{p}} \mid \lambda \right]\geq e^{-C\log p}$ for some $C>0$. 
We also have $\mathbb{P}[\lambda \in I_{\lambda}] = \mathbb{P}[\lambda > \lambda_1] - \mathbb{P}[\lambda >\lambda_2]$. Next, we record the following upper and lower bounds for the incomplete Gamma function \citep{pinelis2020exact},
\begin{align*}
    \dfrac{x^a e^{-x} (x - a - 1)}{(x - a)^2 + a} \leq \Gamma (a, x) = \int_{x}^\infty t^{a - 1} e^{-t} dt \leq \dfrac{x^a e^{-x}}{x-a}.
    \end{align*}
Also, we have $(a/e)^{a-1} \leq \Gamma(a) \leq (a/2)^{a-1}$. Set $\alpha = a_\lambda$ and $\beta = b_\lambda$. Thus,
\begin{align*}
\mathbb{P}(\lambda \in I_\lambda) &= \frac{\beta^\alpha}{\Gamma(\alpha)} \Gamma(\alpha, \beta \lambda_1) - \frac{\beta^\alpha}{\Gamma(\alpha)} \Gamma(\alpha, 2\beta\lambda_1) \\
& = e^{\log \beta} e^{-(\alpha - 1)\log \alpha} e^{(\alpha-1)\log 2}\Gamma(\alpha, \beta\lambda_1) - e^{\log \beta} e^{-(\alpha - 1)\log \alpha} e^{(\alpha-1)}\Gamma(\alpha, 2\beta\lambda_1)\\
& \geq e^{\log \beta} e^{-(\alpha - 1)\log \alpha} e^{K(\alpha-1)} (\beta\lambda_1)^{\alpha} e^{-C \beta \lambda_1}, \quad \text{for } \alpha = O(\beta\lambda_1)\\
& = e^{\log \beta} e^{-(\alpha - 1)\log \alpha} e^{(\alpha-1)\log 2} e^{\alpha \log \beta \lambda_1} e^{-C \beta \lambda_1}\\
& \geq e^{-C\beta \lambda_1} \geq e^{-Cs \log p},
\end{align*}
for $\beta = O\left(s\sqrt{\frac{s\log p}{np}}\right)$. Hence, we have that,
\begin{align*}
    \int_{I_\lambda} \mathbb{P}\left[\norm{\theta_{S^c}}_2 < \frac{\delta}{2} \mid \lambda \in I_\lambda \right] \pi(\lambda) d\lambda \geq e^{-Cs\log p},
\end{align*}
for some $C>0$.
We now handle the term $\mathbb{P}[|\theta_j - \theta_{0j}| < \frac{\delta}{2\sqrt{s}}]$. Since $\lambda \sim \text{Gamma}(\alpha, \beta)$, marginal density of $\theta_j$ is:
\begin{align*}
    f(\theta_j) = \dfrac{\beta^{\alpha}}{2\Gamma(\alpha)} \dfrac{\Gamma(\alpha +1)}{(|\theta_j| + \beta)^{\alpha + 1}} = \dfrac{\alpha \beta^{\alpha}}{2(|\theta_j| + \beta)^{\alpha + 1}}.
\end{align*}
Let $A = \{\theta_j : |\theta_j - {\theta_{0j}}|\leq \delta^\star, j \in S\}$ where $\delta^\star = \frac{\delta}{2\sqrt{s}}$. We have for $\alpha, \beta$ as chosen above
\begin{align*}
    \mathbb{P}\left[\theta_j \in A \right] &= \int_{A} \dfrac{\alpha \beta^{\alpha}}{2(|\theta_j| + \beta)^{\alpha +1}} d\theta_j \\
    &\geq e^{-\log 2}e^{\alpha \log \beta} e^{\log \alpha } e^{\log \delta^\star}e^{-(\alpha+1) (|\theta_{0j}|+ \delta^\star + \beta)}  \\
    & \geq e^{-C \alpha \log (\frac{1}{\beta})},
\end{align*}
which establishes that $\pi_L\left\lbrace \theta: \norm{\theta - \theta_0}_2 < \delta \right\rbrace>e^{-Cn \bar{\epsilon}_n^2}$. 

 \circled{2} {\bf Metric entropy:} Due to the relation between the Hellinger metric and the Euclidean balls, the metric entropy of the set of densities $\mathcal{P} = \{p_\theta: \theta \in \Theta\}$ can be studied using $\norm{\theta_1 - \theta_2}_2$. Define the following sieve on $\mathcal{P}$:
$$\mathcal{P}_n = \left\lbrace \theta \in \Theta: \sum \ind\{|\theta_j|> \delta_n\}\leq Cn\bar{\epsilon}_n^2, \norm{\theta}_2 \leq B\right\rbrace, $$
where $B$ is a sufficiently large number to be chosen later and $\delta_n = s \bar{\epsilon}_n/( \sqrt{p \log p})$. Arguments similar to \citet[][Theorem 3.1]{bhattacharya2015dirichlet} establish that: 
$$\log \mathcal{N}(\xi\bar{\epsilon}_n, \mathcal{P}_n, \norm{\cdot}_2) \leq C n\bar{\epsilon}_n^2,$$
where $\xi > 0$ is a positive number, $C$ is a positive constant and $\mathcal{N}(\epsilon, \mathcal{A}, \norm{\cdot})$ denotes the $\epsilon$-covering number of the metric space $(\mathcal{A}, \norm{\cdot})$.

\circled{3} {\bf Prior support on $\mathcal{P}_2$:} Next, we work with the (unrestricted) prior probability of the complement sieve, that is, $\pi_L(\mathcal{P}_n^c) $. We have that:
\begin{align*}
    \pi_L(\mathcal{P}_n^c) \leq \mathbb{P}(N \geq Cn \bar{\epsilon}_n^2 + 1) + \pi_L(\norm{\theta}_2 > B),
\end{align*}
where $N$ is the random variable denoting the number of elements in a $p(p+1)/2$-dimensional vector with marginal distribution $p(\theta) = \alpha \beta^{\alpha}/[2(|\theta| + \beta)^{(\alpha +1)}]$ exceeds $\delta_n$ in absolute value. Clearly, $N \sim \text{Bin}(p(p+1)/2, \nu)$ where $\nu = \mathbb{P}(|\theta| > \delta_n)$. Write $p^\star = p(p+1)/2$. From \cite{song2017nearly} we obtain that
\begin{align*}
    \mathbb{P}(N \geq Cn \bar{\epsilon}_n^2 + 1) &\leq 1 - \Phi\left\lbrace (2p^\star H[\nu,  Cn\bar{\epsilon}_n^2/p^\star])^{1/2}\right\rbrace\\
    & \leq (2\pi)^{-1/2} (p^\star H[\nu, Cn\bar{\epsilon}_n^2/p^\star])^{1/2} \exp\{(p^\star H[\nu, Cn\bar{\epsilon}_n^2/p^\star])^{1/2}\},
\end{align*}
where 
$$p^\star H[\nu, Cn\bar{\epsilon}_n^2/p^\star] = Cn\bar{\epsilon}_n^2 \log \left( \dfrac{Cn\bar{\epsilon}_n^2}{p^\star \nu}\right) + (p^\star - Cn\bar{\epsilon}_n^2) \log \left( \dfrac{p^\star - Cn\bar{\epsilon}_n^2}{p^\star - p^\star \nu}\right).$$
Hence, it suffices to show that $p^\star H[\nu, Cn\bar{\epsilon}_n^2/p^\star] \geq O(n \bar{\epsilon}_n^2)$. Following \citet[][Theorem 4.6]{sagar2024precision}, this in turn is true if we can establish $\nu = p^{-\gamma}$ for some positive constant $\gamma$.
Recall $\alpha = O(s\log p)$ and $\beta = O(s\sqrt{s\log p/np}) = O(s\bar{\epsilon}_n/\sqrt{p\log p})$, so that $\delta_n/\beta = O(1)$. We shall show that this probability is upper bounded by $p^{-\gamma}$ for some $\gamma>0$. Indeed,
\begin{align*}
    \mathbb{P}(|\theta| > \delta_n) &= 2\int_{\delta_n}^\infty p(\theta) d\theta \\
    & = \int_{\delta_n}^\infty \dfrac{\alpha \beta^{\alpha}}{(\theta + \alpha)^{\alpha + 1}} d\theta\\
    & = \dfrac{\alpha}{\alpha + 2} \dfrac{1}{(1 + \delta_n /\beta)^{\alpha +2}} \beta^{-2}\\
    & \approx  \beta^{-2} e^{-C(\alpha+2)\log 2} \leq p^{-\gamma},
\end{align*}
for some $\gamma>0$. 
Finally, we have,
\begin{align*}
    \mathbb{\pi}_L(\norm{\theta}_2 > B) &\leq \pi_L(\norm{\theta}_1 > B) \\
    & = \int \mathbb{P}\left[ \norm{\theta}_1 > B \mid \lambda\right] p(\lambda) d\lambda \\
    & = \int \mathbb{P}\left[Y > B \mid \lambda\right] p(\lambda) d\lambda, \quad Y = \sum_{j=1}^{p(p+1)/2}|\theta_j| \\
    & = \int \mathbb{P}\left[ \frac{2Y}{\lambda} > \frac{2B}{\lambda} \mid \lambda\right] p(\lambda) d\lambda  = \int \mathbb{P}\left[ \chi^2_{p(p+1)} > \frac{2B}{\lambda} \mid \lambda\right] p(\lambda) d\lambda,
\end{align*}
where we have used the fact that if $\theta_j\mid \lambda \overset{iid}{\sim} \text{Laplace}(1/\lambda)$ for $j = 1, \ldots, p$, then $(2/\lambda) \sum_{j} |\theta_j| \sim \chi^2_{2p}$. Hence, by Markov's inequality, we obtain,
\begin{align*}
    \mathbb{\pi}_L(\norm{\theta}_2 > B) &\leq \int \dfrac{p(p+1)}{(2B/\lambda)} g(\lambda) d\lambda\\
    &=  \dfrac{p(p+1)}{2B} \dfrac{\alpha}{\beta},
\end{align*}
with $\alpha/\beta = O(\sqrt{nps\log p})$. Hence, setting $B = O(\sqrt{nps\log p}) e^{k n\bar{\epsilon}_n^2}$ for some large positive constant $k$, we obtain that $\pi_L(\norm{\theta}_2 > B) \leq e^{-Cn\bar{\epsilon}_n^2}$ for some $C>0$. Finally, since by assumption $\pi_L(\Theta) \geq e^{-K\log p}$ for some $K>0$, we have the desired result.

\subsection{Proof of Theorem \ref{thm:EM_RBM}}\label{proof:em_rbm}

Recall the proposed estimator of $\nabla_\theta \log z(\theta)$ defined as $T_N(\theta) = T_N^{(\nabla_\theta z)}(\theta)/T_N^{(z)}(\theta)$. Next, consider the following sequence of iterates 
\begin{align*}
\theta^{(t+1)} =& \theta^{(t)} + \gamma_t [\nabla g(\v; \theta^{(t)}, \theta^{(t)}) - T_{N_t}(\theta^{(t)})] \\
=& \theta^{(t)} + \gamma_t [\nabla g(\v; \theta^{(t)}, \theta^{(t)}) - \nabla \log z(\theta^{(t)})] + \gamma_t [\nabla \log z(\theta^{(t )}) - \mathbb{E}\{T_{N_t} (\theta^{(t)})\}] \\
&+\gamma_t [\mathbb{E}\{T_{N_t} (\theta^{(t)})\} - T_{N_t}(\theta^{(t)})]\\
=& \theta^{(t)} + \gamma_t m(\theta^{(t)}) + \gamma_t R_t + \gamma_t E_t, 
\end{align*}
where $m(\theta^{(t)}) = [\nabla_\theta \log z(\theta^{(t )}) - \mathbb{E}\{T_{N_t} (\theta^{(t)})\}]$, $R_t = [\nabla_\theta \log z(\theta^{(t )}) - \mathbb{E}\{T_{N_t} (\theta^{(t)})\}] $, and $E_t = [\mathbb{E}\{T_{N_t} (\theta^{(t)})\} - T_{N_t}(\theta^{(t)})|\mathcal{F}_t]$, where $\mathcal{F}_t$ is the filtration at step $t$. Since construction of $T_{N_t}(\theta_t)$ only requires knowledge of $\theta^{(t)}$, this error term is conditionally independent of $\theta^{(0)} \ldots, \theta^{(t-1)}$ given $\theta^{(t)}$. In stochastic approximation literature \citep{robbins1951stochastic, delyon1999convergence}, the term $m(\theta)$ is known as the mean-field of the algorithm, the random variables $R_t$ can be thought of vanishingly small remainder or bias terms and $E_t$ is a random error with mean $0$. Furthermore, we allow the number of Monte Carlo samples to increase with the iterations. This makes the remainder term $R_t$ to converge to $0$ as the number of iterations increase. In the following series of steps, we establish some key properties about the estimator $T_N(\theta)$ which will be key in establishing the theorem.
\vspace{10pt}

\noindent { \bf \circled{A} The sequence of functions $\mathbf{\{T_N(\theta)_{N\geq 1}\}}$ is equicontinuous in $\mathbf{\Theta}$:} By the mean value theorem, it is enough to show that there exists $C>0$ such that:
$$\norm{\nabla_\theta T_N(\theta)}_F < C,$$
for all $N\geq 1$. We have, 
\begin{align*}
    &\nabla T_N(\theta) = T_N^{(z)}(\theta) \nabla T_N^{(\nabla_\theta z)}(\theta) - \nabla T_N^{(\nabla_\theta z)}(\theta)\nabla T_N^{(z)}(\theta),\\
    \text{i.e., } &\norm{\nabla T_N(\theta)}_F \leq |T_N^{(z)}(\theta)| \norm{\nabla T_N^{(\nabla_\theta z)}(\theta)}_F + \norm{\nabla T_N^{(\nabla_\theta z)}(\theta)\nabla T_N^{(z)}(\theta)}_F.
    \end{align*}
    Note that $T_N^{(z)}(\theta)  = N^{-1}\sum_{i=1}^N \exp\{(\theta - \phi)' X_{i\bullet}\}$. Since $\Theta$ is compact, and $\theta \mapsto e^{\theta'y} \in \mathbb{R}$ is continuous, the map admits a lower and upper bound, say $[l_y, u_y]$. Here, the dependence of $l, u$ on $y$ is implicit. However, since $y$ is also bounded as is the case in our setting, and the map $y \mapsto e^{a'y}$ is continuous, there exists $l = \inf_y l_y > -\infty$ and $u = \sup_y u_y <\infty$. Thus, $T_N^{(z)}(\theta)\geq l $. Also,
    \begin{align*}
        \norm{\nabla T_N^{(\nabla_\theta z)}(\theta)}_F &= N^{-1}\sum_{i=1}^N \norm{\exp\{(\theta -\phi)' X_{i \bullet}\} X_{i\bullet} X_{i\bullet}'}_F\\
        & \leq N^{-1}\sum_{i=1}^N\exp\{(\theta -\phi)' X_{i \bullet}\}\norm{ X_{i\bullet} X_{i\bullet}'}_F\\
        & \leq u \sqrt{(p+m)},
        \end{align*}
since elements of the $(p+m) \times 1$ vector $X_{i\bullet}$ are either 0 or 1. Next, $\norm{\nabla T_N^{(\nabla_\theta z)} (\theta)}_2 \leq N^{-1} \sum_{i=1}^n \exp\{(\theta - \phi)'X_{i\bullet}\}\norm{X_{i\bullet}}_2 \leq u\sqrt{(p+m)}$. Hence, 
\begin{align*}
    \norm{\nabla T_N^{(\nabla_\theta z)}(\theta)\nabla T_N^{(z)}(\theta)}_F &\leq \norm{\nabla T_N^{(\nabla_\theta z)} (\theta)}_2 \norm{\nabla T_N^{(z)}(\theta)}_2\\
    & \leq N^{-1}u\sqrt{(p+m)} \sum_{i=1}^N  \norm{\exp\{(\theta - \phi)' X_{i\bullet}\}X_{i\bullet}}_2\\
    & \leq u^2 (p+m).
\end{align*}
Thus, we have that $T_N(\theta) \to \nabla_\theta \log z(\theta)$ almost surely and the sequence is equicontinuous on $\Theta$. Hence, $T_N(\theta) \to \nabla_\theta \log z(\theta)$ uniformly on $\Theta$ with probability $1$ \citep[Exercise 7.16]{rudin1964principles}.

\noindent {\bf \circled{B} $\mathbf{\text{lim}_{t \to \infty}\sum_{k=1}^t \gamma_k E_k}$ exists almost surely:} First, note that $\norm{T_N(\theta)}_2 < \infty$. This is a consequence of the arguments presented above. Let $M_t = \sum_{k=1}^t \gamma_k E_k$ and $E_0 = 0$. Then $M_k$ is an $\mathcal{F}$-martingale with $\mathcal{F} = \cap_{t=1}^\infty \mathcal{F}_t$. Moreover,
\begin{align*}
    \sum_{t=1}^\infty \mathbb{E}\left[ \norm{M_{t} - M_{t-1}}_2^2 \mid \mathcal{F}_{t-1}\right] \leq \sum_{t=1}^\infty \gamma_t^2 \mathbb{E}[\norm{E_t}_2^2 \mid \mathcal{F}_{t-1}] < \infty.
\end{align*}
This implies that $\lim_{t \to \infty} \sum_{k=1}^\infty \gamma_kE_k $ exists almost surely.

\noindent {\circled{C} \bf The estimator $\mathbf{T_N(\theta)}$ is asymptotically uniformly unbiased:} First note that $T_N(\theta)$ is a vector. We recall that if $x_n$ is a sequence in $\mathbb{R}^p$, then $x_n \to x \in \mathbb{R}^p$ iff $\Psi(x_n) \to \Psi(x)$ for all linear functional $\Psi$. Now, $|\Psi(T_N(\theta))| \leq \norm{\Psi}\norm{T_N(\theta)},$ which is bounded following similar arguments as above. Hence, by the dominated convergence theorem, 
$$\lim_{N\to \infty} \mathbb{E}[\Psi(T_N(\theta))] \to \mathbb{E} [\lim_{N\to \infty} \Psi(T_N(\theta))] = \Psi(\nabla_\theta \log z(\theta)).$$
Since $\Psi (\mathbb{E} [T_N(\theta)] )= \mathbb{E}[\Psi(T_N(\theta))]$, we have that: 
$$\lim_{N\to \infty} \mathbb{E}[T_N(\theta)] \to \nabla_\theta \log z(\theta).$$
Moreover, by uniform convergence of $T_N(\theta)$ to $\nabla_\theta \log z(\theta)$, for every $\epsilon>0$, there exists $N(\epsilon)$ such that:
\begin{align*}
    \norm{\mathbb{E}[T_N(\theta) - \nabla_\theta \log z(\theta)]}_2 &= \norm{\int [T_N(\theta) - \nabla_\theta \log z(\theta)]\prod_{i=1}^N p_{\phi}(X_{i\bullet}) dX_{i\bullet}}_2\\
    &\leq \int \norm{T_N(\theta) - \nabla_\theta \log z(\theta)}_2 \prod_{i=1}^N p_{\phi}(X_{i\bullet}) dX_{i\bullet}\\
    & \leq \epsilon,
\end{align*}
for all $\theta$.
An immediate consequence of this is that there exists a sequence of Monte Carlo sample sizes $\{N_t\}_{t\geq 1}$ such that $\norm{R_t}_2 = \norm{\mathbb{E}[T_{N_t}(\theta^{(t)})] - \nabla_\theta \log z(\theta^{(t)})}_2 \to 0$ as $t\to \infty$.

\noindent {\bf \circled{D} The mean field satisfies $\mathbf{m(\theta) = \nabla_\theta \log p_\theta(\v)}$:} By Fisher's identity, we have:
\begin{align*}
    \nabla_\theta \log p_\theta(\v) = \dfrac{\nabla_\theta p_\theta(\v)}{ p_\theta (v)} = \dfrac{\nabla_\theta \sum_{\h} p_\theta(\v, \h)}{ p_\theta (v)} &= \mathbb{E}[ \nabla_\theta \log p_\theta(\v, \h) \mid \v, \theta]\\ 
    & = \nabla_\theta g(\v; \theta, \theta)  -\nabla_\theta \log z(\theta).
\end{align*}

\noindent{\bf \circled{E} $\mathbb{E}[\mathbf{\norm{\h}_2^2} \mid \theta, \v] < \infty$:} This follows from the fact that $\norm{\h}_2^2$ is bounded and $\h\mid \v, \theta$ also follows an Ising distribution.

The conclusion of the theorem follows from \citet[][Theorems 2 and 5]{delyon1999convergence} by noting that the complete data model is a member of the exponential family satisfying the regularity conditions for these theorems and due to \circled{A}, \circled{B}, \circled{C}, \circled{D}. 

{\color{red}
\subsection{Convergence of the stochastic proximal gradient descent algorithm}\label{sec:conv}
In Theorem \ref{thm:l1_consistency} of the main draft, we provided convergence rate of the $\ell_1$-penalized estimator of the maximum likelihood as the solution of the optimization problem $\argmin_\theta Q(\theta)$ defined in Equation~\eqref{eq:l1_penalized_estimator}. The result is novel in the sense that it provides convergence guarantees similar to \cite{rothman2010sparse} for PEGMs other than the Gaussian graphical model. However, it depends on the analytical availability of $\nabla_\theta \log z(\theta)$, which of course is not available, and is the primary focus of the article. Here, we study the convergence of the proximal gradient descent algorithm when $\nabla_\theta \log z(\theta)$ is estimated by the proposed importance sampling estimator. Recall,
\begin{align}\label{eq:l1_penalized_estimator_supp}
    \hat{\theta}_\lambda = \argmin_{\theta \in \Theta} -\ell(\theta) + \lambda\norm{\theta}_1 &= \argmin_{\theta \in \Theta} Q(\theta) \nonumber\\
    & = \argmin_{\theta \in \Theta} \sum_{i=1}^n \nabla_\theta q_\theta(\mathbf{X}) - n\nabla_\theta \log z(\theta) + \lambda \norm{\theta}_1.
\end{align}
This optimization problem is solved by replacing $\nabla_\theta \log z(\theta)$ by the estimator $T_N^{(\nabla_\theta \log z)}$. In particular, our algorithm has the following system of iterations:
\begin{align*}
    \tilde{\theta}^{(t+1)} &= \theta^{(t)} - \gamma_t \left[ \sum_{i=1}^n \nabla_\theta q_\theta(\mathbf{X}) - T_N^{\nabla_\theta \log z}\right], \nonumber\\
    \theta^{(t+1)} &= \text{prox}_{\lambda\gamma_t}^{\norm{\cdot}_1}(\tilde{\theta}^{(t+1)}),
\end{align*}
where $\gamma_t$ is the step size, and $\text{prox}_{\lambda \gamma_t}^{\norm{\cdot}_1}$ is the well-known soft-thresholding operator.
This follows the framework of \cite{rosasco2020convergence}. The natural question is whether this random sequence of iterates converge in some sense to the minimizer of $Q(\theta)$. For this we make the following assumptions:
\begin{assumption}\label{ass:compact_theta}
    The parameter space $\Theta$ for which $z(\theta)$ is finite is compact and there exists $c>0$ such that $z(\theta) \geq c$.
\end{assumption}
\begin{assumption}\label{ass:unique_minimizer}
    There exists a unique minimizer $\theta^\star$ of $Q(\theta)$.
\end{assumption}
Assumption \ref{ass:compact_theta} ensures $\nabla_\theta z(\theta)$ is $\beta$-Lipschitz gradient for some $\beta >0$. This is not true in general when $\Theta$ is not compact; e.g. the Poisson distribution. Assumption \ref{ass:unique_minimizer} allows us to study the convergence of our stochastic proximal algorithm. 

Next, we note that since $T_N^{(\nabla_\theta \log z)}$ is a ratio-type estimator, it is biased in general in finite samples. Define the bias as:
\begin{equation*}
    b_N(\theta) = \mathbb{E}\left[ T_N^{(\nabla_\theta \log z)}\right] - \frac{\nabla_\theta z(\theta)}{z(\theta)},
\end{equation*}
so that,
\begin{align*}
    \norm{b_N(\theta)} &\leq \norm{\mathbb{E}\left[ T_N^{(\nabla_\theta \log z)}\right]} + \norm{\frac{\nabla_\theta z(\theta)}{z(\theta)}}\\
    & \leq \mathbb{E}\left[\norm{T_N^{(\nabla_\theta \log z)}}\right] + \norm{\frac{\nabla_\theta z(\theta)}{z(\theta)}} \leq p\tau_1(\theta)/N,
\end{align*}
for some function $\tau_1(\theta)$ since $\Theta$ is compact and $p$-dimensional. Let $\mathcal{A}_t = \sigma(\theta^{(1)}, \ldots, \theta^{(t)})$. We then have
\begin{align*}
    \mathbb{E}\left[ \norm{T_N^{(\nabla_\theta \log z)} - \nabla_\theta\log z(\theta)}_2^2 \bigg \vert \mathcal{A}_t\right] &= \sum_{j=1}^p \mathbb{E}\left[\left(T_{N,j}^{(\nabla_\theta \log z)} - \frac{\partial \log z(\theta)}{\partial \theta_j}\right)\bigg \vert \mathcal{A}_t\right]^2,\\
    &= \sum_{j=1}^p \left[\text{Var}\left(T_{N,j}^{(\nabla_\theta \log z)}\right) + b_{N,j}(\theta)^2\right],
\end{align*}
where $T_{N,j}^{(\nabla_\theta \log z)}$ is the $j$-th element of the proposed gradient estimator. For our analysis here, we assume that two batches of samples  $(Y_1^{(1)}, \ldots Y_N^{(1)})$ and $(Y_1^{(2)}, \ldots Y_N^{(2)})$ are drawn independently from $p_\phi$ to construct the estimator: 
$$T_{N,j}^{(\nabla_\theta \log z)} = \dfrac{\sum_{i=1}^N \frac{\nabla_{\theta_j}q_\theta\left(Y_i^{(1)}\right)}{q_{\phi}\left(Y_i^{(1)}\right)}}{\sum_{i=1}^N \frac{q_\theta\left(Y_i^{(2)}\right)}{q_{\phi}\left(Y_i^{(2)}\right)}} = \dfrac{T_{N,j}^{(1)}(\theta)}{T_{N,j}^{(2)}(\theta)}.$$
Assuming that the conditions of Proposition~\ref{prop:exp} hold, we then have by the central limit theorem that both the numerator and the denominator are asymptotically normal. Also, by construction, they are independent. Hence, by delta method, we get that,
$$\text{Var}\left(T_{N,j}^{(\nabla_\theta \log z)}\right) = \left(\frac{\partial \log z(\theta)}{\partial \theta_j}\right)^2\tau_2^2(\theta)/N + o(1/N),$$
where $\tau_2^2(\theta) = \text{Var}\left( T_{N,j}^{(1)}(\theta)\right) + \text{Var}\left( T_{N,j}^{(2)}(\theta)\right)$. Thus, we obtain,
\begin{align*}
    \mathbb{E}\left[ \norm{T_N^{(\nabla_\theta \log z)} - \nabla_\theta\log z(\theta)}_2^2 \bigg \vert \mathcal{A}_t\right] \leq \frac{C(\theta)}{N}\left[ 1 + \norm{\nabla_\theta \log z(\theta)}_2^2\right] \leq C\left[ 1 + \norm{\nabla_\theta \log z(\theta)}_2^2\right],
\end{align*}
since $\Theta$ is compact. We now make an assumption about the step sizes which is common in the literature of stochastic optimization.
\begin{assumption}\label{ass:step_sizes}
    The step-sizes satisfy $\gamma_t \leq 1/\beta$, $\sum \gamma_t = \infty$, $\sum{\gamma_t^2} < \infty$ and the Monte Carlo sample size $N = N(t)$ is chosen such that $\sum{\gamma_t}\norm{b_N(\theta^{(t)})} < \infty$.
\end{assumption}
The key takeaway is that to control the bias term one needs to increase $N$ with the number of steps. This also appears frequently in the literature, see, e.g., \citet{atchade2017perturbed}. We now have the following result.
\begin{theorem}
    Under assumptions \ref{ass:compact_theta}, \ref{ass:unique_minimizer}, \ref{ass:step_sizes}, $\mathbb{E}\norm{\theta^{(t)} - \theta^\star}_2^2 \to 0$.
\end{theorem}
\begin{proof}
    The result follows from Corollary 3.3 of \cite{rosasco2020convergence}.
\end{proof}
}


\section{A Simple Accept--Reject Sampler for $p_\theta(\cdot)$\label{app:sampler}}
Generating samples from a Boltzmann distribution, of which PEGMs are special cases, remains a challenge; especially without resorting to computationally expensive MCMC techniques. However, our approach also accomplishes this goal. Specifically, for a given $\theta$, suppose we simulate $Y\sim p_\phi(\cdot),$ where $\phi=\mathrm{diag}(\theta)$. This is cheap, and does not require MCMC. Moreover, $z(\phi)$ is analytically available. Next, recall, $p_\theta(x) = q_\theta(x)/z(\theta)$. Thus, an accept--reject sampler for $p_\theta(\cdot)$ consists of the following two steps:
\begin{enumerate}
\item Simulate $Y\sim p_\phi(\cdot)$
\item Accept $Y$ with probability $r=q_\theta(Y)/\{Mp_\phi(Y)\}$ where $M$ is such that $q_\theta(Y) < M p_\phi(Y)$ for all $Y$.
\end{enumerate}
Thus, $M$ must be chosen to satisfy the criterion:
\begin{align*}
\exp\left(\sum_j \theta_j Y_j + \sum_{j\ne k} \theta_{jk} Y_j Y_k\right)&< M \exp\left(\sum_j \theta_j Y_j\right)/z(\phi),\\
\text{i.e., } \exp\left(\sum_{j\ne k} \theta_{jk} Y_j Y_k\right) &<Mz(\phi)^{-1}.
\end{align*}
However, a simple choice of $M$ is available for several PEGMs. For example, in a Poisson graphical model, all off-diagonal elements of a valid $\theta$ are negative, and one may take $Mz(\phi)^{-1}=1$, i.e., $M=z(\phi)$. Similarly, for Ising, where $Y_j=\{0,1\}$ for all $j$, a simple choice is $\{M: M z(\phi)^{-1}>\exp(p^2\max_{jk} \theta_{jk})\}$. Other, less obvious choices of $M$ may result in better average acceptance probabilities. We do not explore the details in this work.

\section{Algorithmic Details for Section \ref{sec:PEGM_inference}}\label{app:computational_details}
\subsection{Tuning $\lambda$ in \eqref{eq:high_dim_projected_gradient_descent} and model selection}
Here we describe how the value of the tuning parameter $\lambda$ is selected in \eqref{eq:high_dim_projected_gradient_descent}. Generally, pseudo-likelihood methods have a regression structure in that the likelihood of $X_j \mid X_{-j}$ is a generalized linear model for PEGMs. As a result, out-of-sample prediction is typically used for tuning parameter selection. This is, however, not the case when a full-likelihood analysis is implemented. In our work, we instead consider the out-of-sample log-likelihood to select the value of $\lambda$ given estimated $\hat{\theta}_\lambda$ is available for a grid of values of $\lambda$: $[\lambda_l, \lambda_u]$. Specifically, we divide the observed data randomly into $K$-folds, of which $K-1$ folds are used to obtain $\hat{\theta}_\lambda$ for each value of $\lambda$. The log-likelihood $\ell(\hat{\theta}_\lambda)$ is computed on the remaining fold. The process is repeated for each fold, and we choose $\hat{\theta}_\lambda$ to be our estimator for which we obtain the maximum out-of-sample log-likelihood. This generally results in improved Frobenius norm estimates of $\theta$ as reported in Table \ref{tab:Ising_PGM_table}. However, it is well known that cross-validated estimates in penalized estimation procedures may provide sub-optimal performance in terms of model selection \citep{buhlmann2011statistics}. To address this, we construct an estimator of $\theta$ that selects edges between nodes which are most often selected. Suppose, for $R$ choices of $\lambda$, we obtain estimates $\hat{\mathbf{S}}_{\lambda_r} = \{(j,k):\hat{\theta}_{\lambda_r, jk} \neq 0\}$. Define $I_{jk}^r = \mathbb{I}\{(j,k) \in \hat{\mathbf{S}}_{\lambda_r}\}$. Then, we estimate the structure of the graph $G$ by:
\begin{equation}
    \hat{\mathbf{S}} = \left\lbrace (j,k): \frac{1}{R} \sum_{r=1}^R I^r_{jk} > \pi_{thr}, \right\rbrace
\end{equation}
where $0<\pi_{thr}<1$. The procedure closely resembles stability selection \citep{meinshausen2010stability}, except here we do not perform subsampling. Our experiments with and without subsampling resulted in similar performance. Hence, we omit that step. For the threshold $\pi_{thr}$, we use $0.6$. Results from this choice are reported under the MCC column in Table \ref{tab:Ising_PGM_table}. For fair comparison, we also implemented this structure learning method with pseudo likelihood-based approaches.

\subsection{HMC for PEGMs}
Suppose we observe iid data $\mathbf{X} = (X_{1\bullet}, \ldots, X_{n\bullet})$ from some $\text{PEGM}(\theta)$. Consider the following prior structure over $\theta$: $\pi(\theta) = \int \left[\prod_{j\leq k} \pi(\theta_{jk} \mid \tau)\right]g(\tau)d\tau$, i.e. the elements $\theta_{jk}$ are conditionally independent given a global shrinkage parameter $\tau$. The posterior is then:
\begin{equation}\label{eq:bayes_posterior}
    \pi(\theta \mid \mathbf{X})\propto \left[\prod_{i=1}^n \dfrac{q_\theta(X_{i\bullet})}{z(\theta)}\right] \pi(\theta).
\end{equation}
 We make two simplifying assumptions. First, we fix  $\tau = 1$, noting that our procedure can be easily incorporated in the general case when $\tau \sim g$ with an additional layer of sampling. Second, we assume that the convexity of $\Theta$ can be represented as $C(\theta) \geq 0$ for a suitable choice of a differentiable $C(\theta)$. 

Suppose $\pi(\theta_{jk})$ admits the marginal representation $\pi(\theta_{jk}) = \int \pi(\theta_{jk}\mid \rho_{jk}) \pi(\rho_{jk})d\rho_{jk},$ with respect to some latent variables $\rho_{jk}$, and $\pi(\theta_{jk}\mid \rho_{jk})$ is differentiable. This is indeed the case for any prior that is a scale mixture of Gaussian. Then samples from the augmented posterior $\pi(\theta, \{\rho_{jk}\}_{j\leq k}\mid \mathbf{X})$ can be obtained by employing a Gibbs sampler that iterates between sampling from $\pi(\theta \mid \{\rho_{jk}\}_{j\leq k}, \mathbf{X})$ and $\pi(\{\rho_{jk}\}_{j\leq k} \mid \theta, \mathbf{X})$. We use a HMC sampler to sample from $\pi(\theta \mid \{\rho_{jk}\}_{j\leq k}, \mathbf{X})$ and sample $\rho_{jk}$'s in a block since these are conditionally independent. To ensure that the HMC sampler lives within $C(\theta)$, we adopt the constrained HMC sampler from \citet{betancourt2011nested}. A full description of the transition operation of the sampler is provided in Algorithm \ref{algo:constrained_HMC} where we write $\theta = \text{vech}(\theta)$ with a slight abuse of notation. Let $U(\theta)$ denote the negative of the logarithm of $\pi(\theta \mid \{\rho_{jk}\}_{j \leq k}, \mathbf{X})$ up to proportionality constants, i.e. $U(\theta) = -\sum_{i=1}^n \log q_\theta(X_{i\bullet}) + n \log z(\theta) - \sum_{j\leq k} \log \pi(\theta_{jk}\mid \rho_{jk})$, and $\nabla_\theta U(\theta)$ denote its gradient which is available following Proposition \ref{prop:geyer_grad} and due to the assumption that $\pi(\theta_{jk}\mid \rho_{jk})$ is differentiable.
The new state generated according to Algorithm \ref{algo:constrained_HMC} is accepted with the standard Metropolis correction \citep[Section 5.3.2.1]{neal2011mcmc}.
\begin{algorithm}[!t]
    \caption{Constrained HMC sampler for updating $\theta \mid \{\rho_{jk}\}_{j \leq k}, \mathbf{X}$}
    {\bf Input}: $\theta$ [current state], $\epsilon$ [step size], $L$ [number of leapfrog steps]

    {\bf Output}: $\theta$ [next update]
    \begin{algorithmic}
\State Sample $p \sim \Gauss(0, I)$ and set $p = p - \frac{1}{2}\epsilon \nabla_\theta U(\theta)$, 
\While{$l \leq L$}
\State $\theta =  \theta + \epsilon p$
\If{$C(\theta)\geq 0$}
    \State $p \gets p - \epsilon \nabla_\theta U(\theta)$
\Else
    \State $r \gets  \dfrac{\nabla C(\theta)}{\norm{\nabla C(\theta)}} $
    \State $p \gets p - 2(r'p)r$
\EndIf
\EndWhile
\State $\theta \gets \theta + \epsilon p$
\State $p \gets p - \frac{1}{2}\epsilon \nabla_\theta U(\theta)$
\end{algorithmic}
\label{algo:constrained_HMC}    
\end{algorithm}


\section{Additional Numerical Experiments}
\subsection{Quality assessment of the Monte Carlo estimates}\label{sup:mc_imp}
 We conduct a simulation study on Gaussian graphical models (GGMs) to assess the validity and efficiency of our Monte Carlo estimates. Clearly, the GGM case is not intractable, i.e., it has a known $z(\theta)$ and $\nabla_\theta z(\theta)$. However, precisely because of this reason, the GGM case serves as a useful \emph{oracle}. Let $\Theta = \mathcal{M}^{+}_p$, the space of $p\times p$ positive definite matrices, which is the valid parameter space for a GGM.  Consider a random vector $x = (x_1, \ldots, x_p)^T \in \mathbb{R}^p$ following a GGM with zero mean and precision matrix $\theta \in \Theta$. Then, the model can be seen to be a PEGM with density $p_\theta(x) = q_\theta(x)/z(\theta)$ with: 
$$q_\theta(x) = \exp\left\{-\frac{1}{2}\left(\sum_{j=1}^p \theta_{jj}x_j^2 + 2 \sum_{j< k} \theta_{jk} x_j x_k\right)\right\},$$
with standard properties of a multivariate Gaussian density yielding:
\begin{equation}
z(\theta) = (2\pi)^{p/2}\text{det}(\theta)^{-1/2},\;\; \nabla_\theta z(\theta) = -\frac{1}{2}(2\pi)^{p/2} \text{det}(\theta)^{-1/2}\theta^{-1}, \;\;\nabla_\theta \log z(\theta)=- \frac{1}{2}\theta^{-1},\label{eq:truth}
\end{equation}
and, for $\phi=\mathrm{diag}(\theta),$ one obtains:
$$z
(\phi)=(2\pi)^{p/2}\prod_{i=1}^{p} \theta^{-1/2}_{ii}.
$$ 
Recall that in an intractable PEGM, none of the quantities in Equation~\eqref{eq:truth} is available in closed form in general, although $z(\phi)$ is available. However, the tractable Gaussian case with known truth in Equation~\eqref{eq:truth} provides a test case for evaluating the quality of our estimates against the truth, when only the knowledge of $q_\theta(x)$ and $\nabla_\theta q_\theta(x)$ is used in forming the estimates, similar to what we have done for intractable PEGMs.

Recall that we define the respective importance sampling estimates of $z(\theta)/z(\phi), \nabla_\theta z(\theta)/z(\phi)$ and $\nabla_\theta \log z(\theta)$ as: 
 $$T^{(z)}_N = \frac{1}{N}\sum^N_{i = 1}\frac{q_{\theta}(X_{i\bullet})}{q_{\phi}(X_{i\bullet})},\; T^{(\nabla_\theta z)}_N=\frac{1}{N}\sum^N_{i = 1} \frac{\nabla q_{\theta}(X_{i\bullet})}{q_{\phi} (X_{i\bullet})},  \text{ and, } T^{(\nabla_\theta \log z)}_N: = \frac{T^{(\nabla_\theta z)}_N}{T^{(z)}_N},
 $$ 
 where $X_{i\bullet}, \ldots, X_{N\bullet}$ are $N$ independent and identically distributed samples drawn from the distribution $p_{\phi}(x)$.

Satisfying the condition of the Proposition \ref{prop:exp} that $2u+\phi\in \Theta$, where $u=\theta-\mathrm{diag}(\theta)$, we generate dense precision matrices $\theta = \Sigma^{-1}$ in low-dimensional ($p = 5$) case and sparse precision matrices in high-dimensional scenarios ($p = 50, 100, 300$). 
\begin{itemize}
    \item \textbf{Mixed Covariance.} Let  $\Sigma = \{\sigma_{ij}\},$ where $\sigma_{ii} = 2$ and $\sigma_{ij} = \sigma_{ji}$ follows $U(-0.6, 0.6)$. 
    \item \textbf{Band graph}. Let $\theta = \{\theta_{ij}\},$ where $\theta_{ii} = 3$, $\theta_{i,i+1} = \theta_{i+1,i} =  0.3$ and $\theta_{ij} = 0$ for $|i - j| \geq 2$. In this case $2u + \phi \in \Theta$.
\end{itemize}
For each scenario, we compute $T^{(z)}_N$, $T^{(\nabla_\theta z)}_N$ and the ratio $T^{(\nabla_\theta z)}_N/T^{(z)}_N$ across $100$ replicates with varying Monte Carlo sample sizes $N$. We then compare these estimators to the true values of $z(\theta)/z(\phi)$, $\nabla_{\theta} z(\theta)/z(\phi)$ and $\nabla_{\theta}\log z(\theta)$, computing the following quantities for each Monte Carlo replication:
\footnotesize
\begin{equation}
\mathrm{SE}(z) =\left\lbrace\frac{z(\theta)}{z(\phi)} - T^{(z)}_N\right\rbrace^2,\; \mathrm{Fr}(\nabla_\theta z) = \frac{||\frac{\nabla_\theta z(\theta)}{z(\phi)} - T^{(\nabla_\theta z)}_N||_F}{p^2} ,\; \mathrm{Fr}(\nabla_\theta \log z) = \frac{||\nabla_{\theta}\log z(\theta) - \frac{T^{(\nabla_\theta z)}_N}{T^{(z)}_N}||_F}{p^2},\label{eq:estimates}
\end{equation}
\normalsize
where for the latter two, we scale the Frobenius norm by the number of elements in the matrix $\theta$ to make the quantities comparable across $p$.

Table \ref{tab:GGM}  reports the average of these numbers across the 100 Monte Carlo replicates. In both low and high dimensional cases, the estimates $T^{(z)}_N$ are close to the truth $z(\theta)/z(\phi)$, with MSE decreasing at the expected rate of $N^{-1}$. {Histograms of $T^{(z)}_N$ for $p = 100$ in  Figure \ref{fig:TZ_densityplot} demonstrate that the estimator is centered around the true value and follows an approximately normal distribution.} Similarly, $\mathrm{Fr}(\nabla_{\theta}z)$ and $\mathrm{Fr}(\nabla_{\theta}\log z)$ confirm the stability of the estimates $T^{(\nabla z)}_N$ and $T^{(\nabla z)}_N/T^{(z)}_N$ for all $p$ we considered. 

We remark here that good performance is achieved only when the conditions of the Proposition \ref{prop:exp} are satisfied. For instance, with $\theta_{i, i+1} = \theta_{i+1, i} = -0.7$ in the same band graph structure for $p=100$, we have $(2u+\phi)\not\in \Theta$, and this  results in a high MSE of $2979.6$ for $T^{(z)}_N$ with $N = 1000$. 

Additionally, we compare $T^{(\nabla_\theta z)}_N/T^{(z)}_N$ with the estimator $G_N^{(\nabla_\theta \log z)}$, mentioned in Remark \ref{rem:diag} that uses draws from $p_\theta(\cdot)$. In practice, these draws are obtained via running a Gibbs sampler at the current $\theta$. Figure \ref{fig:comptime_GT} shows that $G_N^{(\nabla_\theta \log z)}$ requires over 25 times the computation time to reach the approximately same $\mathrm{Fr}(\nabla_\theta \log z)$ for $p=100$. For a fixed computational budget, the proposed estimator performs much better than the Gibbs sampling based estimator. The reasons for these are: (a) unlike the i.i.d. samples from $p_\phi(\cdot)$ used for our importance sapling estimate, the Gibbs sampler takes some time to mix, and, (b) even after mixing, updating all coordinates takes long for large $p$, since batch sampling is not possible in Gibbs.

\begin{table}[!h]
  \caption{
  Mean (sd) of the performance metrics for the estimates as defined in Equation~\eqref{eq:estimates} for Gaussian Graphical Models across 100 replicates. 
\label{tab:GGM}}
  \centering
 \scalebox{0.8}{
  \begin{tabular}{ccccccc}
    \toprule
     $p$ &   Model & $N$  &  $\mathrm{SE}(z)$  & $\mathrm{Fr}(\nabla_\theta z) $ & $\mathrm{Fr}(\nabla_\theta \log z) $ \\
     
    \midrule
 \multirow{2}{*}{$5$} & \multirow{2}{*}{Mixed Covariance} & $1000$  & 0.0007 (0.001) & 0.015 (0.005) & 0.01 (0.004)\\
        & & $10000$   &  0.00006 (0.00009) & 0.005 (0.002) & 0.004 (0.001)\\
         \midrule
      \multirow{2}{*}{50} & \multirow{2}{*}{Band Graph} & $5000$ & 0.0002 (0.0003)  & 0.00009 (0.000008) & 0.00007 (0.000005)\\
     
      &  &  $50000$ &  0.00003 (0.00003) & 0.00003 (0.000002) & 0.00002 (0.000001) \\
        \midrule
      \multirow{2}{*}{$100$} & \multirow{2}{*}{Band Graph} & $5000$ & 0.001 (0.002) & 0.00008 (0.00001) & 0.00005 (0.000006)\\
     
      &  &  $50000$ & 0.0001 (0.0001) & 0.00002 (0.000001) & 0.00001 (0.000001)\\
      \midrule
     
      \multirow{2}{*}{$300$} & \multirow{2}{*}{Band Graph} & $5000$ &  0.20 (0.98) & 0.0002 (0.0002) & 0.00004 (0.00003)\\
     
      &  &  $50000$ & 0.008 (0.01)  & 0.00006 (0.00001) &  0.00001 (0.00003)\\
    \bottomrule
    \end{tabular}
 }
    \end{table}

\begin{figure}[!h]
\centering
    \begin{subfigure}{0.48\textwidth}
       \includegraphics[height = 5cm, width = 7cm]{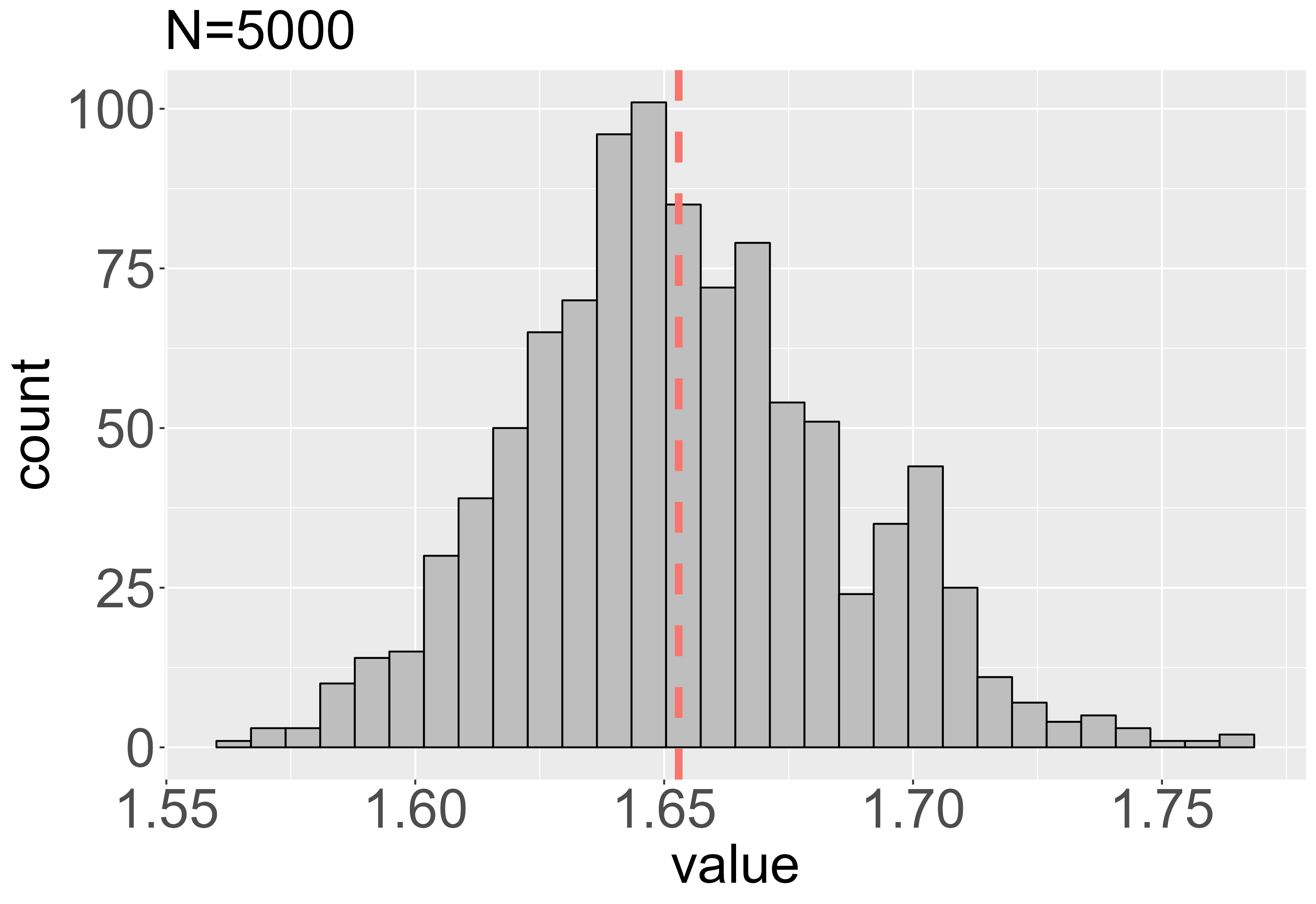}
   \end{subfigure}
   \begin{subfigure}{0.48\textwidth}
       \includegraphics[height = 5cm, width = 7cm]{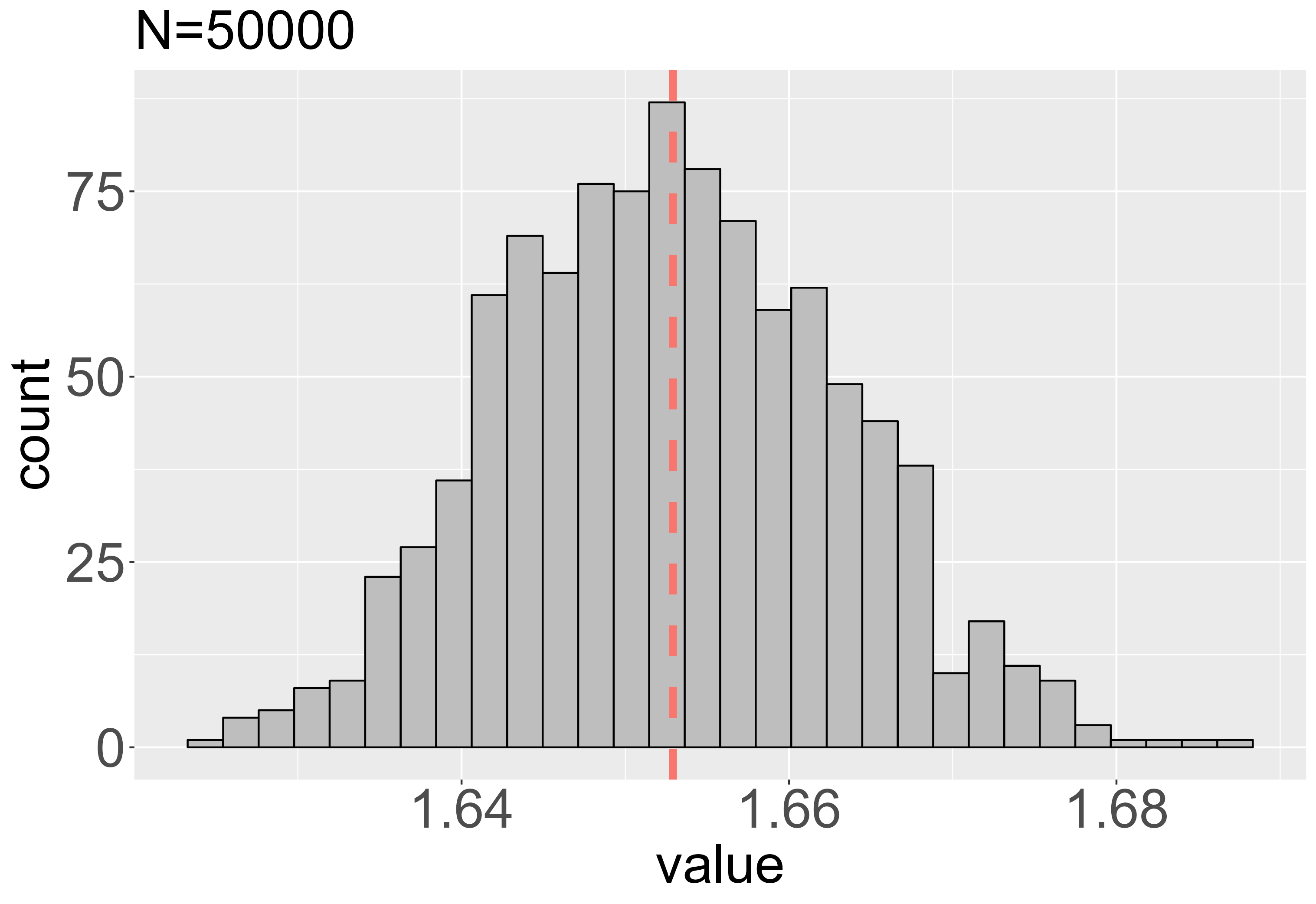}
   \end{subfigure}
    \caption{Histograms of 1000 replications of $T_{N}^{(z)}$ for $p = 100$ with Monte Carlo sample sizes $N = 5,000$ and $50,000$. The vertical red dashed line indicates the true value of ${z(\theta)}/{z(\phi)}$. Note the smaller spread of the histogram on the right.}
    \label{fig:TZ_densityplot}
\end{figure}

\begin{figure}[!h]
\centering
   \includegraphics[width =0.75\textwidth]{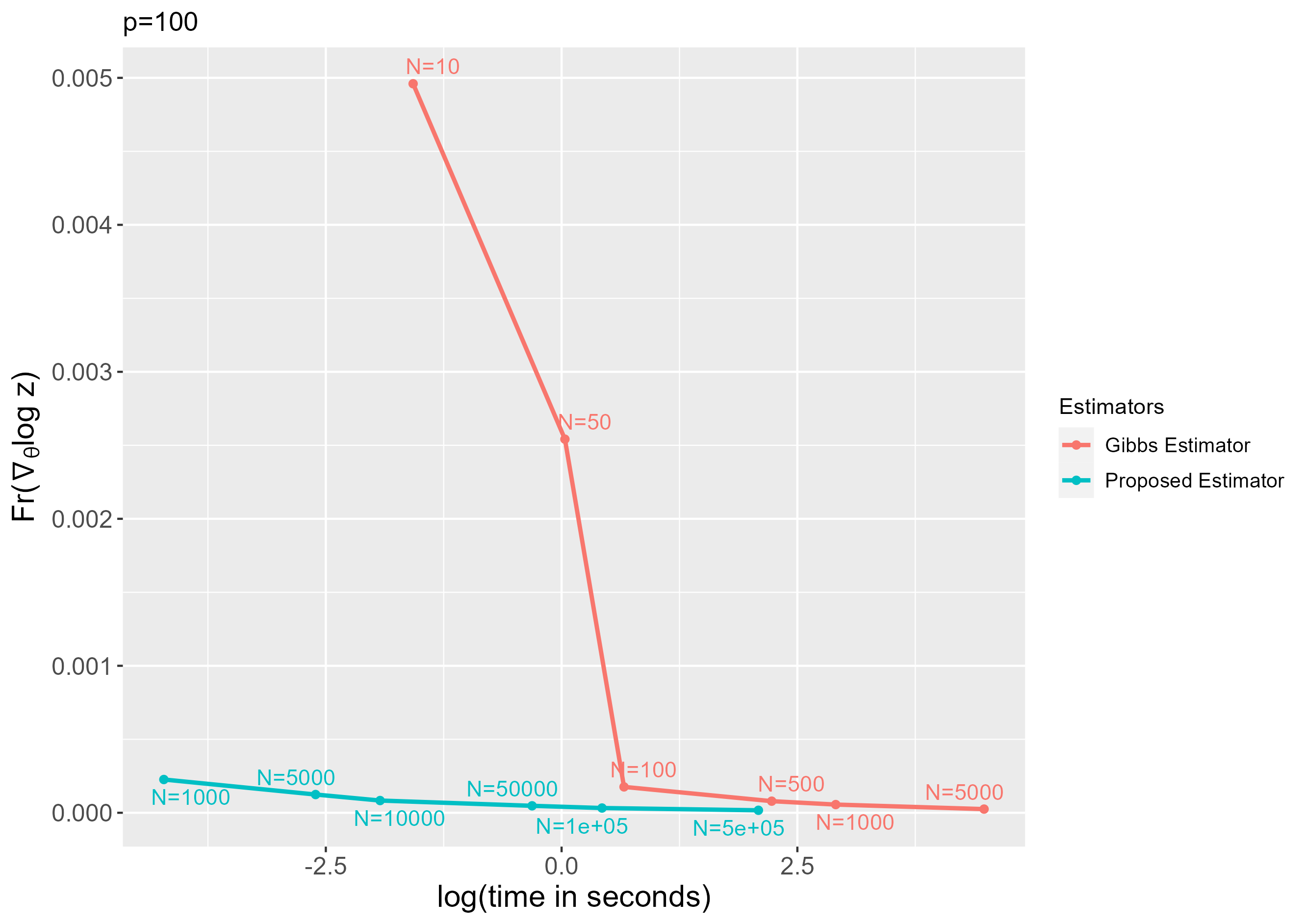}
    \caption{$\mathrm{Fr}(\nabla_\theta \log z)$ vs. log of computation time in seconds for $p=100$ using the band graph structure precision matrix.  The Monte Carlo sample size for the importance sampling estimator and the MCMC sample size for the Gibbs estimator are labeled at the data points.}
    \label{fig:comptime_GT}
\end{figure}

{\color{red}
\subsection{Choice of the importance distribution}
A crucial feature of our procedure is the choice of the importance distribution that depends on the current parameter estimate $\theta^{(t)}$ via our choice $\phi^{(t)}=\mathrm{diag}(\theta^{(t)})$. To see why any fixed proposal would not likely work, one may again look at the canonical and tractable Gaussian graphical model (GGM) case, as described in Supplementary Section~\ref{sup:mc_imp}. Recall, in this case, 
    $p_\theta(x) = q_\theta(x)/z(\theta)$, with: 
$$
q_\theta(x) = \exp\left\{-\frac{1}{2}\left(\sum_{j=1}^p \theta_{jj}x_j^2 + 2 \sum_{j< k} \theta_{jk} x_j x_k\right)\right\},
$$
  and, for $\phi^{(t)}=\mathrm{diag}(\theta^{(t)}),$ one obtains:
$$
\frac{q_{\theta^{(t)}}(x)}{q_{\phi^{(t)}}(x)}  = \exp\left\{- \sum_{j< k} \theta_{jk}^{(t)} x_j x_k\right\}.
$$   
The variance of this quantity is directly controllable as an application of Proposition~\ref{prop:exp}. However, consider now $\tilde\phi^{(t)}=\mathrm{diag}(\theta^{(0)})$, a fixed trial density. This yields:
$$
\frac{q_{\theta^{(t)}}(x)}{q_{\tilde\phi^{(t)}}(x)}  = \exp\left\{-\frac{1}{2}\left(\sum_{j=1}^p (\theta^{(t)}_{jj}- \theta^{(0)}_{jj})x_j^2 + 2 \sum_{j< k} \theta^{(t)}_{jk} x_j x_k\right)\right\}.
$$  
Now, the terms involving $x_j^2$ no longer cancel, and, recognizing this ratio itself as a Gaussian density (up to proportionality), it is clear the variance is unbounded if:
\begin{align}
\theta^{(t)}_{jj}- \theta^{(0)}_{jj} <0.\label{eq:violation}
\end{align}
Since our proximal gradient or HMC update $\theta^{(t)}\to\theta^{(t+1)}$ does nothing to prevent the violation of Equation~\eqref{eq:violation} from occurring, importance sampling with respect to any fixed proposal is likely to (eventually) lead to failure. In contrast, in our approach, when the parameter moves $\theta^{(t)}\to\theta^{(t+1)}$, then so does the parameter of the trial density $\{\phi^{(t)}=\mathrm{diag}(\theta^{(t)})\}\to \{\phi^{(t+1)}=\mathrm{diag}(\theta^{(t+1)})\}$. As a result, the target density is \emph{never too far} from our trial density in importance sampling, since the coefficients of the self terms $x_j$ always match, and only those of the interaction terms $x_jx_k$ differ between target and trial. In this way, Proposition~\ref{prop:exp} always remains applicable.

}

{\color{red} 
\subsection{Sensitivity analysis}
In this section, we conduct a sensitivity analysis of the proposed estimator against two critical assumptions. Specifically, finite variance of the proposed importance sampling estimator under assumptions detailed in Proposition \ref{prop:exp} and the sparsity assumption in Proposition \ref{prop:importance_sample_size} which enables estimation with reasonable Monte Carlo sample size $N$.\\

\noindent {\bf The assumptions in Proposition \ref{prop:exp}:} We present simulation results for the Gaussian graphical model, for which $z(\theta)$ is analytically available. Specifically, we consider a multivariate Gaussian distribution with a banded 
precision matrix $\theta(\rho)$, where the diagonal entries are 3, the 
first off-diagonal entries are $\rho$, and the second off-diagonal entries 
are $-\rho$, with $p = 50$. For each $\theta$, we compute the smallest eigenvalue $\lambda_{\min}$ of the 
matrix $U(\theta) = 2\theta - \mathrm{diag}(\theta)$, which governs whether the 
assumption of Proposition~\ref{prop:exp} is satisfied. We classify each configuration into 
one of two regimes: \textit{Not Violated}, when $\lambda_{\min} > 0.0001$; and
\textit{Violated}, when $\lambda_{\min} < -0.0001$. Figure~\ref{fig:assumption_check_p2} demonstrates that as $\rho$ varies from  $-1.2$ to $0.6$ in increments of $0.2$, both violated and non-violated regimes are encountered. For each setting, we repeated the estimation of $z(\theta)$ 100 times and computed the $\log(\text{MSE})$. As shown in the figure, the $\log(\text{MSE})$ is substantially smaller in the \emph{Not Violated} regime, confirming that the accuracy of our estimator degrades when the assumption of Proposition~\ref{prop:exp} is not satisfied.

    \begin{figure}[!t]
\centering
   \includegraphics[width =0.9\textwidth]{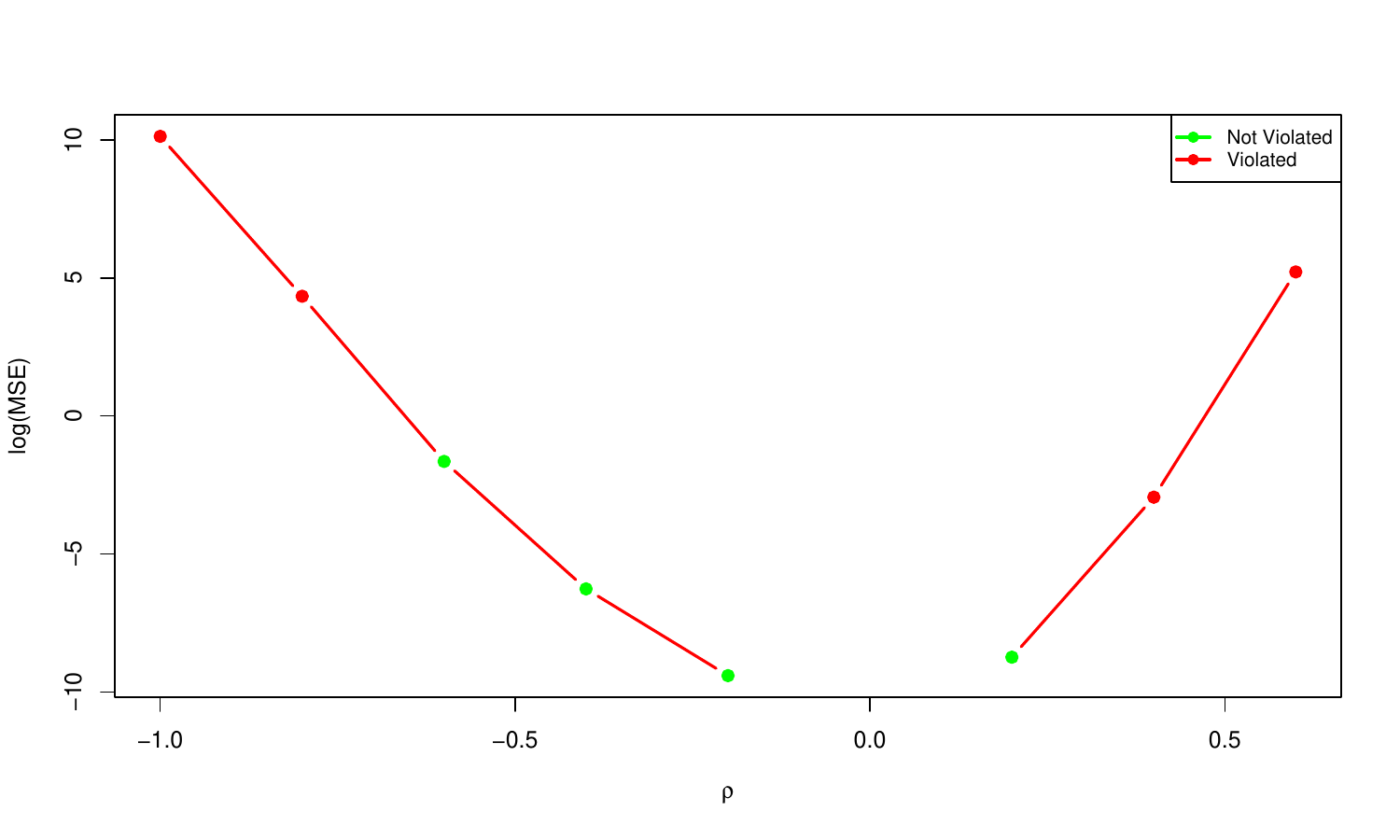}
    \caption{$\log(\text{MSE})$ of $\hat{z}(\theta)$ across Varying non-zero off-diagonal values $\rho$ for a banded Gaussian graphical model.}
    \label{fig:assumption_check_p2}
\end{figure}

However, while the assumption is crucial, it is also mild. For example, 
     \begin{enumerate}
\item For Ising models and Boltzmann machines, the parameter space is unrestricted, i.e., $\Theta\in\mathbb{R}^{p\times p}$, and thus, the condition is always satisfied. For the same reason, the condition is always met in the Potts model \citep{potts1952some}.

\item For Poisson graphical model, the restriction is that the off-diagonal terms are non-positive \citep{besag1974spatial}. Thus, if $u+\phi=\theta \in \Theta$ is such that $u_{jk}\le 0$, then clearly, $2u_{jk}\le 0$, and the condition is always met for this model too.
\end{enumerate}

\noindent {\bf The assumptions in Proposition \ref{prop:importance_sample_size}:} We conduct an additional Gaussian simulation in which sparsity is controlled by the band width of the precision matrix, defined as the number of nonzero off-diagonal bands on each side of the main diagonal. Larger band width corresponds to a denser precision structure. We vary the band width from 1 to 9 and evaluate the proposed estimator over Monte Carlo sample size $N \in \{10^3, 5\times10^3,\, 10^4,\,5\times10^4,\, 10^5, \,5\times10^5\}$, using MSE as the accuracy metric. As shown in Table \ref{tab:sparsity}, MSE of $\frac{\hat{z}(\theta)}{z(\phi)}$ decreases consistently as (N) increases for all band width settings, while at fixed (N), larger band width generally yields higher error. These results support the expected sparsity effect, i.e., denser precision matrices are more challenging and require larger sample sizes to achieve comparable estimation accuracy.

\begin{table}[h]
\centering
\scalebox{0.8}{
\begin{tabular}{c|cccccc}
\hline
 \text{Band width} & $N=1000$ & $N=5000$ & $N=10000$ & $N=50000$ & $N=1\times10^5$ & $N=5\times 10^5$ \\
\hline
$1$ & $3.18\times10^{-6}$ & $4.48\times10^{-7}$ & $2.52\times10^{-7}$ & $5.67\times10^{-8}$ & $2.77\times10^{-8}$ & $4.51\times10^{-9}$ \\
$2$ & $4.49\times10^{-6}$ & $8.16\times10^{-7}$ & $3.56\times10^{-7}$ & $1.21\times10^{-7}$ & $2.82\times10^{-8}$ & $1.02\times10^{-8}$ \\
$3$ & $5.61\times10^{-6}$ & $1.28\times10^{-6}$ & $5.83\times10^{-7}$ & $9.98\times10^{-8}$ & $5.90\times10^{-8}$ & $1.52\times10^{-8}$ \\
$4$ & $6.31\times10^{-6}$ & $1.47\times10^{-6}$ & $8.56\times10^{-7}$ & $1.57\times10^{-7}$ & $6.20\times10^{-8}$ & $1.45\times10^{-8}$ \\
$5$ & $8.02\times10^{-6}$ & $1.62\times10^{-6}$ & $8.41\times10^{-7}$ & $1.59\times10^{-7}$ & $7.70\times10^{-8}$ & $1.87\times10^{-8}$ \\
$6$ & $9.47\times10^{-6}$ & $1.61\times10^{-6}$ & $1.11\times10^{-6}$ & $2.72\times10^{-7}$ & $8.16\times10^{-8}$ & $1.55\times10^{-8}$ \\
$7$ & $1.09\times10^{-5}$ & $1.60\times10^{-6}$ & $9.41\times10^{-7}$ & $1.93\times10^{-7}$ & $9.35\times10^{-8}$ & $1.34\times10^{-8}$ \\
$8$ & $1.08\times10^{-5}$ & $1.75\times10^{-6}$ & $9.69\times10^{-7}$ & $1.61\times10^{-7}$ & $7.71\times10^{-8}$ & $1.78\times10^{-8}$ \\
$9$ & $1.17\times10^{-5}$ & $1.80\times10^{-6}$ & $9.92\times10^{-7}$ & $2.33\times10^{-7}$ & $1.02\times10^{-7}$ & $2.39\times10^{-8}$ \\
\bottomrule
\end{tabular}
}
\caption{MSE of proposed importance sampling estimator of $z(\theta)$ in the Gaussian case for $p=10$. Here, we vary the sparsity of the precision matrix $\theta$ by controlling the number of bands around the main diagonal. The importance sample size is denoted by $N$. }
\label{tab:sparsity}
\end{table}
}

\subsection{Likelihood vs. pseudo likelihood-based CIs in low dimensions}\label{sec:LD_simulations}
Results in Table \ref{tab:CI_width} substantiate that likelihood-based inference provides the best results in low dimensions ($p = 3, 5$). 
 Specifically, we consider a dense true parameter $\theta_0$, where each element is set to $\theta_{0, jk} = -0.5$ for the PGM and $\theta_{0, jk} = -1$ for the Ising model. For each $\theta_0$, we generate datasets with a sample size of 100 ($n = 100$) and construct 95\% bootstrap confidence intervals using 500 bootstrap samples. The average coverage and width of these intervals are reported over 100 replications of this process.
Results in Table \ref{tab:CI_width} indicate that while both maximum likelihood estimate (MLE) and pseudo maximum likelihood estimate (MPLE) have comparable coverage probabilities for the PGM, widths of confidence intervals higher for PMLE, often by a factor as large as 4. For the Ising model, MLE significantly outperforms PMLE in terms of both average coverage rate and confidence interval length. Notably, no penalization was used in low dimensions, either for full likelihood or for pseudolikelihood.
\begin{table}[h!]
    \centering
    \scalebox{.65}{
    \begin{NiceTabular}{ccc||cc|cc||cc|cc}
    \toprule
    & & & \multicolumn{4}{c||}{Ising} & \multicolumn{4}{c}{PGM}\\
    \midrule \addlinespace
        Setting & $p$ & $n$ & \multicolumn{2}{c|}{Coverage} & \multicolumn{2}{c||}{Width} & \multicolumn{2}{c|}{Coverage} & \multicolumn{2}{c}{Width}  \\ \addlinespace
        \toprule
        & & & MLE & MPLE  & MLE & MPLE & MLE & MPLE  & MLE & MPLE \\ \addlinespace
       \midrule
        \multirow{2}{*}{LD} & 3 & \multirow{2}{*}{100}  &0.92 &  0.77 &  0.92 & 3.37 &  0.94 &  0.92 &  1.42 &  2.11\\
         & 5 &  &  0.93 & 0.61 &   0.88 &  3.68 &  0.94 &  0.91 &  1.11 &   4.45 \\
         \bottomrule
    \end{NiceTabular}
    
    }
    \caption{Coverage and width of confidence intervals of maximum likelihood (MLE) and maximum pseudo-likelihood (MPLE) procedures for two PEGMs: Ising and PGM.}
    \label{tab:CI_width}
\end{table}

\subsection{A comparison of RBM and BM representational power}\label{sec:supp_bm}
Existing literature on training performance of the CD algorithm and its other versions focus mostly on benchmark datasets. Nonetheless, in this section we report the results of CD training versus the proposed approach in controlled simulation experiments. 
We compare the quality of the corresponding density estimate in terms of the total variation distance: $\text{TV}(p_{\theta_0}, p_{\theta}) = \sum_{\v \in \{0,1\}^p} |p_{\theta_0}(\v) - p_{\hat{\theta}}(\v)|$ where the dimension of $\theta_0$ is the data generating value of the parameter and $\hat{\theta}$ is the estimate obtained from the two methods, noting that the dimensions of $\theta_0$ and $\hat{\theta}$ might be different given that the number of hidden variables is unknown in reality. We consider two cases $(p,m) = (2, 2)$ and $(p, m) = (3, 4)$ and set the sample size $n = 1000$. Elements of the true data generating $\theta_0$ are set to 0 with probability 0.5 and are drawn from $\text{Uniform}[-1,1]$ with probability 0.5, respecting the bipartite structure of the graph. When implementing CD we set $k =1$, although we did not notice any significant difference in results for $k$ up to 10. We report in Table~\ref{tab:FLvCD} results averaged over 50 replications where we abbreviate the proposed method here by FL (for \emph{full likelihood}). All hyperprameters for training for the two algorithms are kept same, including the step size and convergence criteria. Two key observations can be made from the results. First, for increasing values of the number of hidden variables $m$, the fit is better, which is expected. Second, the performance of FL is better compared to CD in almost all cases, albeit by a small margin. 

We next consider experiments where we study the representational power of RBMs to that of BMs. Since both method learn distributions of visible variables, we want to investigate whether allowing connections within layers as in BM (see also Figure \ref{fig:BM}), leads to similar or better learning power for a smaller number of hidden variables. For this we consider two situations $p = 3, 5$ which is the dimension of the visible variable. We generate the visible variables from a multivariate probit model where $y_j = \mathbb{I}(z_j >0)$ and $z = (z_1, \ldots, z_p)^\T \sim \Gauss(0, \Sigma)$. We set $\Sigma = (\Sigma)_{ij}$ with $\Sigma_{ii} = 1$ and $\Sigma_{ij} = 0.5$, $i \neq j = 1, \ldots, p$. This model is often used to capture dependent multivarite binary data \citep{ashford1970multi, chakraborty2023bayesian}. We fix the sample size of the visible variables to $n = 1000$, and fit RBM and BM to the data letting the number of hidden variables $m$ vary between $6$--$12$ for $p = 3$, and $10$--$18$ for $p = 5$, with increments of 1. Ideally, with enough hidden variables, RBMs should be able to approximate the data generating distribution \citep{le2008representational}. We show the total variation distance between the true distribution and the fitted distribution from both models in Figure \ref{fig:RBMvBM1}, which reveals that BM achieves better approximation to the distribution of the visible variables with a fewer number of hidden variables, at least for this data generating distribution. Indeed, the gain in approximation for BM is almost 50\% for $p = 3$ and 33\% for $p = 5$. The trade-off is the computation time. The key bottleneck of training a full BM is the fact that the conditional distribution of the hidden variables depend on the corresponding visible variables which is not the case for RBMs. This allows for much faster training of RBMs. While a theoretical investigation into the representational power of BM is beyond the scope of this work, we hope that our numerical experiments serve as a platform for further investigation into the properties of a full BM.
\begin{table}
    \parbox{0.45 \textwidth}{
    \centering
    \begin{NiceTabular}[b]{c|c|ccc}
    \hline
    & & $m = 4$ & $m = 5$ & $m= 6$\\
    \cmidrule{1-5}
        \multirow{2}{*}{$p = 2$} & FL & 0.70 & 0.62 & 0.56 \\
         & CD & 0.72 & 0.65 & 0.57\\
         \hline
         & & $m = 16$ & $m = 17$ & $m= 18$\\
    \cmidrule{1-5}
        \multirow{2}{*}{$p = 3$} & FL & 3.66 & 3.56 & 3.54  \\
         & CD & 3.66 & 3.59 & 3.58\\
         \hline
    \end{NiceTabular}

    }
    \hfill 
    \parbox{0.45 \textwidth}{
    \centering
    \begin{NiceTabular}[b]{c|c|cc}
    \toprule
        & $m$ & RBM& BM  \\
        \hline \\
         \multirow{4}{*}{TV} & 20 & 60.29 & 48.91\\
        & 30 &  58.03 & 48.55\\
        & 40 &  55.38 & 46.09\\
        & 50 & 52.81& 43.66\\
        \bottomrule
    \end{NiceTabular}   
    }
      
    \caption{\emph{(Left)} Training performance of RBMs using the proposed method (FL)and CD in terms of the estimated density of the visible variables with $p=2, 3$ where we report the total variation distance between the true distribution and the estimated distribution. \emph{(Right)} Results from a high-dimensional experiment (visible variable has dimension $p=100$) where we report the total variation distance on a held-out test set. Throughtout, $m$ denotes the number of hidden variables.} 
    \label{tab:FLvCD}
    \end{table}
\begin{figure}
   \begin{subfigure}{0.48\textwidth}
       \includegraphics[height = 5cm, width = 7cm]{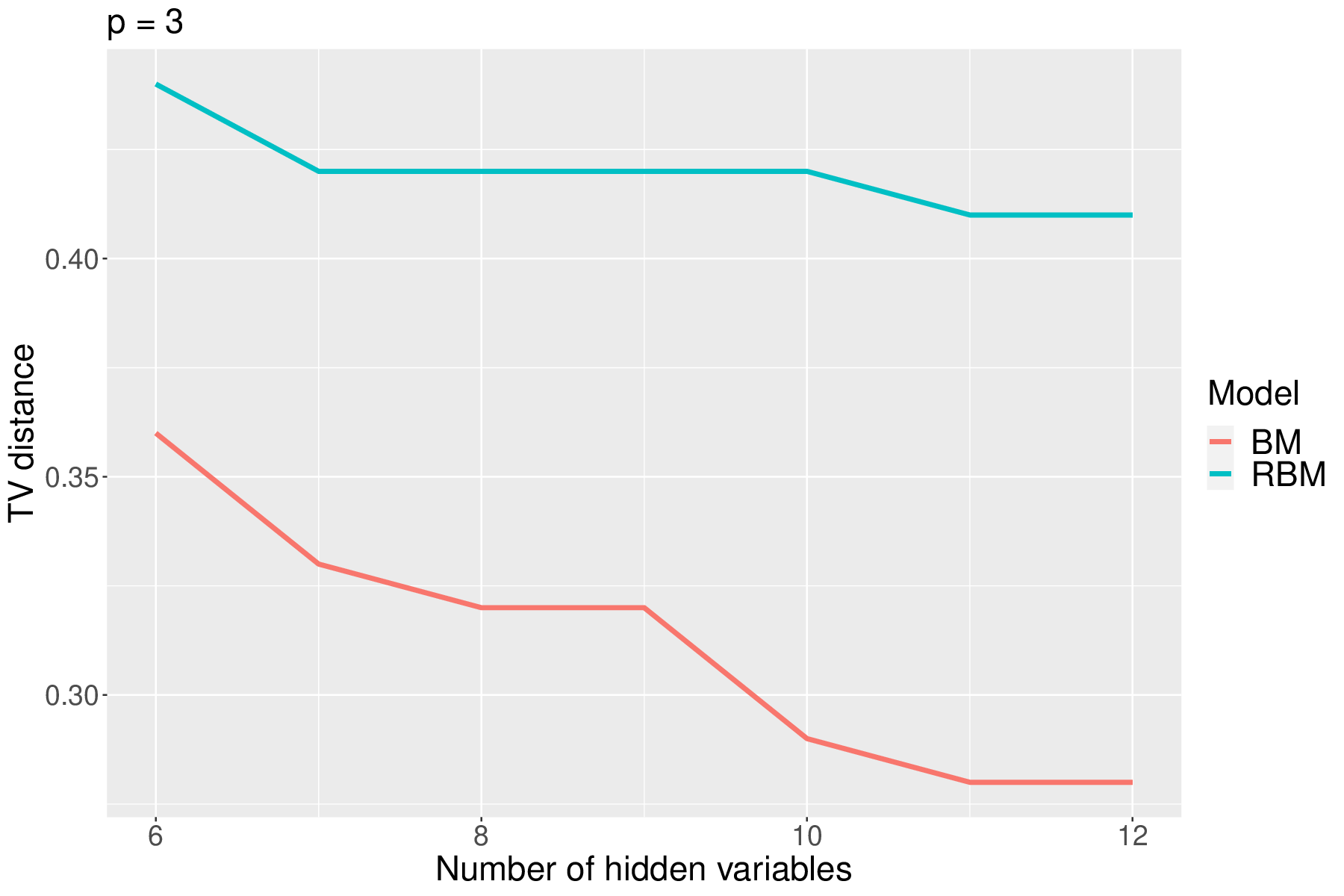}
   \end{subfigure}
   \begin{subfigure}{0.48\textwidth}
       \includegraphics[height = 5cm, width = 7cm]{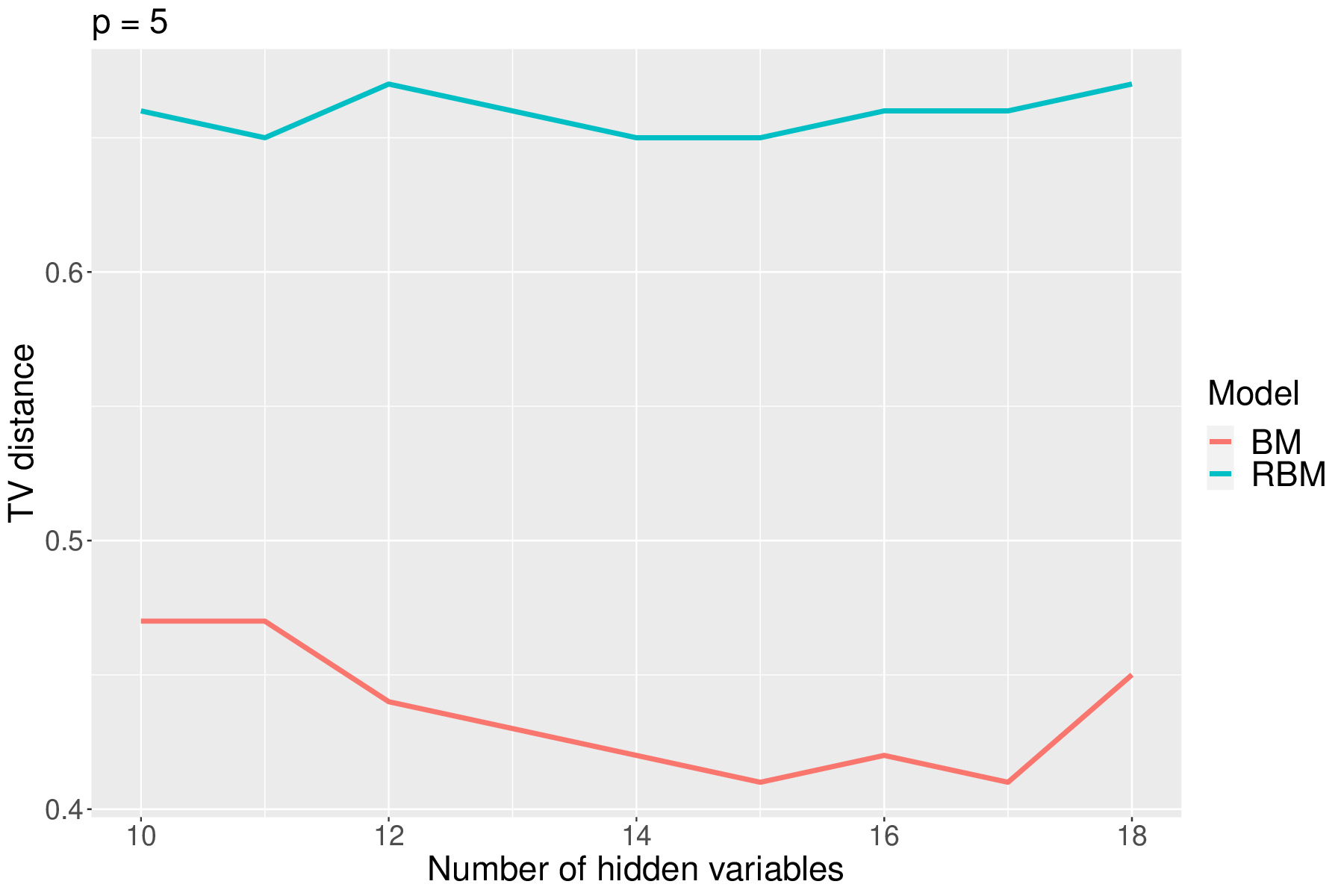}
   \end{subfigure}
    \caption{RBM vs BM total variation distance when the dimension of the visible variables is $p = 3$ and $5$.}
    \label{fig:RBMvBM1}
\end{figure}

Next, we report our results from a high-dimensional experiment comparing the representational power of RBM and BM with different choices of the number of hidden variables. Here, we consider $p = 100$, and let $m = 20, 30, 40, 50$.The observed data is generated from a probit model as before. Estimating the total-variation distance for this very large sample space is computationally prohibitive. Instead, we compare the fit on a test set of size $n_t = 500$. Specifically, we compute $\sum_{\v \in \mathcal{T}} |p(\v) - p_{\hat{\theta}}(\v)| = \text{TV}$, where $\mathcal{T}$ is the set of test points and $p_{\hat{\theta}}(\v)$ is computed following Remark \ref{rm:BM_marginal} in the main document.  We randomly selected $50$ configurations of the hidden variables from $\{0,1\}^m$. Using each of these configurations, we obtained an estimate of $p_{\hat{\theta}}(\v)$ for $\v \in \mathcal{T}$. A final estimate is constructed using the average of these estimates, which is then used to compute $\text{TV}$ as defined above.  In the right panel of Table \ref{tab:FLvCD} we report the results, which reconfirms that for the same number of hidden variables, the representational power of BM seems to be better than that of an RBM. 
{\color{red} 
\subsection{Comparison of the proposed method with alternatives for unnormalized models }\label{sec:new_experiments}
In Section \ref{sec:sims}, we compared the proposed method against pseudo-likelihood. Here, we compare it against other approaches for unnormalized models and the results are summarized in Table~\ref{tab:Ising_comparison}.  Specifically, we consider annealed importance sampling or AIS \citep{neal2001annealed} and noise contrastive estimation or NCE \citep{gutmann2012noise} and focus on the Ising model. AIS is another importance samplig based Monte Carlo approach for inference in these models where a sequence of importance distributions is used to approximate $z(\theta)$ and $\nabla_\theta z(\theta)$. In this sense, AIS is also a likelihood-based approach. As a result, performance of AIS and the proposed method is broadly comparable. However, the key improvement is in runtime of the two approaches, which is also reported. On the other hand, in NCE, $z(\theta)$ is treated as an unknown parameter which is estimated through a careful choice of noise distribution, and thus, NCE is not a likelihood-based approach. 

For AIS, we use an annealing schedule with $T = 20$ intermediate distributions that interpolate between a tractable initial distribution, taken to be the independent model corresponding to a diagonal parameter matrix, and the target distribution parameterized by $\theta$. At each intermediate stage, one full Gibbs sweep is performed to generate the next sample in the chain. As with the proposed method, we incorporate an $\ell_1$ penalty and report two performance measures: the Frobenius norm at the regularization parameter selected via cross-validation, and the MCC based on stable selection. For NCE, we use a noise distribution consisting of independent Bernoulli variables with success probability $0.5$ across all $p$ dimensions, and set the number of noise samples equal to the observed sample size $n$. 

Overall, the proposed approach is consistently the best performer statistically, while also being comparable in speed to NCE, which performs substantially worse.

  \begin{table}[!t]
    \centering
    \scalebox{.7}{
    \begin{NiceTabular}{ccc||ccc|cc|ccc}
    \toprule
    \midrule \addlinespace
        Setting & $p$ & $n$ & \multicolumn{3}{c|}{ $||\theta_0 - \hat{\theta}||_F^2$} & \multicolumn{2}{c}{MCC}  & \multicolumn{3}{c}{Run Time (s.)} \\ \addlinespace
        \toprule
        & & & $PMLE_{IM}$ & $PMLE_{AIS}$ & NCE & $PMLE_{IM}$ & $PMLE_{AIS}$ & $PMLE_{IM}$ & $PMLE_{AIS}$ & NCE \\ \addlinespace
       \midrule
        \multirow{2}{*}{LD} & 3 & \multirow{2}{*}{100} & $\mathbf{0.57 (0.11)}$ & $\mathbf{0.57 (0.11)}$ & 0.84 (0.11) &   -- & -- &  0.13  & 15.75  & $\mathbf{0.03}$ \\
         & 5 &  & $\mathbf{1.22 (0.10)}$ & 1.23 (0.10) & 1.89 (0.06) & -- & -- &   0.17  & 28.96  & $\mathbf{0.04}$ \\
         \midrule \addlinespace
         \multirow{2}{*}{HD} & 20 &\multirow{4}{*}{100} & $\mathbf{9.18 (0.61)}$ & 9.41 (0.42) & 13.42 (0.02)  & $\mathbf{0.67 (0.05)}$ & 0.53 (0.06) &  1.45 & $4\times10^2$ & $\mathbf{0.09}$  \\
          & 30 &  & $\mathbf{12.85 (0.99)}$ & 13.72 (0.57) & 18.55 (0.02) & $\mathbf{0.78 (0.02)}$ & 0.53 (0.02) &  2.20 & $6\times10^2$ & $\mathbf{0.16}$  \\
          & 40 & & $\mathbf{18.61} (1.31)$  & 20.69 (0.48) & 26.46 (0.07) & $\mathbf{0.71 (0.05)}$ & 0.55 (0.03) &  2.68 & $1\times10^4$ & $\mathbf{0.17}$ \\
          & 50 & & $\mathbf{26.89 (0.50)}$  & 29.29 (0.29) &  33.80 (0.02) & $\mathbf{0.71 (0.03)}$ & 0.46 (0.02) &   3.11 & $2\times10^4$ & $\mathbf{0.26}$ \\
          \midrule \addlinespace
         \multirow{2}{*}{UHD} & 100 & 200 & $\mathbf{54.35 (0.93)}$ & 63.33 (0.06) & 67.83 (0.01)  &  $\mathbf{0.80 (0.02)}$ & 0.69 (0.02) &   5.08 & $5\times10^4$ & $\mathbf{1.37}$  \\
          & 200 & 400& $\mathbf{128.97 (0.14)}$ & -- & 134.33 (0.01) & $\mathbf{0.70 (0.01)}$ & -- &   14.19 & -- & $\mathbf{8.61}$  \\
         & 300 & 600 & $\mathbf{195.97 (0.06)}$ & -- & 198.36 (0.01) & $\mathbf{0.56 (0.01)}$ & -- &   30.65 & -- &   $\mathbf{19.74}$ \\
         \bottomrule
    \end{NiceTabular}
    
    }
    \caption{Mean (sd) of MSE = $||\hat{\theta} - \theta_0||_F^2$, Matthews correlation coefficient (MCC) and run time (seconds) across 10 data sets for penalized maximum likelihood estimation (PMLE) under the proposed independence model (IM) and annealed importance sampling (AIS), and  noise contrastive estimation (NCE). Best performers are bolded.}
    \label{tab:Ising_comparison}
\end{table}
}
\section{Additional Data Analysis Results}
\subsection{Additional results for movie ratings network}

This section presents additional results from fitting the Ising model to the $Movielens$ dataset using the PMLE-based method described as in Section \ref{sec:Ising_movie}. Our primary focus is to reveal the connections among the most frequently rated 50 movies. The resulting graph structure is shown in Figure \ref{fig:Ising_movie_sw_lr}, left panel. Table \ref{tab:movie_id} provides the match between abbreviations and the legends of the movie titles, genres, and years of release. We also present the top 13 most connected movies from the resulting graph structure in Table \ref{tab:movies_top}. The top three of them are from two well known movie franchises. These two movie franchises (\emph{Star Wars} and \emph{Lord of the Rings}) also form two separate cliques within the graph, as shown in Figure \ref{fig:Ising_movie_sw_lr}, right panel. A heatmap of estimated $\hat{\theta}_\lambda$ with the penalty parameter $\lambda = 10$ is in Figure \ref{fig:Ising_movie_estimated_theta}, where red and blue denote positive and negative connections.

\begin{table}[]
    \centering
    \scriptsize
    
    \scalebox{0.8}{
    \begin{tabular}{c|c|c}
    \hline
 \textbf{Abbr.} & \textbf{Title (Year)} & \textbf{Genres} \\
\hline

  TS  & Toy Story (1995) & Adventure$|$Animation$|$Children$|$Comedy$|$Fantasy \\
 
Br & Braveheart (1995) & Action$|$Drama$|$War \\
 
MP & Monty Python and the Holy Grail (1975) & Adventure$|$Comedy$|$Fantasy \\

SW-V & Star Wars: Episode V - The Empire Strikes Back (1980) & Action$|$Adventure$|$Sci-Fi \\
 
PB
 & Princess Bride (1987) & Action$|$Adventure$|$Comedy$|$Fantasy$|$Romance \\
 
RL & Raiders of the Lost Ark  (1981)  & Action$|$Adventure\\

SW-VI & Star Wars: Episode VI - Return of the Jedi (1983) & Action$|$Adventure$|$Sci-Fi\\

Te & Terminator(1984) 
 & Action$|$Sci-Fi$|$Thriller\\
 
GD & Groundhog Day (1993) & Comedy$|$Fantasy$|$Romance \\

BF & Back to the Future (1985) & Adventure$|$Comedy$|$Sci-Fi \\

IJ & Indiana Jones and the Last Crusade (1989) & Action$|$Adventure \\

A13 & Apollo 13 (1995) & Adventure$|$Drama$|$IMAX \\

MB & Men in Black (1997) &Action$|$Comedy$|$Sci-Fi \\

GW & Good Will Hunting (1997) & Drama$|$Romance \\

Ti & Titanic (1997) & Drama$|$Romance \\

SPR & Saving Private Ryan (1998) & Action$|$Drama$|$War \\

Ma & Matrix (1999) & Action$|$Sci-Fi$|$Thriller \\

SW-IV & Star Wars: Episode IV - A New Hope (1977) & Action$|$Adventure$|$Sci-Fi \\

SS  & Sixth Sense (1999) & Drama$|$Horror$|$Mystery \\

AB & American Beauty (1999) & Drama$|$Romance \\

FC & Fight Club (1999) & Action$|$Crime$|$Drama$|$Thriller \\

PF & Pulp Fiction (1994) & Comedy$|$Crime$|$Drama$|$Thriller \\

SR  & Shawshank Redemption (1994) & Crime$|$Drama \\

TM & Twelve Monkeys  (1995) & Mystery$|$Sci-Fi$|$Thriller \\

AV & Ace Ventura: Pet Detective (1994) & Comedy \\

FG & Forrest Gump (1994) & Comedy$|$Drama$|$Romance$|$War \\

Gl & Gladiator (2000) & Action$|$Adventure$|$Drama \\

LK & Lion King (1994) & Adventure$|$Animation$|$Children$|$Drama$|$Musical$|$IMAX \\

Sp & Speed (1994) & Action$|$Romance$|$Thriller \\

TL & True Lies (1994) & Action$|$Adventure$|$Comedy$|$Romance$|$Thriller \\

Me  & Memento (2000) & Mystery$|$Thriller \\

Sh & Shrek (2001) & Adventure$|$Animation$|$Children$|$Comedy$|$Fantasy$|$Romance \\

Fu  & Fugitive (1993) & Thriller \\

Se & Seven (1995) & Mystery$|$Thriller \\

JP & Jurassic Park (1993) & Action$|$Adventure$|$Sci-Fi$|$Thriller \\

LR1 & Lord of the Rings: The Fellowship of the Ring (2001) & Adventure$|$ Fantasy \\

US & Usual Suspects, The (1995) & Crime$|$Mystery$|$Thriller \\

SL  & Schindler's List (1993) & Drama$|$War \\

DK & Dark Knight (2008) & Action$|$Crime$|$Drama$|$IMAX \\

Al & Aladdin (1992) & Adventure$|$Animation$|$Children$|$Comedy$|$ Musical\\

T2  & Terminator 2: Judgment Day (1991) & Action$|$Sci-Fi\\

DW & Dances with Wolves (1990) & Adventure$|$Drama$|$Western \\

Ba & Batman (1989) & Action$|$Crime$|$Thriller \\

SiL  & Silence of the Lambs (1991) & Crime$|$Horror$|$Thriller \\

LR2 & Lord of the Rings: The Two Towers (2002) & Adventure$|$Fantasy \\

Fa & Fargo (1996) & Comedy$|$Crime$|$Drama$|$Thriller \\

LR3  & Lord of the Rings: The Return of the King (2003) &Action$|$Adventure$|$Drama$|$Fantasy \\

ID & Independence Day (1996) & Action$|$Adventure$|$Sci-Fi$|$Thriller\\

In & Inception (2010) & Action$|$Crime$|$Drama$|$Mystery$|$Sci-Fi$|$Thriller$|$IMAX \\

God & Godfather (1972) & Crime$|$Drama \\

    \hline
    \end{tabular}
    }
    \caption{Abbreviations for the movie data set.}
    \label{tab:movie_id}
\end{table}

\begin{table}[]
    \centering
    \begin{tabular}{c|c}
    \hline
    \textbf{Title (Year)} & \textbf{Number of Degree} \\
    \hline
    Star Wars: Episode IV - A New Hope (1977) & 20 \\
    Star Wars: Episode V - The Empire Strikes Back (1980) & 20 \\
    Lord of the Rings: The Fellowship of the Ring (2001) & 19 \\
    Fight Club (1999) & 17 \\
    Memento (2000)  & 17 \\
    Independence Day (1996) & 16 \\
    American Beauty (1999) & 16 \\
    Inception (2010) & 16 \\
    Twelve Monkeys (1995) & 15 \\
    Pulp Fiction (1994) & 15 \\	
    Speed (1994) & 15 \\
    True Lies (1994)  & 15 \\
    Titanic (1997) & 15 \\
    \hline
    \end{tabular}
    \caption{Movies with the highest degree of connectivity: top 13.}
    \label{tab:movies_top}
\end{table}

\begin{figure}[h!]
\begin{subfigure}[t]{0.6\textwidth}
    \centering
\includegraphics[width = 9cm, height = 10cm]{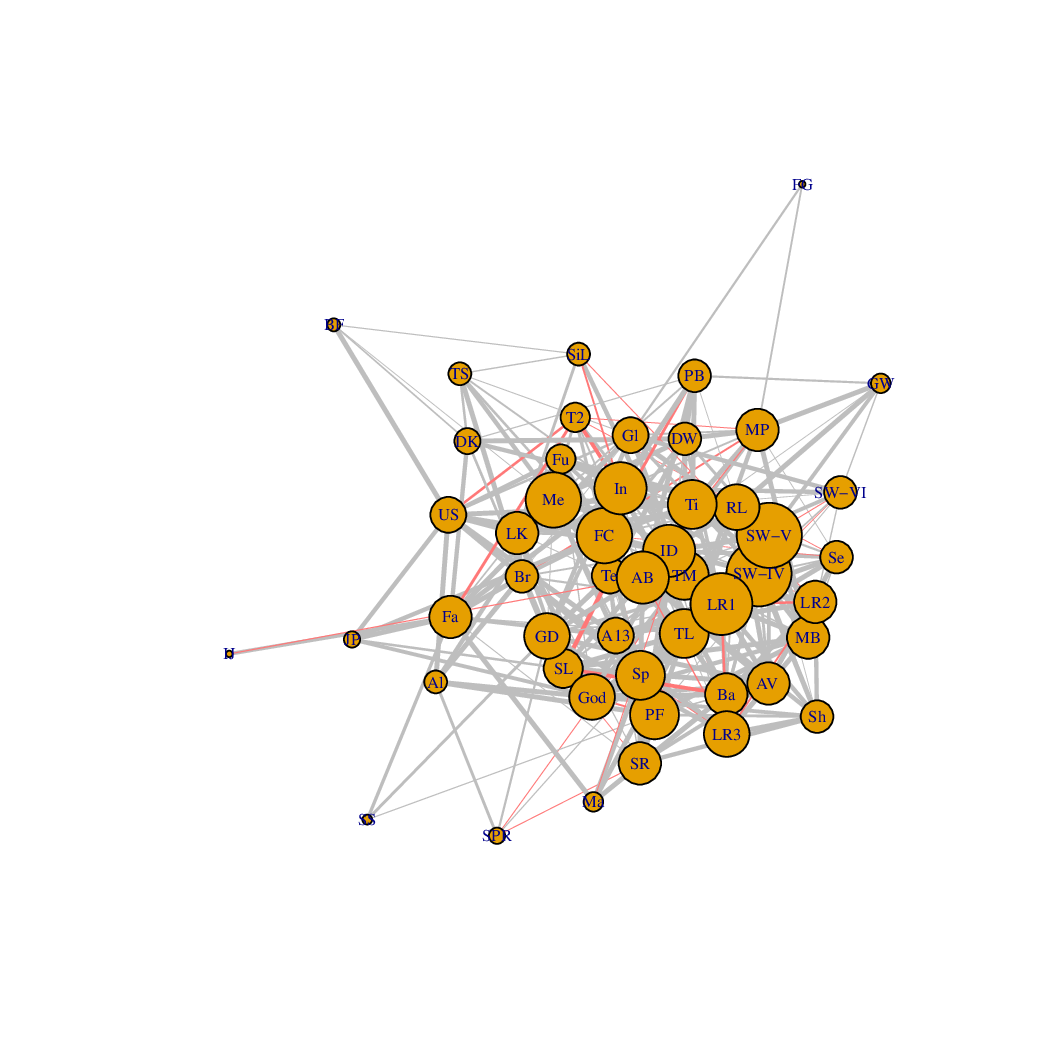}

\end{subfigure}
\begin{subfigure}[t]{0.35\textwidth}
\centering
\includegraphics[width=6cm, height = 7cm]{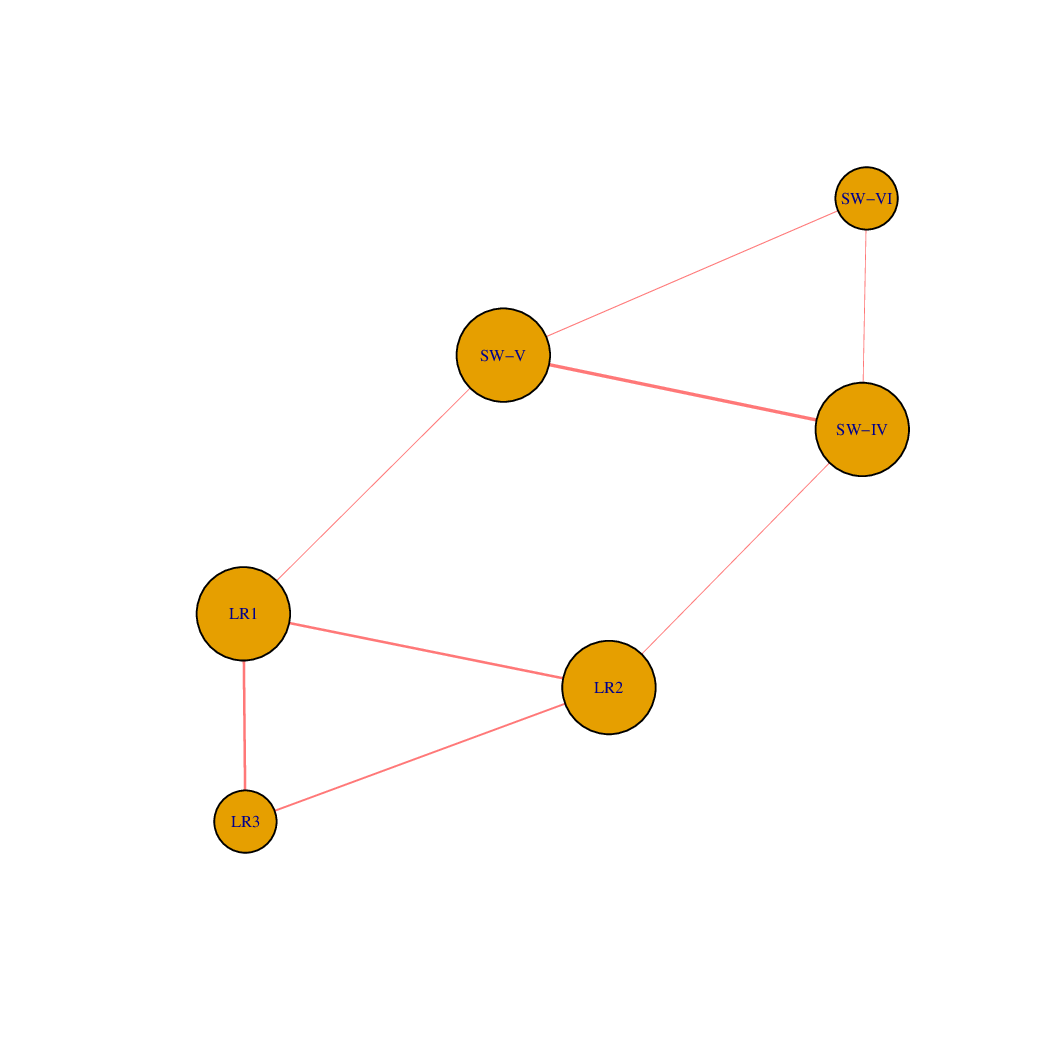}
\end{subfigure}
 \caption{(\emph{Left}). The movie ratings network from the PMLE-based Ising model. The edge widths are proportional to the stable selection probability, and the node sizes are proportional to the degree. Red and gray denote positive and negative edges. (\emph{Right}). Extracted subgraph of the movie ratings network with \textit{Star Wars} and \textit{Lord of the Rings} franchises. }
 \label{fig:Ising_movie_sw_lr}
\end{figure}


\begin{figure}[h!]
\centering
\includegraphics[width=1\textwidth]{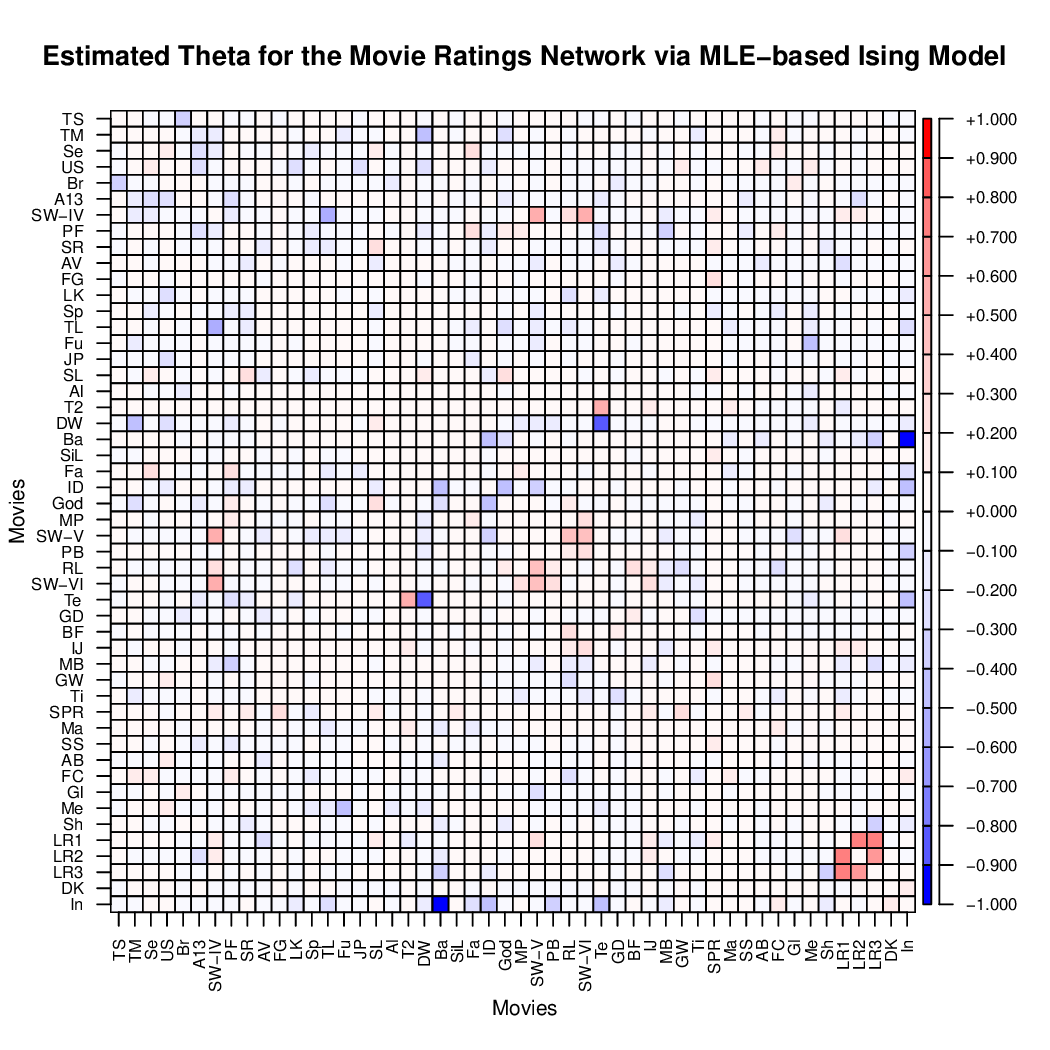}
    \caption{Parameter estimates for the movie ratings network.}
\label{fig:Ising_movie_estimated_theta}
\end{figure}





\clearpage

\subsection{Additional results for breast cancer network}

Additional results on the breast cancer miRNA networks from the full-likelihood fit and pseudo-likelihood fit of PGM discussed in Section \ref{sec:PGM_breast} are presented below. The dependency structures among miRNA expression profiles of 353 genes modeled by the aforementioned methods are comparable, as shown in Figures \ref{fig:PGM_breast_cancer_graph_MLE_neg} and~\ref{fig:PGM_breast_cancer_graph_MPLE_neg}. These two figures contain subgraphs of the resulting networks of the two methods. Common edges in from the two networks are colored blue. 


\begin{figure}[h!]
\centering
\includegraphics[width=0.9\textwidth]{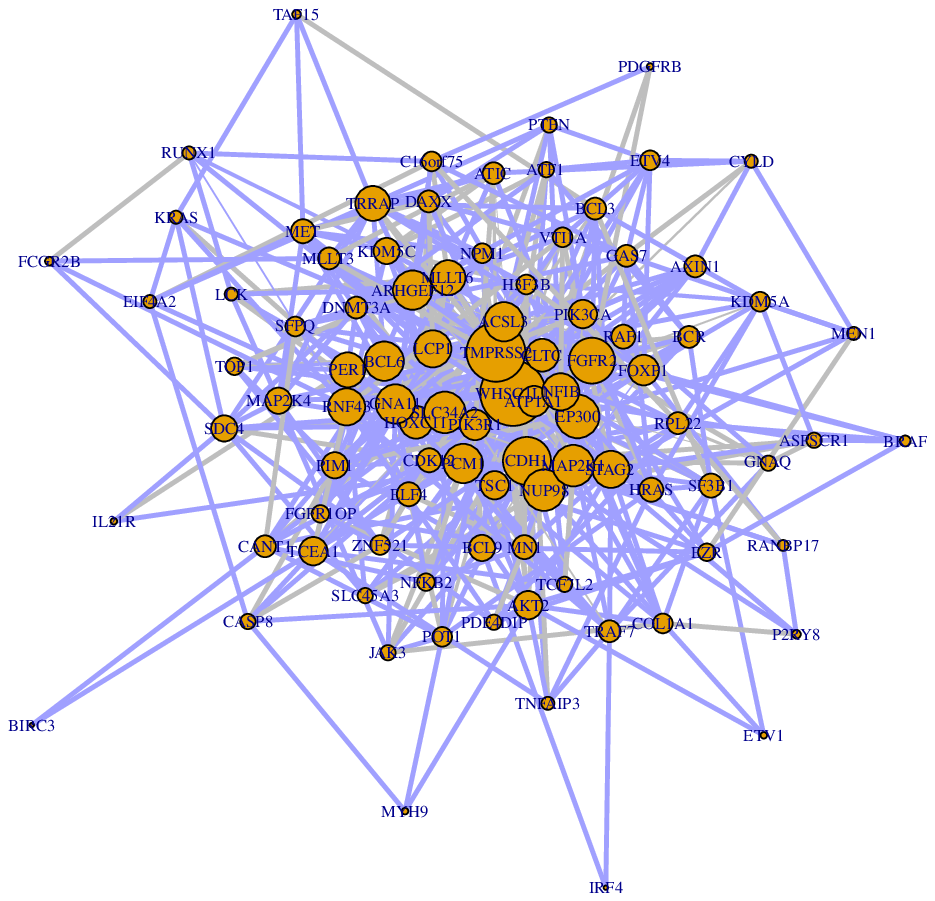}
    \caption{Subgraph of a PGM netwrok for the Breast Cancer miRNA Network obtained using the full-likelihood method. Nodes with degrees of at least $5$ in the full network are plotted here. The edge widths are proportional to the stable selection probability, and the node sizes are proportional to the degree. Blue edges appear in both full-likelihood and pseudo-likelihood networks. }
\label{fig:PGM_breast_cancer_graph_MLE_neg}
\end{figure}

\begin{figure}[h!]
\centering
\includegraphics[width=0.9\textwidth]{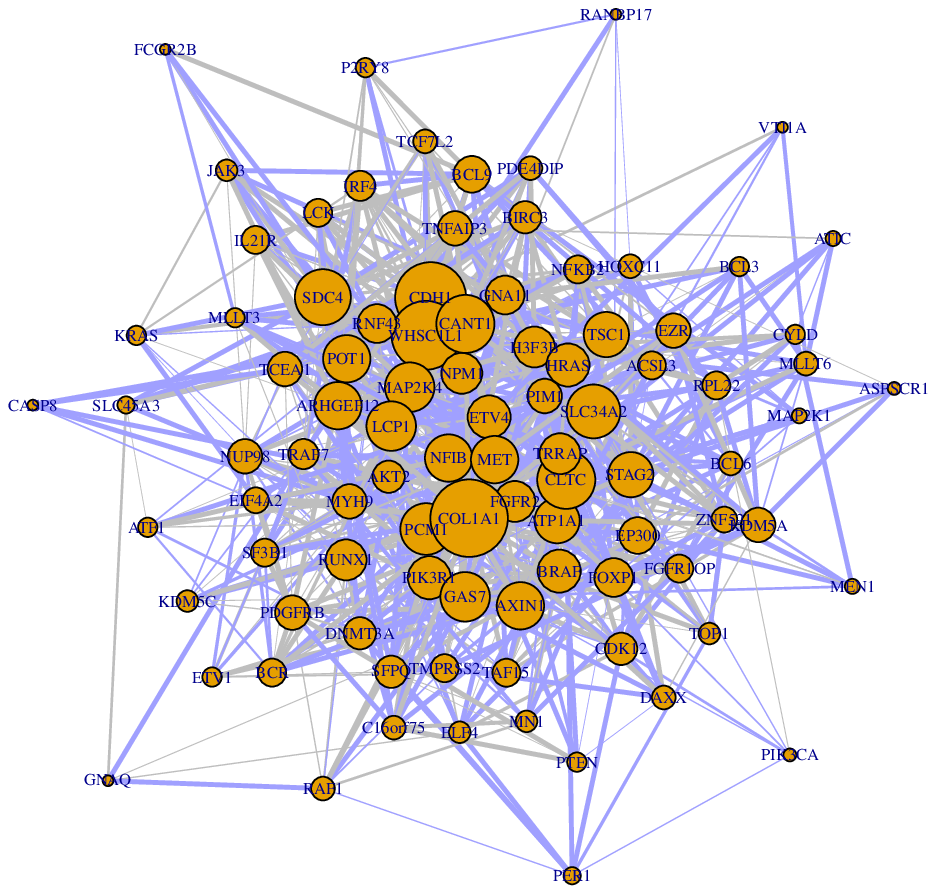}
    \caption{Subgraph of a PGM netwrok for the Breast Cancer miRNA Network obtained using the pseudo-likelihood method. Nodes with degrees of at least $5$ in the full network are plotted here. The edge widths are proportional to the stable selection probability, and the node sizes are proportional to the degree. Blue edges appear in both full-likelihood and pseudo-likelihood networks.}
\label{fig:PGM_breast_cancer_graph_MPLE_neg}
\end{figure}

\end{document}